\documentclass[
  reprint,
  superscriptaddress,
  amsmath,amssymb,
  aps
]{revtex4-2}

\usepackage[T1]{fontenc}
\usepackage[utf8]{inputenc}
\usepackage{lmodern}
\usepackage{microtype}
\usepackage[margin=25mm]{geometry}
\usepackage[normalem]{ulem}
\usepackage{amsmath,amssymb,amsthm}
\usepackage{mathtools}
\usepackage{stmaryrd}
\usepackage{dsfont}
\usepackage{braket}
\usepackage{bm}
\usepackage{graphicx}
\usepackage{xcolor}
\usepackage{here}
\usepackage{array}
\usepackage{booktabs}
\usepackage{makecell}
\usepackage{longtable}
\usepackage{tabularx}
\usepackage{pict2e}
\usepackage{diagbox}
\usepackage{lineno}
\usepackage{tikz}
\usepackage{tikz-cd}
\usepackage{quantikz}
\usepackage{algorithm}
\usepackage{algpseudocode}
\usepackage{ytableau}
\usepackage[
    colorlinks=true,
    linkcolor=blue,
    citecolor=blue,
    urlcolor=blue
]{hyperref}

\newcolumntype{Y}{>{\centering\arraybackslash}X}
\ytableausetup{boxsize=1em}

\makeatletter
\newcommand{\namedlabel}[2]{\begingroup
  #2
  \def\@currentlabel{#2}
  \phantomsection\label{#1}\endgroup}
\makeatother

\newtheorem{Def}{Definition}
\newtheorem{AccessDef}{Definition}
\newtheorem{Thm}[Def]{Theorem}
\newtheorem{Cor}[Def]{Corollary}
\newtheorem{Prop}[Def]{Proposition}
\newtheorem{Lem}[Def]{Lemma}
\newtheorem{Remark}{Remark}

\newtheorem{Task}{Task}

\newlength{\statementaboveskip}
\makeatletter
\g@addto@macro\th@plain{\thm@preskip=\statementaboveskip\relax
  \thm@postskip=4pt plus1pt minus1pt\relax
}
\makeatother

\newcommand{\compactdisplayspacing}
  {\abovedisplayskip=2pt plus0.5pt minus0.5pt\relax
   \belowdisplayskip=2pt plus0.5pt minus0.5pt\relax
   \abovedisplayshortskip=1pt plus0.5pt minus0.5pt\relax
   \belowdisplayshortskip=1.5pt plus0.5pt minus0.5pt\relax
  }
\AddToHook{cmd/normalsize/after}{\compactdisplayspacing}
\AtBeginDocument{\compactdisplayspacing}

\makeatletter
\AddToHook{env/enumerate/begin}
  {\apptocmd{\@listi}
     {\topsep=4pt plus1pt minus1pt\relax
      \partopsep=0pt\relax
      \itemsep=2pt plus0.5pt minus0.5pt\relax
      \parsep=0pt\relax
     }{}{}}
\makeatother

\makeatletter
\newcommand{\paragraphwithoutdash}
  {\@startsection{paragraph}{4}{\parindent}{\z@}{\z@}{\normalfont\normalsize\itshape}}
\makeatother
\RenewDocumentCommand{\paragraph}{s o m}
  {\IfBooleanTF{#1}
     {\paragraphwithoutdash*{#3---}}
     {\IfNoValueTF{#2}
        {\paragraphwithoutdash{#3---}}
        {\paragraphwithoutdash[#2]{#3---}}}}

\newcommand{\bigO}[1]{\mathcal{O}\!\left(#1\right)}

\newcommand{\bigOmega}[1]{\Omega\!\left(#1\right)}

\newcommand{\rank}{\operatorname{rank}}
\newcommand{\poly}{\operatorname{poly}}

\newcommand{\norm}[1]{\left\lVert#1\right\rVert}
\newcommand{\abs}[1]{\left|#1\right|}

\DeclareMathOperator{\Tr}{Tr}
\DeclareMathOperator*{\argmax}{arg\,max}

\begin{document}

\title{An exponential query advantage in learning a channel property from access to a Stinespring dilation unitary and its inverse}
\author{Manaki Arihara}
\affiliation{Department of Physics, The University of Tokyo, 7-3-1 Hongo, Bunkyo-ku, Tokyo 113-0033, Japan}
\author{Satoshi Yoshida}
\affiliation{Department of Physics, The University of Tokyo, 7-3-1 Hongo, Bunkyo-ku, Tokyo 113-0033, Japan}
\author{Natsuto Isogai}
\affiliation{Department of Physics, The University of Tokyo, 7-3-1 Hongo, Bunkyo-ku, Tokyo 113-0033, Japan}
\author{Philip Taranto}
\affiliation{Department of Physics \& Astronomy, University of Manchester, Manchester M13 9PL, United Kingdom}
\author{Mio Murao}
\affiliation{Department of Physics, The University of Tokyo, 7-3-1 Hongo, Bunkyo-ku, Tokyo 113-0033, Japan}
\begin{abstract}
    We show that access to an unknown Stinespring dilation unitary and its inverse can provide an exponential query advantage over black-box channel access and even over forward-only access to the dilation unitary in estimating an intrinsic property of the channel.
    The task is to estimate how close an unknown channel is to the set of unitary channels, as quantified by the \textit{optimal average gate fidelity (OAGF)}, defined as the average gate fidelity maximized over all unitary channels.
    For sufficiently small constant additive error, we estimate OAGF using $O(1)$ queries to a Stinespring dilation unitary and its inverse in total, whereas both forward Stinespring dilation unitary access and black-box channel access require $\Omega(\sqrt{d})$ queries in the worst case, where $d$ is the system dimension.
    Moreover, under standard computational and cryptographic assumptions, we establish a computational advantage of forward and inverse access to the dilation unitary over classical algorithms and quantum algorithms with forward-only dilation unitary access.
    Our results identify the ability to reverse joint system--environment evolution as a resource for efficiently learning properties of the reduced quantum channel.
\par
\end{abstract}

\maketitle
\paragraph*{Introduction.}
Efficient quantum learning of quantum objects, which extracts classical information about unknown quantum states and processes, is a promising application of quantum computers \cite{montanaro2018}. Such learning enables us to understand and predict the behavior of complex physical systems both in nature and in the laboratory, allowing us to assess and subsequently improve the performance of quantum devices.
In any realistic scenario, we are limited by both the resources we can exploit and the way in which the unknown quantum object can be accessed, referred to as the \textit{access model}.
Hence, the goal of efficient quantum learning is to minimize the resources required to accomplish a given learning task, under a specified access model and a set of allowed operations.

For instance, to learn properties of an unknown quantum state $\rho$, a well-studied access model comparison is between \textit{sample access} to $\rho$ and \textit{query access} to a state-preparation unitary, which prepares a purification of $\rho$, and its inverse, both provided by black boxes.
Query access to a state-preparation unitary is at least as powerful as sample access, as one can generate copies of the state by applying the state-preparation unitary and tracing out the environment.
For instance, in quantum state tomography, such state-preparation unitary access provides a quadratic improvement in query complexity as a function of the estimation error compared with access to copies of the state alone \cite{vanapeldoorn2022}; similar advantages have been established for learning other properties~\cite{Rall2019,liu2026poweroracleaccessoptimal,gilyén2019,Zhang2026,Wang2024,Wang2024Renyi,Wang2024Tracedistance,utsumi2025}.
If a fixed-error decision algorithm succeeds for every compatible state-preparation unitary using $q$ queries in total to the unitary and its inverse, the same task can be solved with $O(q^2)$ copies of the state, provided the purification register is sufficiently large~\cite{TangWrightZhandryConjugate}.

Motivated by such access model comparisons for quantum \textit{state} learning, here we investigate the power of access to a Stinespring dilation unitary and its inverse in learning quantum operations, described by \textit{quantum channels}.
Every quantum channel $\mathcal{E}$ can be realized by a so-called Stinespring dilation unitary which acts on the system and an environment prepared in a fiducial state, followed by tracing out the environment~\cite{Stinespring1955}.
Formally, we define the dilation unitary and its access models as follows.
{
\setlength{\statementaboveskip}{2pt}
\begin{AccessDef}[Stinespring dilation unitary and access models]
\label{def:stinespring-access}
\setlength{\abovedisplayskip}{2pt}
\setlength{\belowdisplayskip}{2pt}
\setlength{\abovedisplayshortskip}{1pt}
\setlength{\belowdisplayshortskip}{1.5pt}
Let $\mathcal{E}$ be an unknown channel on a $d$-dimensional system $\mathcal{H}$, and let $\mathcal{H}_E$ be a finite-dimensional environment with a fixed fiducial state $\ket{0}_E$.
A \textit{Stinespring dilation unitary} of $\mathcal{E}$ is a unitary $U$ on $\mathcal{H}\otimes\mathcal{H}_E$ satisfying
\begin{equation}
\label{eq:stinespring-dilation}
    \mathcal{E}(\rho)=\operatorname{Tr}_E[U(\rho\otimes\ket{0}\!\bra{0}_E)U^\dagger]
\end{equation}
for every state $\rho$ on $\mathcal{H}$.

We distinguish the following access models for a fixed unknown $U$:
\begin{enumerate}
    \interlinepenalty=10000\relax
    \item[(i)] Query access to both $U$ and $U^\dagger$.
    \item[(ii)] Forward-only query access to $U$.
\end{enumerate}
\end{AccessDef}
}
Neither the dilation unitary $U$ nor its environment $\mathcal{H}_E$ is uniquely determined by the channel $\mathcal{E}$.
Accordingly, we focus on \textit{intrinsic} properties of the channel, whose values depend only on $\mathcal{E}$ and not on the choice of its dilation unitary.
The minimum environment dimension among such realizations is the \textit{Kraus rank} $r$ of $\mathcal{E}$.
A state-preparation unitary prepares a purification of a quantum state, while a Stinespring dilation unitary implements a quantum channel by acting on the system together with an initialized environment; in both cases, finally tracing out the environment recovers the original quantum object.
For comparison, we define black-box channel access as follows.
\begin{AccessDef}[Black-box channel access]
\label{def:black-box-access}
Let $\mathcal{E}$ be an unknown channel on a $d$-dimensional system $\mathcal{H}$.
\textit{Black-box channel access} to $\mathcal{E}$ is oracle access that, for any auxiliary system with Hilbert space $\mathcal{H}_R$, applies
\begin{equation}
    \rho\longmapsto
    (\mathcal{E}\otimes\mathcal{I}_R)(\rho),
\end{equation}
where $\rho$ is an arbitrary density matrix on $\mathcal{H}\otimes\mathcal{H}_R$ and $\mathcal{I}_R$ denotes the identity channel on $\mathcal{H}_R$.
\end{AccessDef}

In general, there is no universal algorithm that exactly implements a Stinespring dilation unitary, or even a Stinespring dilation isometry, using finitely many black-box queries to the corresponding quantum channel~\cite{liu2025,deltoro2026}.
However, this impossibility does not by itself imply a query advantage for learning intrinsic properties of the channel.
In fact, previous works \cite{GirardiEtAlRandomStinespring,YoshidaEtAlRandomDilation,chen2025} showed that the averaged action of $q$ parallel queries to the same randomly chosen Stinespring dilation isometry can be simulated using $q$ parallel queries to the underlying channel.
Thus, such parallel access to a dilation isometry offers no query advantage for learning intrinsic properties of the channel.
By contrast, access to a Stinespring dilation unitary, its inverse, and their controlled versions reduces the query complexity of unitary certification at fixed accuracy from $\Theta(d)$ under coherent black-box channel access to $\Theta(\sqrt{d})$~\cite{chen2026}.
The two query complexities are therefore quadratically related.
As noted above, for quantum states the corresponding advantage is at most quadratic for every fixed-error decision task~\cite{TangWrightZhandryConjugate}.

These results motivate the following question:
\begin{center}
    \itshape
    Can access to a Stinespring dilation unitary and its inverse provide a superpolynomial query advantage over black-box channel access for learning an intrinsic property of the channel?
\end{center}

\paragraph*{Main Results.}
We answer this question affirmatively by establishing an \textit{exponential} query separation for estimating the \textit{optimal average gate fidelity (OAGF)}---an intrinsic property that quantifies how close a channel is to its nearest unitary channel (Fig.~\ref{fig:oagf}). Moreover, we identify inverse access as the essential ingredient: forward access to the Stinespring dilation unitary alone still requires exponentially many queries, even when the environment is retained and manipulated.

\begin{figure}[H]
    \centering
    \resizebox{0.85\columnwidth}{!}{
\begin{tikzpicture}[x=1bp,y=1bp,text=black]
    \path[use as bounding box] (-1,-16) rectangle (369,211);
%
    \draw[black!65,line width=1.2bp,rounded corners=90bp]
        (3,-12) rectangle (365,207);
    \node[inner sep=0pt,
          font=\sffamily\bfseries\fontsize{16}{19}\selectfont]
        at (184,190) {Set of quantum channels};
%
    \path[fill=black!17,draw=none]
        (116,64) ellipse[x radius=96,y radius=63];
    \node[align=center,inner sep=0pt,
          font=\sffamily\bfseries\fontsize{16}{19}\selectfont]
        at (110,61) {Set of unitary\\channels};
    \coordinate (oagf-optimum) at (197,97);
    \coordinate (oagf-channel) at (282,132);
%
    \draw[black,line width=2.5bp,line cap=round]
        (oagf-optimum) -- (oagf-channel);
    \draw[red,line width=2.5bp,line cap=butt]
        (187,87) -- (207,107)
        (187,107) -- (207,87);
    \filldraw[fill=black,draw=cyan!25!black,line width=1.8bp]
        (oagf-channel) circle[radius=4.1bp];
    \node[inner sep=0pt,font=\fontsize{10}{12}\selectfont,scale=2.4]
        at (191,126) {$\mathcal{U}_{\mathrm{opt}}$};
    \node[inner sep=0pt,font=\fontsize{10}{12}\selectfont,scale=3]
        at (282,160) {$\mathcal{E}$};
    \node[inner sep=0pt,font=\fontsize{10}{12}\selectfont,scale=2.2]
        at (282,82) {$1-F_{\mathrm{OAGF}}(\mathcal{E})$};
\end{tikzpicture}}
    \caption{The optimal average gate fidelity (OAGF), $F_{\mathrm{OAGF}}$, quantifies how close a given quantum channel
    $\mathcal{E}$ is to its closest unitary channel
    $\mathcal{U}_{\mathrm{opt}}$.
    The line represents $1-F_{\mathrm{OAGF}}(\mathcal{E})$, the minimum average gate infidelity over all unitary channels.
    Its estimation is believed to be classically hard: for suitable parameter choices, we establish hardness for DQC1~\cite{PhysRevLett.81.5672}, even given a classical description of the dilation circuit (Theorem~\ref{Thm:DQC1-hard}).}
    \label{fig:oagf}
\end{figure}

\begin{Task}[Estimating the optimal average gate fidelity (OAGF)]
    \label{task:oagf-estimation}
    Let $\mathcal{E}$ be an unknown channel on a $d$-dimensional system $\mathcal{H}$, with $d=2^n$ for an $n$-qubit system.
     We define the optimal average gate fidelity of $\mathcal{E}$ as
    \begin{equation}
        \label{eq:oagf-definition}
        F_{\mathrm{OAGF}}(\mathcal{E}):=\max_{W\in\mathbb{U}(d)}\int d\psi \bra{\psi}W^\dagger\mathcal{E}(\ket{\psi}\!\bra{\psi})W\ket{\psi},
    \end{equation}
    where $d\psi$ is the normalized Haar measure on pure states in $\mathcal{H}$ and $\mathbb{U}(d)$ represents the set of $d$-dimensional unitaries.
    Fix an additive accuracy $\epsilon\in(0,1)$. Unless otherwise stated, $\epsilon$ is independent of the system and environment dimensions. The task is to output a classical estimate $F$ satisfying
    \begin{equation}
    \label{eq:task-oagf-success}
        \Pr\!\left[
            \left|F-F_{\mathrm{OAGF}}(\mathcal{E})\right|\le\epsilon
        \right]\ge\frac{2}{3}.
    \end{equation}
\end{Task}

The average gate fidelity is widely used for benchmarking and certification of quantum devices~\cite{HorodeckiEtAl1999,Nielsen2002,Magesan2011,Roth2018}.
The optimized quantity in Eq.~\eqref{eq:oagf-definition} has also been studied in equivalent forms.
It is an affine rescaling of the \emph{degree of unitarity}~\cite{CuiWang2014}, and its complement is the \emph{incoherent error per gate}~\cite{Feng2016}.
This complement is the average gate infidelity that remains after optimal unitary corrections before and after the channel~\cite{Wallman_2015}.
For example, purity benchmarking has been used to estimate this residual error for Clifford gates on a solid-state qubit~\cite{Feng2016}.
Crucially, the OAGF does \textit{not} depend on which dilation unitary is supplied and is hence an intrinsic property of the channel.
OAGF estimation is also hard for the one-clean-qubit model (DQC1)~\cite{PhysRevLett.81.5672}.
As we show in Theorem~\ref{Thm:DQC1-hard}, no classical polynomial-time algorithm can estimate $F_{\mathrm{OAGF}}(\mathcal{E})$ to additive error $\frac{1}{\poly(n)}$ for all such channels under the assumption $\mathrm{DQC1}\nsubseteq\mathrm{BPP}$, even when a classical description of a $\poly(n)$-size quantum circuit implementing a Stinespring dilation unitary is given.

We compare the query and time complexity of this task under three access models: access to both $U$ and $U^\dagger$ (Definition~\ref{def:stinespring-access}(i)), forward-only access to $U$ (Definition~\ref{def:stinespring-access}(ii)), and black-box channel access (Definition~\ref{def:black-box-access}).
In all three models, arbitrary auxiliary systems, entangled inputs, quantum memory, and intermediate operations, including measurements and classical feedforward, are allowed. These operations may depend on known dimensions, accuracy parameters, stated promises, and measurement outcomes, but are otherwise independent of the unknown channel and dilation. Query complexity counts the total number of oracle uses, maximized over all branches; other operations are not counted.

\begin{Thm}[Upper bound with access to a Stinespring dilation unitary and its inverse]
    \label{thm:upper bound}
    There exists an algorithm that satisfies the following:
    \begin{enumerate}
        \item[(i)] For any additive accuracy $\epsilon\in(0,1)$ fixed independently of $n$, it solves Task~\ref{task:oagf-estimation} using $O(1)$ queries to $U$ and $U^\dagger$ in total.
        \item[(ii)] For $n$-qubit channels of Kraus rank $r=O(1)$, it solves Task~\ref{task:oagf-estimation} with $\epsilon=\frac{1}{\poly(n)}$ in $\poly(n)$ time, provided the supplied dilation acts on $\poly(n)$ qubits and is implementable by a $\poly(n)$-size quantum circuit.
    \end{enumerate}
\end{Thm}

\noindent We emphasize that the algorithm works with any supplied Stinespring dilation unitary $U$ of $\mathcal{E}$, regardless of the environment dimension.

On the other hand, we derive the following lower bound for forward Stinespring dilation unitary access.
\begin{Thm}[Lower bound with forward Stinespring dilation unitary access]
    \label{thm:lower bound}
    Let $n\ge 2$ and $\epsilon<0.04$.
    Any algorithm that solves Task~\ref{task:oagf-estimation} with forward Stinespring dilation unitary access for every $n$-qubit channel and every compatible supplied dilation $U$ requires $\Omega(2^{\frac{n}{2}})$ forward queries to $U$.
    Moreover, the same lower bound holds with access to $U^\dagger$ alone. Both bounds hold even for rank-$2$ channels.
\end{Thm}
The following corollary follows immediately from this theorem, since each black-box channel query can be implemented with one forward query to a supplied dilation unitary by initializing the environment in $\ket{0}_E$ and tracing it out after the query.
\begin{Cor}[Lower bound with black-box channel access]
    \label{cor:black-box-lower-bound}
    Let $n\ge 2$ and $\epsilon<0.04$.
    Any algorithm that solves Task~\ref{task:oagf-estimation} with black-box channel access for every $n$-qubit channel must make $\Omega(2^{\frac{n}{2}})$ channel queries in the worst case. This holds true even for rank-$2$ channels.
\end{Cor}

At fixed additive accuracy $0<\epsilon<0.04$, Theorems~\ref{thm:upper bound} and~\ref{thm:lower bound}, together with Corollary~\ref{cor:black-box-lower-bound}, establish an exponential query advantage for access to $U$ and $U^\dagger$ over forward-only access to $U$ and black-box channel access.
\begin{proof}[Sketch of the algorithm for Theorem~\ref{thm:upper bound}]
A variational identity gives
\begin{equation}
\label{eq:main-environment-optimization}
    F_{\mathrm{OAGF}}(\mathcal{E})
    =\max_{\|\psi\|_2=1}
    \frac{1+\|A_\psi\|_1^2/d}{d+1},
\end{equation}
where
$A_\psi=(I_{\mathcal{H}}\otimes\bra{\psi}_E)
U(I_{\mathcal{H}}\otimes\ket{0}_E)$
and $\|\cdot\|_1$ denotes the trace norm.
We introduce two approximation methods, \emph{spectral truncation} and \emph{spectral filtering}, described as follows.
\emph{Spectral truncation} restricts this optimization,
up to $O(\epsilon)$ additive error, to an environment subspace of dimension at most $R:=\min\{r,\left\lceil\frac{8}{\epsilon}\right\rceil\}$.
\emph{Spectral filtering} prepares the normalized projector onto this subspace.
We then generate Haar-random candidate states using random purification~\cite{girardi2026randompurificationchannelsimple}, followed by measurements of the reference registers.

For each candidate, sample-to-query lifting ~\cite{Wang_2025,TangWrightZhandryConjugate} simulates the preparation-unitary access needed to estimate the objective in Eq.~\eqref{eq:main-environment-optimization} using quantum singular value transformation ~\cite{Gilyen2019SVD} and amplitude estimation ~\cite{Brassard2000} with $U$ and $U^\dagger$.
Selecting the largest estimated objective value yields the OAGF estimate.
At fixed $\epsilon$, the query cost is $O(1)$, independently of the system and environment dimensions.
For $r=O(1)$, the query cost is polynomial in $\frac{1}{\epsilon}$, yielding (ii) under the stated implementation assumptions.
The algorithm and its analysis are given in Theorem~\ref{main-upper-bound} of the Supplemental Material (SM)~\cite{SupplementalMaterial}.
\end{proof}

\begin{proof}[Proof sketch of Theorem~\ref{thm:lower bound}]
We construct two rank-$2$ channels, each an equal mixture of two unitary channels, with a constant OAGF gap.
Independently Haar-random input and output unitaries hide their distinguishing structure while preserving OAGF; each unitary is fixed across queries.
We analyze the two Haar-averaged output states by replacing the Haar-random unitaries with path-recording oracles~\cite{Ma2024}.
The corresponding output states coincide after this replacement, even when the algorithm retains and manipulates the environment.
Bounding the replacement errors therefore gives a trace-norm distance of $O\!\left(\frac{t^2}{d}\right)$ between the original Haar-averaged output states for any adaptive $t$-query algorithm.
Estimating OAGF to within half the gap would distinguish these ensembles with constant advantage, requiring $t=\Omega(\sqrt{d})=\Omega(2^{\frac{n}{2}})$.
The full proof is given in the SM~\cite{SupplementalMaterial} (Theorem~\ref{Thm:lower bound}).
\end{proof}

The separation also rules out universal simulation of algorithms with access to a randomly sampled Stinespring dilation unitary and its inverse using $\poly(\log d,q)$ channel queries, even to sufficiently small constant trace-distance error.
Indeed, OAGF is independent of the dilation, so our constant-query estimator remains valid after averaging over compatible dilations.
Such a simulator would therefore violate Corollary~\ref{cor:black-box-lower-bound}.
Further details are given in Sec.~\ref{sec:separation} of the SM~\cite{SupplementalMaterial}.

\paragraph*{Computational hardness.}
Beyond query complexity, we also investigate the computational complexity of OAGF estimation, even for channels whose Stinespring dilation unitaries admit polynomial-size circuits.
We show that this task is hard for DQC1~\cite{PhysRevLett.81.5672} at inverse-polynomial additive accuracy, even when a classical description of a polynomial-size circuit implementing a Stinespring dilation unitary is given.
\begin{Thm}[OAGF estimation is DQC1-hard]
    The input is a classical description of a $\poly(n)$-size quantum circuit implementing a Stinespring dilation unitary of an $n$-qubit quantum channel $\mathcal{E}_{\poly}$.
    Solving Task~\ref{task:oagf-estimation} to additive accuracy $\epsilon=\frac{1}{\poly(n)}$ is DQC1-hard.
    \label{Thm:DQC1-hard}
\end{Thm}
We highlight that, under the widely believed assumption $\mathrm{DQC1}\nsubseteq\mathrm{BPP}$, Theorem~\ref{Thm:DQC1-hard} rules out estimating OAGF to inverse-polynomial additive accuracy in polynomial time using a classical computer, even when a classical description of the dilation circuit is known.

Theorem~\ref{Thm:DQC1-hard} follows by considering the channel
\begin{equation}
    \mathcal{E}(\rho)=(1-p)\rho+p V\rho V^\dagger,
\end{equation}
where $V$ is an $n$-qubit unitary specified by a polynomial-size circuit.
For a suitable inverse-polynomial choice of $p$, estimating its OAGF to sufficiently small inverse-polynomial additive error yields an estimate of $\left|\frac{\Tr V}{d}\right|^2$ to inverse-polynomial accuracy, which is DQC1-hard~\cite{Khatri_2019}.
The proof of Theorem~\ref{Thm:DQC1-hard} is given in the SM~\cite{SupplementalMaterial}.
\par

We next consider the time complexity of OAGF estimation with forward Stinespring dilation unitary access.
As in Theorem~\ref{Thm:DQC1-hard}, we restrict to channels whose supplied dilation unitaries admit $\poly(n)$-size circuits.
Here, the classical description of the dilation circuit is not given to the quantum algorithm.
Under the PRU assumption below, no universal polynomial-time quantum algorithm solves Task~\ref{task:oagf-estimation} at inverse-polynomial additive accuracy.
\begin{Thm}
    Assume that pseudorandom unitary (PRU) families exist~\cite{Ma2024}.
    Consider a quantum channel $\mathcal{E}_{\poly}$ whose Stinespring dilation unitary has a $\poly(n)$-size quantum circuit implementation.
    For $\epsilon=\frac{1}{\poly(n)}$, there is no universal polynomial-time quantum algorithm with forward Stinespring dilation unitary access that solves Task~\ref{task:oagf-estimation}.
    \label{Thm:forward-hardness}
\end{Thm}
\begin{proof}[Proof sketch of Theorem~\ref{Thm:forward-hardness}]
    Suppose that a polynomial-time quantum algorithm $\mathcal{A}$ estimates OAGF to inverse-polynomial additive accuracy with forward Stinespring dilation unitary access whenever the supplied dilation unitary admits a polynomial-size circuit.
    For each $b\in\{Y,N\}$, consider
    \begin{equation}
        U_b'=(V_1\otimes I)U_b(V_2\otimes I),
    \end{equation}
    where $U_b$ is the dilation unitary used in the proof of Theorem~\ref{thm:lower bound} and admits a polynomial-size circuit.
    The unitaries $V_1,V_2$ are sampled independently from either (i) a PRU family or (ii) the Haar measure.
    In both cases, $V_1,V_2$ are fixed throughout all queries.
    We choose $b\in\{Y,N\}$ uniformly at random, keep it hidden from $\mathcal{A}$, and run $\mathcal{A}$ with forward oracle access to $U_b'$.

    Let $\tilde b$ be the label inferred from its OAGF estimate.
    In the PRU case (i), $U_b'$ admits a polynomial-size circuit, so $\mathcal{A}$ estimates the channel's OAGF to inverse-polynomial additive accuracy.
    Unitary invariance preserves the constant OAGF gap between the two labels, and hence $\Pr[\tilde b=b]\ge \frac{2}{3}$.
    In the Haar case (ii), the quantitative bound in the proof of Theorem~\ref{thm:lower bound} gives $\Pr[\tilde b=b]\le \frac{1}{2}+\operatorname{negl}(n)$ for polynomially many queries.
    Since we know the chosen label $b$, checking whether $\tilde b=b$ distinguishes cases (i) and (ii) in polynomial time with non-negligible advantage, contradicting the PRU assumption.
    Details are given in the proof of Theorem~\ref{Thm:forward-hardness} in the SM~\cite{SupplementalMaterial}.
\end{proof}
\paragraph*{Discussion.}
We showed that, for estimating the OAGF, which quantifies how close a given quantum channel is to a unitary channel, access to a Stinespring dilation unitary and its inverse yields an exponential query advantage over both forward-only access to the dilation unitary and black-box channel access.
An exponential advantage of access to the environment over black-box channel access was demonstrated by Goldin~\cite{goldin2026randompurificationsisometrieschannels}: there, a certain discrimination problem is solved with two forward queries to a dilation isometry, whereas black-box channel access requires exponentially many queries.
In contrast, for OAGF estimation, the exponential cost persists even under forward access to the dilation unitary, which allows the environment input to be manipulated, and disappears only once the inverse is added.

Implementing the inverse of an unknown $d$-dimensional unitary from forward queries is generally costly: it requires $\Omega(d^2)$ queries both for deterministic exact implementation~\cite{Odake_2025} and for approximation to sufficiently small constant diamond-norm error~\cite{Chen_2026}.
Thus, access to both the unitary and its inverse can provide an advantage over forward access alone.
Such advantages have been demonstrated in amplitude amplification and amplitude estimation~\cite{TangWrightInverse}, quantum PAC learning~\cite{isogai2026advantagesamplecomplexityquantum}, and state distinguishing~\cite{TangWrightInverse,wang2026statedistinguishinginversequeries}, with related results for unitary certification using controlled queries to a dilation unitary and its inverse~\cite{chen2026}.
All of these are improvements in the accuracy dependence or quadratic improvements in the dimension.
By contrast, our separation is \textit{exponential} in the number of qubits at fixed accuracy.
This exponential separation does not rely on the quadratic speedups provided by amplitude amplification and amplitude estimation.
Indeed, it persists when both procedures are replaced by direct repetition and measurement.
An exponential advantage of inverse over forward queries is known for permutation inversion with an in-place oracle~\cite{fefferman2018quantumvsclassicalproofs,belovs2023onewayticketlasvegas}.
However, the task itself is to invert the oracle, and a single inverse query suffices.
For OAGF estimation, even inverse-only access requires exponentially many queries in the worst case~\cite{SupplementalMaterial}, so the combination of forward and inverse access is essential.

Theorems~\ref{thm:upper bound}(ii), \ref{Thm:DQC1-hard}, and~\ref{Thm:forward-hardness} give a conditional separation in the computational complexity of OAGF estimation at inverse-polynomial accuracy, even for rank-two channels with polynomial-size dilation circuits.
The trace-estimation reduction underlying Theorem~\ref{Thm:DQC1-hard} already implies DQC1-hardness at this accuracy for estimating the average gate fidelity with the identity channel from a classical circuit description~\cite{Khatri_2019}.
Given a circuit for a reference unitary, a quantum learner can estimate its average gate fidelity with the channel to additive error $\epsilon$ using $\bigO{\frac{1}{\epsilon^2}}$ forward channel queries via SWAP tests on the normalized Choi states~\cite{Harry2001}.
A supplied polynomial-size circuit for an OAGF-maximizing unitary would therefore suffice for efficient estimation with forward access.
Theorem~\ref{thm:lower bound} shows that estimating the optimum from forward access alone can require exponentially many queries, even at constant accuracy.
\par

Finally, our results highlight a difference between states and channels.
For a circuit making $q$ queries in total to a state-preparation unitary and its inverse, $O\!\left(\frac{q^2}{\epsilon}\right)$ copies of the state suffice to approximate, within trace distance $\epsilon$, the circuit output averaged over a specific distribution of compatible preparation unitaries~\cite{TangWrightZhandryConjugate}.
For channels, however, our separation implies that queries to a random dilation unitary and its inverse cannot be simulated with $\mathrm{poly}(\log d,q)$ channel queries, nor reduced to forward queries to the dilation unitary at such a cost (Corollary~\ref{cor:no-go dilation} of the SM~\cite{SupplementalMaterial}).
Thus, whereas the gap between copies of a state and access to its state-preparation unitary and inverse is at most quadratic for decision tasks, the gap between channel access and access to a dilation unitary and its inverse can be exponential.

\paragraph*{Note added.}
After the first arXiv version of this work appeared~\cite{ThisWork_v1}, Goldin~\cite{goldin2026randompurificationsisometrieschannels} showed that, for a certain channel-discrimination problem, forward access to a dilation isometry yields an exponential query advantage over black-box channel access.
In the present version, we have extended our lower bound to forward access to the dilation unitary (Theorem~\ref{thm:lower bound}).
\par

\paragraph*{AI disclosure.}
ChatGPT 5.6 Sol assisted the authors in proving the lower bound, especially in the proof of Lemma~\ref{Lem:Hash equal} in the SM~\cite{SupplementalMaterial}.
The authors verified the argument by hand and derived a simpler proof.
ChatGPT 6 Astra also assisted the authors in revising the manuscript and checking correctness.
The authors take full responsibility for the presentation of the results and their correctness.
\paragraph*{Acknowledgments.}
This work was supported by Japan Society for the
Promotion of Science (JSPS) KAKENHI Grant Numbers 23K21643 and 26K25550, the MEXT Quantum Leap Flagship Program (MEXT QLEAP) JPMXS0118069605, JPMXS0120351339, JST CREST Grant Number JPMJCR25I5, JST ASPIRE Grant Number JPMJAP25A3, JST NEXUS Grant Number JPMJNX26C9, FoPM, WINGS Program, the University of Tokyo, JST BOOST, Japan (JPMJBS2418) and IBM Quantum.

\bibliographystyle{linkedrefs}
\bibliography{references}
\clearpage
\onecolumngrid
\setcounter{page}{1}
\renewcommand{\thepage}{S\arabic{page}}
\setcounter{section}{0}
\renewcommand{\thesection}{S\Roman{section}}
\def\theHsection{supp.\arabic{section}}
\setcounter{equation}{0}
\renewcommand{\theequation}{S\arabic{equation}}
\def\theHequation{supp.\arabic{equation}}
\setcounter{Def}{0}
\renewcommand{\theDef}{S\arabic{Def}}
\def\theHDef{supp.\arabic{Def}}
\setcounter{Remark}{0}
\renewcommand{\theRemark}{S\arabic{Remark}}
\def\theHRemark{supp.\arabic{Remark}}
\setcounter{algorithm}{0}
\renewcommand{\thealgorithm}{S\arabic{algorithm}}
\def\theHalgorithm{supp.\arabic{algorithm}}
\setcounter{figure}{0}
\renewcommand{\thefigure}{S\arabic{figure}}
\def\theHfigure{supp.\arabic{figure}}
\setcounter{table}{0}
\renewcommand{\thetable}{S\arabic{table}}
\def\theHtable{supp.\arabic{table}}
\section*{Supplemental Material.}
\label{sec:supplemental-material}
\section{Notation and definitions.}
\label{sec:supplemental-notation}
For a channel $\mathcal{E}$ on a $d$-dimensional system, define its normalized Choi state~\cite{Choi1975,JAMIOLKOWSKI1972} by
\[
    \mathsf{E}:=(I_R\otimes\mathcal{E})(\ket{\Phi^+}\!\bra{\Phi^+}),
    \qquad
    \ket{\Phi^+}_{RS}:=\frac{1}{\sqrt{d}}\sum_{j=1}^{d}\ket{j}_R\ket{j}_S.
\]
Thus $\Tr\mathsf{E}=1$; for unitary operators $U,V,W$, $\mathsf{U}$, $\mathsf{V}$, and $\mathsf{W}$ denote the normalized Choi states of the corresponding unitary channels $U(\cdot)U^\dagger$, $V(\cdot)V^\dagger$, and $W(\cdot)W^\dagger$, respectively.
For a unitary channel $\mathcal{W}(\rho)=W\rho W^\dagger$, its entanglement fidelity relative to $\mathcal{E}$ is $F_{\mathrm{e}}(\mathcal{E},\mathcal{W}):=\Tr[\mathsf{E}\mathsf{W}]$, and its average gate fidelity is $F_{\mathrm{avg}}(\mathcal{E},\mathcal{W})=\frac{dF_{\mathrm{e}}(\mathcal{E},\mathcal{W})+1}{d+1}$.
Their maxima over unitary channels define the optimal entanglement fidelity (OEF) and OAGF~\cite{Nielsen2002}:
\[
    F_{\mathrm{OEF}}(\mathcal{E}):=\max_{W\in\mathbb{U}(d)}F_{\mathrm{e}}(\mathcal{E},\mathcal{W}),
    \qquad
    F_{\mathrm{OAGF}}(\mathcal{E})=\frac{dF_{\mathrm{OEF}}(\mathcal{E})+1}{d+1}.
\]
The Kraus rank $r$, the minimum number of Kraus operators representing $\mathcal{E}$, equals $\rank\mathsf{E}$.
The environment dimension $m$ of a given dilation satisfies $m\ge r$ and need not be minimal.

\section{Query upper bound with access to a Stinespring dilation unitary.}
\label{sec:stinespring-upper-bound}

We describe a filtering-based algorithm and its error and query-complexity guarantees.
Throughout this section, $d$ is the system dimension and $m$ is the environment dimension of the given dilation. We write $r$ for the Kraus rank of $\mathcal{E}$.

\begin{Thm}[Upper bound for the filtering-based algorithm]
    \label{main-upper-bound}
    Let $\mathcal{E}$ be a quantum channel on a $d$-dimensional Hilbert space $\mathcal{H}$ with known Kraus rank $r$, and let $U\in\mathrm{Dilu}_m(\mathcal{E})$.
    For $\epsilon,\delta\in(0,1)$, set $\tau=\frac{\epsilon}{4}$ as in Algorithm~\ref{alg:max-estimation}.
    Algorithm~\ref{alg:max-estimation} outputs $F_{\mathrm{OAGF}}\in[\frac{1}{d+1},1]$ satisfying
    \begin{equation}
        \Pr\!\left[
            \left|F_{\mathrm{OAGF}}-F_{\mathrm{OAGF}}(\mathcal{E})\right|\le\epsilon
        \right]\ge1-\delta,
    \end{equation}
    using
    \begin{equation}
    \label{eq:main-query-upper-bound}
        \bigO{\frac{R}{\delta^2\epsilon^{\frac{17}{2}}}
        \left(\frac{64}{\epsilon}\right)^{4R-2}\ell^5}
    \end{equation}
    queries to $U$ and $U^\dagger$ in total, where $R:=\min\{r,\left\lceil\frac{8}{\epsilon}\right\rceil\}$ and $\ell:=R\log\!\left(\frac{64}{\epsilon}\right)+\log\!\left(\frac{4}{\delta}\right)$. If the Kraus rank is unknown, the algorithm can instead use the known upper bound $m$.
\end{Thm}

Let $\{K_i\}_{i=1}^m$ be the Kraus operators induced by $U$, defined by
\begin{equation}
    U\bigl(I_\mathcal{H}\otimes\ket{0}_E\bigr)
    =\sum_{i=1}^{m}K_i\otimes\ket{i}_E,
    \qquad
    \sum_{i=1}^{m}K_i^\dagger K_i=I_\mathcal{H}.
\end{equation}
The induced Kraus representation has $m$ operators and need not be minimal.
For any vector $\ket{\psi}_E\in\mathcal{H}_E$, define
\begin{equation}
\label{eq:Apsi}
    A_\psi:=\bigl(I_\mathcal{H}\otimes\bra{\psi}_E\bigr)U\bigl(I_\mathcal{H}\otimes\ket{0}_E\bigr).
\end{equation}
For unit vectors $\ket{\psi}_E\in\mathcal{H}_E$, define the objective
\begin{equation}
\label{eq:environment-objective}
    f(\ket{\psi}_E):=\frac{\|A_\psi\|_1^2}{d^2}.
\end{equation}
By Lemma~\ref{lem:oef-variational}, the optimal entanglement fidelity admits the variational representation
\begin{equation}
\label{eq:algorithm-variational}
    F_{\mathrm{OEF}}(\mathcal{E})
    =\max_{\substack{\ket{\psi}_E\in\mathcal{H}_E\\\|\psi\|_2=1}}f(\ket{\psi}_E).
\end{equation}
Define the environment state
\begin{equation}
\label{eq:quantum state G}
    G:=\Tr_{\mathcal{H}}\left[U\left(\frac{I_{\mathcal{H}}}{d}\otimes\proj{0}_E\right)U^\dagger\right].
\end{equation}
The normalized Choi state $\mathsf{E}$ and $G$ are the two reduced states of the pure state $(I_R\otimes U)(\ket{\Phi^+}_{R\mathcal{H}}\otimes\ket{0}_E)$.
Since their nonzero spectra coincide, $\rank G=\rank\mathsf{E}=r$ for every supplied dilation unitary, including nonminimal ones.
For a spectral threshold $q>0$, write
\begin{equation}
\label{eq:filtered-search-space}
    \Pi_q:=\mathbf{1}_{[q,1]}(G),\qquad
    \mathcal{H}_q:=\operatorname{img}(\Pi_q),\qquad
    s_q:=\rank\Pi_q,\qquad
    \rho_q:=\frac{\Pi_q}{s_q}\quad(s_q>0).
\end{equation}
Lemma~\ref{lem:spectral-truncation}, applied with threshold $q$, bounds the dimension of $\mathcal{H}_q$.

In Algorithm~\ref{alg:max-estimation}, we use the following parameters:
\begin{enumerate}
    \item[(i)] $R$: an integer upper bound on the dimension
    of the search subspace.
    \item[(ii)] $N$: the number of candidate states to be evaluated.
    \item[(iii)] $K$: the number of copies of each candidate
    used for sample-to-query lifting.
    \item[(iv)] $b=R(K-1)+1$: the number of filtered-state copies
    used to generate $K$ copies of one candidate.
    \item[(v)] $L=Nb$: the total number of filtered-state copies
    used for all $N$ candidates.
\end{enumerate}
Let $\{\ket{i}_D\}_{i=1}^R$ be a fixed orthonormal basis of each reference register $D\simeq\mathbb{C}^R$.
The subroutine $\texttt{\upshape FilterG}$ is given in Algorithm~\ref{alg:filter-g}.
The operators $\widetilde{U}$ and $\widetilde{W}_j$ used in candidate evaluation are defined in Step~3 below.
The algorithm below first tests whether $s_q=0$ before filtering.
\par

For $\tau>0$ and a positive integer $M$, define the grid of spectral thresholds by
\begin{equation}
    \label{eq:spectral-threshold-grid}
    T(\tau,M):=\left\{\frac{\tau}{2}+j\frac{\tau}{2M}\;\middle|\;j=0,\ldots,M-1\right\}.
\end{equation}
The spacing of $T(\tau,M)$ is $h=\frac{\tau}{2M}$.

\algrenewcommand{\algorithmicrequire}{\textbf{Input:}}
\algrenewcommand{\algorithmicensure}{\textbf{Output:}}
\begin{algorithm}[H]
    \caption{Estimation of $F_{\mathrm{OAGF}}(\mathcal{E})$ via QSVT filtering and sample-to-query lifting}
    \label{alg:max-estimation}
    \small
    \begin{algorithmic}[1]
        \Require Query access to $U,U^\dagger$; $d,m$; known Kraus rank $r$ (or an upper bound); $\epsilon,\delta\in(0,1)$
        \Ensure A classical estimate $F_{\mathrm{OAGF}}\in[\frac{1}{d+1},1]$
        \State $\tau\gets\frac{\epsilon}{4}$, $\eta\gets\frac{\epsilon}{8}$, $\xi\gets\frac{\epsilon}{2}$, $\gamma\gets\frac{\delta}{8}$
        \label{line:eta}
        \State $R\gets\min\{r,\left\lceil\frac{2}{\tau}\right\rceil\}$
        \State $N\gets\left\lceil\left(\frac{8}{\eta}\right)^{2R-1}\left[2R\log\!\left(\frac{6}{\eta}\right)+\log\!\left(\frac{4}{\delta}\right)\right]\right\rceil$
        \label{line:sample-size}
        \State $\beta\gets\frac{\delta}{4N}$
        \State $K\gets\left\lceil\frac{CN}{\gamma\xi^6}\log^2\!\left(\frac{2}{\beta}\right)\right\rceil$, for a fixed sufficiently large universal constant $C>0$
        \State $b\gets R(K-1)+1$, $L\gets Nb$, $\alpha\gets\frac{\gamma}{L}$
        \State $M\gets 2^{\left\lceil\log_2\frac{64}{\tau\delta}\right\rceil}$, $h\gets\frac{\tau}{2M}$
        \State Sample $q$ uniformly from $T(\tau,M)$
        \State Test whether $s_q=0$ using Lemma~\ref{lem:empty-filtered-subspace}, with failure probability $\frac{\delta}{8}$
        \If{the test reports $s_q=0$}
            \State \Return $F_{\mathrm{OAGF}}\gets\frac{1}{d+1}$
        \EndIf
        \State $\Omega_{q,L}\gets\texttt{\upshape FilterG}(q,L,L\alpha,\frac{\delta}{4})$
        \State $F_{\mathrm{OEF}}\gets0$
        \For{$j=1,2,\ldots,N$}
            \State Apply random purification to $b$ unused registers of $\Omega_{q,L}$ with reference dimension $R$
            \State Measure references in $\{\ket{i}_D\}_{i=1}^R$; keep $K$ environment registers with the most frequent outcome
            \State Simulate $\widehat{f}_j\gets\texttt{\upshape Estimatef}\!\left(\widetilde{U},\widetilde{W}_j;\xi,\beta\right)$ by lifting with trace-norm error at most $\frac{\gamma}{N}$
            \State $\widehat{f}_j\gets\min\{1,\max\{0,\widehat{f}_j\}\}$
            \State $F_{\mathrm{OEF}}\gets\max\{F_{\mathrm{OEF}},\widehat{f}_j\}$
        \EndFor
        \State \Return $F_{\mathrm{OAGF}}\gets\frac{dF_{\mathrm{OEF}}+1}{d+1}$
    \end{algorithmic}
\end{algorithm}

With $\delta_{\mathrm f}=\frac{\delta}{4}$ and $\delta_{\mathrm{test}}=\frac{\delta}{8}$, Lemmas~\ref{lem:filtered-state-preparation} and~\ref{lem:empty-filtered-subspace} use the same grid $T(\tau,M)$. Sampling $q$ takes $\log_2 M$ random bits.
If the algorithm reaches the candidate-evaluation loop, the essential output is
\begin{equation}
\label{eq:algorithm-output}
    F_{\mathrm{OEF}}=\max_{1\le j\le N}\widehat{f}_j.
\end{equation}

\begin{proof}[Proof of Theorem~\ref{main-upper-bound}]
    We first describe the algorithm in three steps.
    We then prove its correctness and derive the total query bound.
    The supporting lemmas are stated and proved in
    Section~\ref{sub:lemma proof}.
    \par
    
    \paragraph*{Step 1. Filtering the effective subspace.}
    We sample $q$ uniformly from $T(\tau,M)$ and use the same threshold throughout the algorithm.
    By Lemmas~\ref{lem:projected F} and~\ref{lem:spectral-truncation},
    restricting the optimization to $\mathcal{H}_q$, where
    $0<s_q\le\min\{r,\left\lfloor\frac{1}{q}\right\rfloor\}\le R=O(\frac{1}{\epsilon})$,
    yields a restricted value that differs from the optimal value
    by at most $q\le\tau=\frac{\epsilon}{4}=O(\epsilon)$.

    We first identify this subspace $\mathcal{H}_q$ by creating the projector $\Pi_q$.
    In this case, the QSVT-based subroutine
    $\texttt{\upshape FilterG}$ prepares a joint $L$-register state
    $\Omega_{q,L}$ approximating $\rho_q^{\otimes L}$ in trace norm,
    where $\rho_q=\frac{\Pi_q}{s_q}$ is the maximally mixed state on $\mathcal{H}_q$.
    Since one query to $U$ implements a state-preparation unitary for $G$, an exact block-encoding of $G$ can be constructed using a constant number of queries to $U$ and $U^\dagger$~\cite[Lemma~45]{Gilyen2019SVD}.
    For the transition width $\Delta>0$ specified in Lemma~\ref{lem:filtered-state-preparation}, the random choice of $q$ ensures with high probability that $[q-\Delta,q+\Delta]$ contains no eigenvalue of $G$.
    We consider a polynomial approximation of the following filter outside this interval:
    \begin{equation}
        g_q(x)=
        \begin{cases}
            \dfrac{\sqrt{q}}{2\sqrt{x}}\quad &(x\ge q),\\
            0\quad &(x<q).
        \end{cases}
    \end{equation}
    We use the inverse-square-root filtering approach of Ref.~\cite[Lemma~19 of the full version]{vanApeldoornGilyen2019SDP}.
    The ideal filter satisfies $g_q(G)Gg_q(G)^\dagger=\frac{q}{4}\Pi_q$, whose normalized state is $\rho_q$.
    QSVT followed by amplitude amplification prepares the joint state $\Omega_{q,L}$ used in the next step; the required preparation guarantee and its implementation outline are stated in Lemma~\ref{lem:filtered-state-preparation}.

    \paragraph*{Step 2. Sampling candidate states.}
    Our search strategy is hypothesis selection: sample candidate environment states, estimate their objective values, and select the largest estimate.
    To create $N$ sample candidates, we divide the $L=Nb$ registers of $\Omega_{q,L}$
    into $N$ groups of $b$ registers.
    For each group, we independently apply random
    purification~\cite{girardi2026randompurificationchannelsimple},
    measure the reference registers, and retain $K$ environment
    registers with the most frequent outcome.
    By Lemma~\ref{lem:haar-lifted-evaluation},
    for the ideal input $\rho_q^{\otimes L}$, this produces
    $N$ independent Haar-random candidates, each with $K$ copies,
    with joint state
    \begin{equation}
    \label{eq:haar-candidate-batches}
        \left(\int_{\ket{\psi}_E\sim\mathrm{Haar}(\mathcal{H}_q)}
            \proj{\psi}_E^{\otimes K}\,d\psi\right)^{\otimes N}.
    \end{equation}
    From $b=R(K-1)+1$ reference outcomes, the most frequent outcome occurs at least $K$ times, yielding $K$ copies of one candidate.

    Let $\mathcal{N}_\eta$ be a finite $\eta$-net of unit vectors in $\mathcal{H}_q$ for $\eta>0$; that is, for every unit vector $\ket{\phi}_E\in\mathcal{H}_q$, there is a vector $\ket{\psi}_E\in\mathcal{N}_\eta$ satisfying $\|\ket{\phi}_E-\ket{\psi}_E\|_2\le\eta$.
    For $s_q\ge2$, the choice of $N$ in Algorithm~\ref{alg:max-estimation} and $s_q\le R$ ensure, by Lemma~\ref{lem:epsilon-net}, that the Haar-random candidates form such a net with probability at least $1-\frac{\delta}{4}$.
    The Lipschitz continuity of $f$ (Lemma~\ref{lem:Lipschitz}) guarantees that this $\eta$-net of Haar-random vectors contains a vector whose objective value is within $2\eta$ of the maximum over unit vectors in $\mathcal{H}_q$.
    The $K$ copies created for each candidate are used for sample-to-query lifting in the next step.
    
    \paragraph*{Step 3. Evaluating the candidates.}
    The details of the subroutine $\texttt{\upshape Estimatef}$ are given in Lemma~\ref{lem:QSVT}.
    For an auxiliary qubit $B$, define the unitary $\widetilde{U}:=X_B\otimes U$.
    For each candidate $\ket{\psi_j}_E$, let $\widetilde{W}_j$ be a preparation unitary satisfying $\widetilde{W}_j\ket{0}_B\ket{0}_E=e^{i\theta_j}\ket{1}_B\ket{\psi_j}_E$.
    We apply sample-to-query lifting~\cite{Wang_2025} to these operators, using the construction of Ref.~\cite[Theorem~1.5 and Section~2.4]{TangWrightZhandryConjugate}.
    The block encoded in Algorithm~\ref{alg:estimate-f} with these operators is
    \begin{equation}
    \label{eq:flag-candidate-block}
        \widetilde{A}_{\psi_j}=e^{-i\theta_j}A_{\psi_j},
        \qquad
        \widetilde{A}_{\psi_j}^\dagger\widetilde{A}_{\psi_j}
        =A_{\psi_j}^\dagger A_{\psi_j}.
    \end{equation}
    Thus the ideal call evaluates the same objective $f(\ket{\psi_j}_E)$.
    For each candidate, lifting uses the retained copies to simulate $\texttt{\upshape Estimatef}(\widetilde{U},\widetilde{W}_j;\xi,\beta)$ and obtain $\widehat{f}_j$.
    In $\texttt{\upshape Estimatef}$, QSVT produces an encoded block $B_{\psi_j}$ approximating $\frac{1}{2}\sqrt{A_{\psi_j}^\dagger A_{\psi_j}}$.
    Applying the corresponding circuit to $\ket{\Phi^+}_{AA'}\ket{0}_{\mathrm{anc}}$ and jointly testing for $\ket{\Phi^+}_{AA'}$ and $\ket{0}_{\mathrm{anc}}$ gives success probability
    \begin{equation}
    \label{eq:candidate-success-probability}
        p_{\psi_j}=\left|\frac{\Tr B_{\psi_j}}{d}\right|^2,
    \end{equation}
    which approximates $\frac{f(\ket{\psi_j}_E)}{4}$.
    The subroutine estimates this probability by amplitude estimation and multiplies the estimate by four.
    It replaces values below $0$ by $0$ and values above $1$ by $1$ to obtain $\widehat{f}_j$.
    Algorithm~\ref{alg:max-estimation} returns the largest estimate, $F_{\mathrm{OAGF}}=\frac{d}{d+1}\max_{1\le j\le N}\widehat{f}_j+\frac{1}{d+1}$.
    Lemma~\ref{lem:haar-lifted-evaluation} accounts for the choice of $K$ and the joint simulation error; the accuracy budget is given below.

    \paragraph*{Correctness.}
    The test that determines whether $s_q=0$ (Lemma~\ref{lem:empty-filtered-subspace}) fails with probability at most $\frac{\delta}{8}$. If this test succeeds and outputs $s_q=0$, the output $\frac{1}{d+1}$ has OAGF error $\frac{dF_{\mathrm{OEF}}(\mathcal{E})}{d+1}\le2q\le\frac{\epsilon}{2}$ by Lemma~\ref{lem:spectral-truncation}.
    By Lemma~\ref{lem:filtered-state-preparation}, the probability that $s_q>0$ and filtering fails is at most $\frac{\delta}{4}$. This bound accounts for both a bad threshold and too few successful amplification attempts.
    By the union bound, the probability of a test or filtering failure is at most $\frac{\delta}{8}+\frac{\delta}{4}$.
    For $s_q>0$, we now condition on successful filtering and write
    \begin{equation}
    \label{eq:restricted-environment-optimization}
        F_q(\mathcal{E})
        :=\max_{\substack{\ket{\psi}_E\in\mathcal{H}_q\\\|\psi\|_2=1}}f(\ket{\psi}_E).
    \end{equation}
    Applying Lemma~\ref{lem:projected F} with $\mathcal{H}_P=\mathcal{H}_q$ and $\Pi_P=\Pi_q$, followed by Lemma~\ref{lem:spectral-truncation}, gives
    \begin{equation}
        0\le F_{\mathrm{OEF}}(\mathcal{E})-F_q(\mathcal{E})\le q\le\tau.
    \end{equation}

    We first analyze the candidates in
    Eq.~\eqref{eq:haar-candidate-batches},
    assuming exact access to their state-preparation unitaries.
    For unit vectors
    $\ket{\psi_1}_E,\ldots,\ket{\psi_N}_E\in\mathcal{H}_q$,
    define the good set
    \begin{equation}
        \mathcal{G}_\eta:=
        \left\{(\ket{\psi_1}_E,\ldots,\ket{\psi_N}_E):
        \max_{1\le j\le N}f(\ket{\psi_j}_E)\ge F_q(\mathcal{E})-2\eta\right\}.
    \end{equation}
    Whenever the candidates form an $\eta$-net, the Lipschitz continuity of $f$ (Lemma~\ref{lem:Lipschitz}) gives
    \begin{equation}
    \label{eq:net-maximum-comparison}
        F_q(\mathcal{E})-2\eta
        \le\max_{1\le j\le N}f(\ket{\psi_j}_E)
        \le F_q(\mathcal{E}).
    \end{equation}
    For $s_q\ge2$, Lemma~\ref{lem:epsilon-net}, with $\delta_{\mathrm{net}}=\frac{\delta}{4}$ and $s_q\le R$, bounds the probability of obtaining such a net from below by $1-\frac{\delta}{4}$.
    For $s_q=1$, the condition defining $\mathcal{G}_\eta$ holds for every choice of candidate vectors.
    Thus, with each $d\psi_j$ denoting normalized Haar measure on $\mathcal{H}_q$,
    \begin{equation}
        \Pr\!\left[(\ket{\psi_1}_E,\ldots,\ket{\psi_N}_E)\in\mathcal{G}_\eta\right]
        =\int_{\mathcal{G}_\eta}d\psi_1\cdots d\psi_N
        \ge1-\frac{\delta}{4}.
    \end{equation}

    For every fixed choice of candidate vectors, Lemma~\ref{lem:QSVT} and a union bound give
    \begin{equation}
        \Pr\!\left[
        \max_{1\le j\le N}\left|\widehat{f}_j-f(\ket{\psi_j}_E)\right|\le\xi
        \,\middle|\,\psi_1,\ldots,\psi_N\right]\ge1-N\beta.
    \end{equation}
    If the candidate vectors satisfy $(\ket{\psi_1}_E,\ldots,\ket{\psi_N}_E)\in\mathcal{G}_\eta$, this evaluation-success event implies
    \begin{align}
        |F_{\mathrm{OEF}}-F_{\mathrm{OEF}}(\mathcal{E})|
        &\le\left|F_{\mathrm{OEF}}-\max_{1\le j\le N}f(\ket{\psi_j}_E)\right|
          +\left|\max_{1\le j\le N}f(\ket{\psi_j}_E)-F_q(\mathcal{E})\right|
          +|F_q(\mathcal{E})-F_{\mathrm{OEF}}(\mathcal{E})|\nonumber\\
        &\le\xi+2\eta+q\le\xi+2\eta+\tau=\epsilon.
    \end{align}
    The relation between OEF and OAGF gives
    \begin{equation}
        |F_{\mathrm{OAGF}}-F_{\mathrm{OAGF}}(\mathcal{E})|
        =\frac{d}{d+1}|F_{\mathrm{OEF}}-F_{\mathrm{OEF}}(\mathcal{E})|
        \le\epsilon.
    \end{equation}

    Let $p(\widehat{f}\mid\psi)$ be the probability that $\texttt{\upshape Estimatef}$ outputs $\widehat{f}$ when run with exact access to $W_\psi$.
    In the notation of Lemma~\ref{lem:haar-lifted-evaluation}, we define $\mathcal{C}(W_\psi)$ as the diagonal density matrix representing this probability distribution:
    \begin{equation}
    \label{eq:classical-output-state}
        \mathcal{C}(W_\psi):=\sum_{\widehat{f}}p(\widehat{f}\mid\psi)\proj{\widehat{f}}.
    \end{equation}
    Here, the states $\ket{\widehat{f}}$ form an orthonormal basis labeled by the possible estimates.
    For the analysis, the final accuracy event corresponds to the two-outcome POVM $\{P_\epsilon,I-P_\epsilon\}$ on the $N$ output registers, where
    \begin{equation}
        P_\epsilon:=
        \sum_{\substack{\widehat{f}_1,\ldots,\widehat{f}_N\\
        \left|\frac{d\max_j\widehat{f}_j+1}{d+1}-F_{\mathrm{OAGF}}(\mathcal{E})\right|\le\epsilon}}
        \proj{\widehat{f}_1,\ldots,\widehat{f}_N}.
    \end{equation}
    The sum includes all outputs for which the final OAGF estimate is within $\epsilon$ of its true value.
    This condition depends only on the estimates and $\mathcal{E}$.
    Thus, the same $P_\epsilon$ can be used for every choice of candidate states.

    We now average over the candidate states, using the success bound $1-N\beta$ on $\mathcal{G}_\eta$:
    \begin{align}
        &\Tr\!\left[
        P_\epsilon
        \left(\int_{\ket{\psi}_E\sim\mathrm{Haar}(\mathcal{H}_q)}
        \mathcal{C}(W_\psi)\,d\psi\right)^{\otimes N}\right]\nonumber\\
        &\quad=\int\Tr\!\left[
        P_\epsilon\bigotimes_{j=1}^{N}\mathcal{C}(W_{\psi_j})\right]
        d\psi_1\cdots d\psi_N\nonumber\\
        &\quad\ge\int_{\mathcal{G}_\eta}(1-N\beta)\,d\psi_1\cdots d\psi_N
        \ge1-\frac{\delta}{4}-N\beta
        =1-\frac{\delta}{2}.
        \label{eq:ideal-candidate-success}
    \end{align}
    Since Lemma~\ref{lem:haar-lifted-evaluation} with $\varepsilon_{\mathrm{f}}=L\alpha$ bounds the trace-norm distance between the actual and ideal joint classical output states by
    \begin{equation}
        L\alpha+\gamma=2\gamma=\frac{\delta}{4},
    \end{equation}
    the failure probability after successful filtering is at most the ideal bound $\frac{\delta}{4}+N\beta$ from Eq.~\eqref{eq:ideal-candidate-success} plus $\frac{L\alpha+\gamma}{2}=\frac{\delta}{8}$.
    Combining this with the test and filtering failure bounds from Lemmas~\ref{lem:empty-filtered-subspace} and~\ref{lem:filtered-state-preparation}, respectively, gives
    \begin{equation}
        \label{eq:oagf-total-success}
        \begin{aligned}
            &\Pr\!\left[\left|F_{\mathrm{OAGF}}-F_{\mathrm{OAGF}}(\mathcal{E})\right|\le\epsilon\right]\\
            &\quad\ge 1
            -\underbrace{\frac{\delta}{8}}_{\text{test failure}}
            -\underbrace{\frac{\delta}{4}}_{\text{filtering failure}}
            -\underbrace{\frac{L\alpha+\gamma}{2}}_{\text{filtered input and lifting}}
            -\underbrace{\left(\frac{\delta}{4}+N\beta\right)}_{\text{net and ideal evaluations}}
            =1-\delta.
        \end{aligned}
    \end{equation}

    \textit{Query complexity.}\par\nopagebreak
    \textit{(i) Test.}~Determining whether $s_q=0$ uses
    \begin{equation}
        \bigO{\frac{1}{\epsilon^3\delta}
        \log\!\left(\frac{8}{\epsilon}\right)\log\!\left(\frac{16}{\delta}\right)}
    \end{equation}
    queries by Lemma~\ref{lem:empty-filtered-subspace} with $\tau=\frac{\epsilon}{4}$ and $\delta_{\mathrm{test}}=\frac{\delta}{8}$.

    \textit{(ii) Filtering.}~To produce $K$ copies of one candidate, Algorithm~\ref{alg:max-estimation} uses $b=R(K-1)+1\le RK$ filtered registers. With $\gamma=\frac{\delta}{8}$ and $\xi=\frac{\epsilon}{2}$, the choice of $K$ in the algorithm gives
    \begin{equation}
        \label{eq:filtering-copy-bound}
        L=Nb\le RNK=\bigO{\frac{RN^2}{\delta\epsilon^6}\log^2\!\left(\frac{2}{\beta}\right)}.
    \end{equation}
    Preparing these $L$ registers by Lemma~\ref{lem:filtered-state-preparation} with $\tau=\frac{\epsilon}{4}$, $\varepsilon_{\mathrm{f}}=L\alpha=\frac{\delta}{8}$, and $\delta_{\mathrm{f}}=\frac{\delta}{4}$ takes
    \begin{equation}
        \bigO{\frac{L}{\epsilon^{\frac{5}{2}}\delta}
        \log\!\left(\frac{L}{\epsilon\delta}\right)}
    \end{equation}
    queries, including those used to prepare the input states.

    \textit{(iii) Candidate evaluation.}~The $N$ calls to $\texttt{\upshape Estimatef}$ use
    \begin{equation}
        \label{eq:total-cost-before-simplification}
        \bigO{\frac{N}{\xi^3}\log\!\left(\frac{2}{\beta}\right)}
        =\bigO{\frac{N}{\epsilon^3}\log\!\left(\frac{8N}{\delta}\right)}
    \end{equation}
    queries to $U$ and $U^\dagger$ by Lemma~\ref{lem:QSVT}. Sample-to-query lifting simulates the preparation-unitary queries within these calls using the $K$ copies of each candidate and known operations.

    Adding the three costs gives a total of
    \begin{equation}
        \bigO{
            \underbrace{\frac{1}{\epsilon^3\delta}\log\!\left(\frac{8}{\epsilon}\right)
                \log\!\left(\frac{16}{\delta}\right)}_{\text{(i) test}}
            +\underbrace{\frac{L}{\epsilon^{\frac{5}{2}}\delta}
                \log\!\left(\frac{L}{\epsilon\delta}\right)}_{\text{(ii) filtering}}
            +\underbrace{\frac{N}{\epsilon^3}\log\!\left(\frac{8N}{\delta}\right)}_{\text{(iii) \texttt{Estimatef} calls}}
        }
    \end{equation}
    queries to $U$ and $U^\dagger$.
    The choice of $N$ in Algorithm~\ref{alg:max-estimation} comes from the net bound in Lemma~\ref{lem:epsilon-net}, using $s_q\le R$ and $\delta_{\mathrm{net}}=\frac{\delta}{4}$.
    Substituting $\eta=\frac{\epsilon}{8}$ into this choice and using $\beta=\frac{\delta}{4N}$ and $\ell=R\log\!\left(\frac{64}{\epsilon}\right)+\log\!\left(\frac{4}{\delta}\right)$ gives
    \begin{equation}
        \label{eq:query-cost-parameter-bounds}
        N=\bigO{\left(\frac{64}{\epsilon}\right)^{2R-1}\ell},
        \qquad
        \log\!\left(\frac{2}{\beta}\right)=\log\!\left(\frac{8N}{\delta}\right)=O(\ell).
    \end{equation}
    Substituting the logarithmic bound into Eq.~\eqref{eq:filtering-copy-bound} gives
    \begin{equation}
        L=\bigO{\frac{RN^2}{\delta\epsilon^6}\ell^2},
        \qquad
        \log\!\left(\frac{L}{\epsilon\delta}\right)=O(\ell),
    \end{equation}
    since $\log N$, $\log R$, $\log\ell$, $\log\!\left(\frac{1}{\epsilon}\right)$, and $\log\!\left(\frac{1}{\delta}\right)$ are all $O(\ell)$.
    The total query count is therefore
    \begin{equation}
        \bigO{\frac{\ell^2}{\epsilon^3\delta}
        +\frac{RN^2}{\delta^2\epsilon^{\frac{17}{2}}}\ell^3
        +\frac{N}{\epsilon^3}\ell}.
    \end{equation}
    The filtering bound dominates the other two terms since $R,N,\ell\ge1$ and $0<\epsilon,\delta<1$. Substituting the bound on $N$ from Eq.~\eqref{eq:query-cost-parameter-bounds} gives
    \begin{equation}
        \bigO{\frac{R}{\delta^2\epsilon^{\frac{17}{2}}}
        \left(\frac{64}{\epsilon}\right)^{4R-2}\ell^5},
    \end{equation}
    as stated in Eq.~\eqref{eq:main-query-upper-bound}. Since $R\le\left\lceil\frac{8}{\epsilon}\right\rceil$, this bound is uniform in the environment dimension $m$.

    \paragraph*{Gate complexity for Theorem~\ref{thm:upper bound}(ii).}
    For constant Kraus rank, inverse-polynomial accuracy, and constant failure probability, the query and copy counts and the QSVT degrees are polynomial in $n$. The supplied dilation acts on $\poly(n)$ qubits, so the environment register has width $\lceil\log_2 m\rceil=O(\poly(n))$.
    Random purification and sample-to-query lifting therefore use $\poly(n)$ gates~\cite[Lemma~2.11 and Sec.~2.4]{TangWrightZhandryConjugate}.
    The QSVT phases can be computed classically in $\poly(n)$ time~\cite[Corollary~10]{Gilyen2019SVD}, and the known state preparations, swaps, and reflections also have polynomial gate cost.
    Each call to $U$ or $U^\dagger$ uses $\poly(n)$ gates, so the total running time is $\poly(n)$.
\end{proof}

\subsection{Supporting lemmas.}
\label{sub:lemma proof}

\subsubsection{Subspace reduction.}
\label{subsubsec:subspace-reduction}
We establish the variational representation and the guarantees used to construct the search subspace.

\begin{Lem}[Variational representation]
    \label{lem:oef-variational}
    For any unitary Stinespring dilation $U$ of $\mathcal{E}$,
    \begin{equation}
    \label{eq:environment-optimization}
        F_{\mathrm{OEF}}(\mathcal{E}) =\max_{\substack{\ket{\psi}_E\in\mathcal{H}_E\\\|\psi\|_2=1}}f(\ket{\psi}_E),
    \end{equation}
    where the maximization is taken over all pure quantum states on the environment system $E$.
\end{Lem}

\begin{proof}[Proof of Lemma~\ref{lem:oef-variational}]
Since the entanglement fidelity between the channel $\mathcal{E}$ and a unitary $W$ is given by $\sum_i \frac{\abs{\Tr(W^\dagger K_i)}^2}{d^2}$, we have
\begin{equation}
\label{eq:oef-kraus-overlap}
    F_{\mathrm{OEF}}(\mathcal{E})=\frac{1}{d^2}\max_{W\in \mathrm{U}(d)}\sum_{i=1}^{m}\left|\Tr(W^\dagger K_i)\right|^2.
\end{equation}
For fixed $W\in \mathrm{U}(d)$, define $z_i:=\Tr(W^\dagger K_i)$ for $i=1,\ldots,m$, and write $z=(z_1,\ldots,z_m)\in\mathbb{C}^m$.
For coefficients $c_i$ satisfying $\sum_i|c_i|^2=1$, the Cauchy-Schwarz inequality gives
\begin{equation}
    \left|\sum_{i=1}^{m}c_i z_i\right|^2\le\left(\sum_{i=1}^{m}|c_i|^2\right)\left(\sum_{i=1}^{m}|z_i|^2\right)=\|z\|_2^2.
\end{equation}
If $z\ne0$, equality is attained by choosing $c_i=\frac{\overline{z_i}}{\|z\|_2}$; if $z=0$, equality holds for every admissible choice of the coefficients.
Therefore,
\begin{equation}
\label{eq:cauchy-schwarz-max}
    \sum_{i=1}^{m}|z_i|^2=\max_{\sum_i|c_i|^2=1}\left|\sum_{i=1}^{m}c_i z_i\right|^2.
\end{equation}
Combining this identity with $\max_{W\in \mathrm{U}(d)}|\Tr(W^\dagger A)|=\|A\|_1$ gives
\begin{align}
    F_{\mathrm{OEF}}(\mathcal{E})
    &=\frac{1}{d^2}\max_{\sum_i|c_i|^2=1}\max_{W\in \mathrm{U}(d)}\left|\Tr\!\left(W^\dagger\sum_{i=1}^{m}c_iK_i\right)\right|^2\nonumber\\
    &=\frac{1}{d^2}\max_{\sum_i|c_i|^2=1}\left\|\sum_{i=1}^{m}c_iK_i\right\|_1^2.
\end{align}
Setting $\ket{\psi}_E=\sum_{i=1}^{m}\overline{c_i}\ket{i}$ yields $A_{\psi}=\sum_{i=1}^{m}c_iK_i$ and proves the claim.
\end{proof}

For any vector $\ket{\psi}_E\in\mathcal{H}_E$, Eq.~\eqref{eq:quantum state G} and the expression for $A_\psi$ in Eq.~\eqref{eq:Apsi} imply
\begin{equation}
\label{eq:Apsi-support}
    \bra{\psi}_E G\ket{\psi}_E=\frac{1}{d}\Tr(A_\psi^\dagger A_\psi)=\frac{1}{d}\|A_\psi\|_2^2.
\end{equation}

\begin{Lem}[Loss from subspace restriction]
\label{lem:projected F}
    Let $\mathcal{H}_P$ be a nonzero subspace of $\mathcal{H}_E$, and let $\Pi_P$ be the orthogonal projector onto $\mathcal{H}_P$.
    Define
    \begin{equation}
        F_P(\mathcal{E})
        :=\max_{\substack{\ket{\psi}_E\in\mathcal{H}_P\\\lVert\psi\rVert_2=1}}f(\ket{\psi}_E).
    \end{equation}
    Then
    \begin{equation}
        0\le F_{\mathrm{OEF}}(\mathcal{E})-F_P(\mathcal{E})\le\left\|(I_{\mathcal{H}_E}-\Pi_P)G(I_{\mathcal{H}_E}-\Pi_P)\right\|_\infty.
    \end{equation}
\end{Lem}

\begin{proof}[Proof of Lemma~\ref{lem:projected F}]
    For $W\in \mathrm{U}(d)$, define
    \begin{equation}
        \ket{v(W)}_E:=\frac{1}{d}\sum_{i=1}^{m}\Tr(W^\dagger K_i)\ket{i}_E.
    \end{equation}
    For every $\ket{\psi}_E=\sum_{i=1}^{m}\psi_i\ket{i}_E\in\mathcal{H}_E$,
    \begin{align}
        |\braket{\psi|v(W)}|^2
        &=\frac{1}{d^2}\left|\sum_{i=1}^{m}\overline{\psi_i}\Tr(W^\dagger K_i)\right|^2\nonumber\\
        &=\frac{1}{d^2}\left|\Tr(W^\dagger A_\psi)\right|^2\nonumber\\
        &\le\frac{1}{d^2}\Tr(W^\dagger W)\Tr(A_\psi^\dagger A_\psi)\nonumber\\
        &=\frac{\|A_\psi\|_2^2}{d}
          =\bra{\psi}_E G\ket{\psi}_E,
    \end{align}
    where the inequality follows from the Hilbert--Schmidt Cauchy--Schwarz inequality, and the last equality follows from Eq.~\eqref{eq:Apsi-support}.
    Consequently,
    \begin{equation}
        \ket{v(W)}\!\bra{v(W)}_E\preceq G,
    \end{equation}
    where $A \preceq B$ for matrices $A, B$ denotes $B-A\ge 0$.
    Writing $Q_P:=I_{\mathcal{H}_E}-\Pi_P$, we obtain
    \begin{equation}
        Q_P\ket{v(W)}\!\bra{v(W)}_E Q_P\preceq Q_PGQ_P,
    \end{equation}
    and hence
    \begin{equation}
        \|Q_P\ket{v(W)}_E\|_2^2=\left\|Q_P\ket{v(W)}\!\bra{v(W)}_E Q_P\right\|_\infty\le\|Q_PGQ_P\|_\infty.
    \end{equation}
    Equation~\eqref{eq:oef-kraus-overlap} gives
    \begin{equation}
        F_{\mathrm{OEF}}(\mathcal{E})=\max_{W\in \mathrm{U}(d)}\|\ket{v(W)}_E\|_2^2.
    \end{equation}
    Moreover, Eqs.~\eqref{eq:Apsi} and~\eqref{eq:environment-objective}, together with the trace-norm variational identity, imply
    \begin{align}
        F_P(\mathcal{E})&=\max_{\substack{\ket{\psi}_E\in\mathcal{H}_P\\\|\psi\|_2=1}}\max_{W\in \mathrm{U}(d)}|\braket{\psi|v(W)}|^2\nonumber\\
        &=\max_{W\in \mathrm{U}(d)}\|\Pi_P\ket{v(W)}_E\|_2^2.
    \end{align}
    Since $\Pi_P$ and $Q_P$ are complementary orthogonal projectors,
    \begin{align}
        F_{\mathrm{OEF}}(\mathcal{E})
        &=\max_{W\in \mathrm{U}(d)}\left(\|\Pi_P\ket{v(W)}_E\|_2^2+\|Q_P\ket{v(W)}_E\|_2^2\right)\nonumber\\
        &\le F_P(\mathcal{E})+\|Q_PGQ_P\|_\infty.
    \end{align}
    Finally, $F_P(\mathcal{E})\le F_{\mathrm{OEF}}(\mathcal{E})$ because its maximization is restricted to unit vectors in $\mathcal{H}_P$.
\end{proof}

\begin{Lem}[Spectral truncation]
    \label{lem:spectral-truncation}
    Let $G$ be a quantum state on the $m$-dimensional environment space, and let $q\in\left[\frac{\tau}{2},\tau\right]$ for $\tau\in(0,1]$.
    For $\Pi_q=\mathbf{1}_{[q,1]}(G)$, $\mathcal{H}_q=\operatorname{img}(\Pi_q)$, and $s_q=\rank\Pi_q$, one has
    \begin{equation}
        \left\|(I_{\mathcal{H}_E}-\Pi_q)G(I_{\mathcal{H}_E}-\Pi_q)\right\|_\infty\le q,
        \qquad s_q\le\min\{\rank G,\left\lfloor\frac{1}{q}\right\rfloor\}\le\frac{2}{\tau}.
    \end{equation}
    Moreover, if $\left\|G\right\|_\infty<q$, i.e. $s_q=0$,
    \begin{equation}
        F_{\mathrm{OEF}}(\mathcal{E})\le 2q.
    \end{equation}
\end{Lem}
\begin{proof}[Proof of Lemma~\ref{lem:spectral-truncation}]
    The first claim follows directly from the definition of the operator norm.
    Since $G$ is a quantum state,
    \begin{align}
        1&=\sum_i\lambda_i\\
        &\ge s_q q+\sum_{\lambda_i<q}\lambda_i\\
        &\ge s_q q,
    \end{align}
    where $\lambda_i$ are the eigenvalues of $G$.
    Thus,
    \begin{equation}
        s_q\le\frac{1}{q}.
    \end{equation}
    Since $q>0$, the range of $\Pi_q$ is in the support of $G$, so $s_q\le\rank G$. As $s_q$ is an integer, this gives $s_q\le\min\{\rank G,\left\lfloor\frac{1}{q}\right\rfloor\}$.

    If $\left\|G\right\|_\infty<q$, for any nonzero subspace $P\subset\mathcal{H}_E$,
    \begin{equation}
        F_{\mathrm{OEF}}(\mathcal{E})\le\left\|\Pi_P G\Pi_P\right\|_\infty+\left\|Q_P GQ_P\right\|_\infty\le2q.
    \end{equation}
\end{proof}

\begin{Lem}[\texttt{\upshape FilterG}]
    \label{lem:filtered-state-preparation}
    Let $L$ be a positive integer and let $\varepsilon_{\mathrm{f}},\delta_{\mathrm{f}}\in(0,1)$.
    Let $q$ be sampled uniformly from the grid $T(\tau,M)$ in Eq.~\eqref{eq:spectral-threshold-grid} with $M=2^{\left\lceil\log_2\frac{16}{\tau\delta_{\mathrm f}}\right\rceil}$.
    The same threshold $q$ is used for all $L$ registers and we assume $s_q>0$.
    There exists an algorithm that outputs $\Omega_{q,L}$ satisfying
    \begin{equation}
        \left\|\Omega_{q,L}-\left(\frac{\Pi_q}{s_q}\right)^{\otimes L}\right\|_1\le\varepsilon_{\mathrm{f}}
    \end{equation}
    with success probability at least $1-\delta_{\mathrm{f}}$ using
    \begin{equation}
        \bigO{\frac{L}{\tau^{\frac{5}{2}}\delta_{\mathrm{f}}}\log\left(\frac{L}{\tau\varepsilon_{\mathrm{f}}}\right)} 
    \end{equation}
    queries to $U$ and $U^\dagger$.
\end{Lem}
Algorithm~\ref{alg:filter-g} describes $\texttt{\upshape FilterG}$ for $s_q>0$.
The values $d,\tau$ and query access to $U,U^\dagger$ are implicit in the subroutine notation.
\begin{algorithm}[H]
    \caption{$\texttt{\upshape FilterG}(q,L,\varepsilon_{\mathrm{f}},\delta_{\mathrm{f}})$}
    \label{alg:filter-g}
    \small
    \begin{algorithmic}[1]
        \Require Query access to $U,U^\dagger$; $d,\tau$; the previously sampled threshold $q\in T(\tau,M)$; $L\in\mathbb{N}$; $\varepsilon_{\mathrm{f}},\delta_{\mathrm{f}}\in(0,1)$
        \Ensure A joint $L$-register state $\Omega_{q,L}$
        \State $\alpha\gets\frac{\varepsilon_{\mathrm{f}}}{L}$, $\Delta\gets\frac{h}{4}$, where $h$ is the grid spacing in Eq.~\eqref{eq:spectral-threshold-grid}
        \State Construct an exact block-encoding $U_G$ of $G$ using $U,U^\dagger$
        \State Compute $\widehat{g_q}$ as in Eq.~\eqref{eq:filter-polynomial-approximation}
        \State Construct $U_{\widehat{g_q}}$ by QSVT using $U_G,U_G^\dagger$, with synthesis error $\le \frac{q\alpha}{32}$
        \For{$j=1,\ldots,L$}
            \State Prepare a purification of $G$ using $U$, with filter ancillas in $\ket{0}_{\mathrm{anc}}$
            \State Apply $U_{\widehat{g_q}}$ and amplify the zero-ancilla outcome
            \State Postselect $\ket{0}_{\mathrm{anc}}$ and retain only the environment state $\sigma_j$
        \EndFor
        \State \Return $\Omega_{q,L}\gets\bigotimes_{j=1}^{L}\sigma_j$
    \end{algorithmic}
\end{algorithm}
\begin{proof}
    Keep the sampled threshold $q$ fixed and set the single-register accuracy to $\alpha:=\frac{\varepsilon_{\mathrm{f}}}{L}$.
    Choose the polynomial approximation and block-encoding synthesis accuracies to be $\frac{q\alpha}{32}$ each, so the total encoded-block error is at most $\frac{q\alpha}{16}$.
    Set $\Delta:=\frac{h}{4}$. The choice of $M$ gives $\Delta=\Theta(\tau^2\delta_{\mathrm f})$.
    We first prepare a purification of $G$ using $U$ and then use the construction in Ref.~\cite[Lemma~45]{Gilyen2019SVD} to obtain an exact block-encoding $U_G$ of $G$.
    Each query to $U_G$ or $U_G^\dagger$ requires only a constant number of queries to $U$ and $U^\dagger$.
    We approximate the filter
    \begin{equation}
        g_q(x)=
        \begin{cases}
            \frac{\sqrt{q}}{2}\frac{1}{\sqrt{x}}\quad&(x\ge q)
            \\ 0\quad&(x<q).
        \end{cases}
    \end{equation}
    Apply the bounded polynomial approximation theorem~\cite[Theorem~68]{Gilyen2019SVD} to the even extension $x\mapsto g_q(|x|)$ on
    \begin{equation}
        [-1,-q-\Delta]\cup[-q+\Delta,q-\Delta]\cup[q+\Delta,1].
    \end{equation}
    Taking the real and even parts of the resulting polynomial gives an efficiently computable real even polynomial $\widehat{g_q}$ of degree $\bigO{\frac{1}{\Delta}\log\!\left(\frac{16}{q\alpha}\right)}$ satisfying
    \begin{equation}
    \label{eq:filter-polynomial-approximation}
        \left\|\widehat{g_q}-g_q\right\|_{[0,q-\Delta]\cup[q+\Delta,1]}\le\frac{q\alpha}{32},
        \qquad \left\|\widehat{g_q}\right\|_{[-1,1]}\le\frac{1}{2}.
    \end{equation}
    Its parity and bound allow real-polynomial QSVT~\cite[Corollary~18]{Gilyen2019SVD} to block-encode $\widehat{g_q}(G)$ using $U_G$ and $U_G^\dagger$, since $G$ is positive semidefinite.

    The smallest grid point is $\frac{\tau}{2}$. If $\lambda_i<\frac{\tau}{4}$, then every grid point $q$ satisfies $q-\lambda_i>\frac{\tau}{4}\ge\Delta$, since $\Delta=\frac{\tau}{8M}$.
    It therefore suffices to count the eigenvalues of $G$ that are at least $\frac{\tau}{4}$. Their number satisfies
    \begin{equation}
        \label{eq:grid-eigenvalue-count}
        \#\left\{i:\lambda_i\ge\frac{\tau}{4}\right\}
        \le\frac{4}{\tau}\sum_{i:\,\lambda_i\ge\frac{\tau}{4}}\lambda_i
        \le\frac{4}{\tau}\Tr G
        =\frac{4}{\tau}.
    \end{equation}
    For each such eigenvalue, the interval $[\lambda_i-\Delta,\lambda_i+\Delta]$ contains at most one grid point because $2\Delta<h$. Since the $M$ grid points are equally likely, we obtain
    \begin{equation}
        \label{eq:grid-threshold-bound}
        \Pr\!\left[\exists i:\ |q-\lambda_i|\le\Delta\right]
        \le\frac{4}{\tau M}.
    \end{equation}
    Our choice of $M$ bounds this probability by $\frac{\delta_{\mathrm f}}{4}$, leaving at least $\frac{\delta_{\mathrm f}}{2}$ for the amplification attempts.

    Implement the QSVT unitary $U_{\hat{g_q}}$ with synthesis error at most $\frac{q\alpha}{32}$, and denote its actual encoded block by
    \begin{equation}
        \widehat{B}_q:=(\bra{0}_{\mathrm{anc}}\otimes I_E)U_{\hat{g}_q}(\ket{0}_{\mathrm{anc}}\otimes I_E).
    \end{equation}
    Thus $U_{\widehat{g_q}}$ is a $\left(1,a,\frac{q\alpha}{32}\right)$-block-encoding of $\widehat{g_q}(G)$, where $a$ is its number of ancilla qubits.
    Suppose that $q$ is at distance greater than $\Delta$ from every eigenvalue of $G$. Setting $B_q:=g_q(G)$ gives
    \begin{equation}
        \|\widehat{B}_q-B_q\|_\infty
        \le\|\widehat{B}_q-\widehat{g_q}(G)\|_\infty+\|\widehat{g_q}(G)-g_q(G)\|_\infty
        \le\frac{q\alpha}{16}.
    \end{equation}
    Both satisfy $\|B_q\|_\infty\le1$ and $\|\widehat{B}_q\|_\infty\le1$.
    Applying $U_{\widehat{g_q}}$ to initialized ancillas and $G$ and postselecting on the all-zero outcome of the ancilla measurement produces the subnormalized state $\tilde{\sigma}=\widehat{B}_q G\widehat{B}_q^\dagger$.
    The ideal subnormalized state is
    \begin{equation}
        B_qGB_q^\dagger=\frac{q}{4}\Pi_q=\frac{qs_q}{4}\rho_q.
    \end{equation}
    We compare these encoded blocks directly:
    \begin{align}
        \left\|\tilde{\sigma}-\frac{q}{4}\Pi_q\right\|_1
        &\le\|(\widehat{B}_q-B_q)G\widehat{B}_q^\dagger\|_1+\|B_qG(\widehat{B}_q^\dagger-B_q^\dagger)\|_1\\
        &\le\|\widehat{B}_q-B_q\|_\infty\|G\|_1\|\widehat{B}_q\|_\infty
        +\|B_q\|_\infty\|G\|_1\|\widehat{B}_q-B_q\|_\infty\\
        &\le\frac{q\alpha}{8}.
    \end{align}
    In particular, $|\Tr\tilde{\sigma}-\frac{qs_q}{4}|\le \frac{q\alpha}{8}<\frac{qs_q}{4}$ for $s_q>0$.
    Thus $\Tr\tilde{\sigma}>0$ and the normalized output $\sigma:=\frac{\tilde{\sigma}}{\Tr\tilde{\sigma}}$ satisfies
    \begin{align}
        \|\sigma-\rho_q\|_1
        &\le\left\|\frac{\tilde{\sigma}}{\Tr\tilde{\sigma}}-\frac{4\tilde{\sigma}}{qs_q}\right\|_1
        +\frac{4}{qs_q}\left\|\tilde{\sigma}-\frac{q}{4}\Pi_q\right\|_1\\
        &=\frac{4}{qs_q}\left(\left|\Tr\tilde{\sigma}-\frac{qs_q}{4}\right|
        +\left\|\tilde{\sigma}-\frac{q}{4}\Pi_q\right\|_1\right)\\
        &\le\frac{\alpha}{s_q}\le\alpha.
    \end{align}
    The choice of filter accuracy uses the known threshold $q$ and does not require knowing $s_q$.
    Amplitude amplification prepares $\sigma$ with constant success probability using $O\!\left(\frac{1}{\sqrt{qs_q}}\right)$ calls to the filter circuit and its inverse~\cite{Brassard2000}.
    Since $s_q\ge1$, we can set a uniform limit of $O\!\left(\frac{1}{\sqrt{q}}\right)$ calls per attempt using only $q$.

    In the finite-bit implementation of Algorithm~\ref{alg:max-estimation}, we repeat these attempts on fresh inputs to collect $L$ successful filter outputs.
    Let $X$ denote the number of successes in $T$ independent attempts with this $q$.
    After amplitude amplification, each attempt succeeds with probability
    $p\ge p_0$, where $p_0>0$ is a universal constant. Thus $\mathbb{E}[X]=Tp$.

    The algorithm aborts if $X<L$. When $Tp\ge 2L$, the Chernoff bound gives
    \begin{equation}
        \label{eq:filtering-abort-chernoff}
        \Pr[X<L]\le\Pr\!\left[X\le\frac{Tp}{2}\right]
        \le\exp\!\left(-\frac{Tp}{8}\right).
    \end{equation}
    To ensure that $Tp\ge 2L$ and that the abort probability is at most
    $\frac{\delta_{\mathrm f}}{2}$, we choose
    \[
        T=\left\lceil\frac{1}{p_0}
        \max\!\left\{2L,\,8\log\frac{2}{\delta_{\mathrm f}}\right\}
        \right\rceil.
    \]
    Algorithm~\ref{alg:max-estimation} sets $L=Nb$ with $b\ge1$ and
    $N\ge\log\frac{4}{\delta}=\log\frac{1}{\delta_{\mathrm f}}$.
    Hence $T=O(L)$. Adding the probability of selecting a bad threshold,
    which is at most $\frac{\delta_{\mathrm f}}{2}$, gives a total
    filtering failure probability of at most $\delta_{\mathrm f}$.

    Using the same $q$ and filter, prepare $L$ independent output states $\sigma_1,\ldots,\sigma_L$ with single-register accuracy $\alpha$ and write $\Omega_{q,L}:=\bigotimes_{j=1}^{L}\sigma_j$.
    On the common threshold condition, the joint error satisfies
    \begin{equation}
        \left\|\Omega_{q,L}-\rho_q^{\otimes L}\right\|_1
        \le\sum_{j=1}^{L}\left\|\sigma_j-\rho_q\right\|_1
        \le L\alpha=\varepsilon_{\mathrm{f}}.
    \end{equation}
    We sample $q$ only once and use it for all $L$ preparations.
    The threshold failure probability in Eq.~\eqref{eq:grid-threshold-bound} is therefore counted once for all $L$ registers.
    Using the stated amplification cost, the query complexity in total is
    \begin{equation}
        L\,\bigO{\frac{1}{\Delta}\log\left(\frac{16}{q\alpha}\right)}\times\bigO{\frac{1}{\sqrt{q}}}=\bigO{\frac{L}{\tau^{\frac{5}{2}}\delta_{\mathrm{f}}}\log\left(\frac{L}{\tau\varepsilon_{\mathrm{f}}}\right)}.
    \end{equation}
\end{proof}
\begin{Lem}
    \label{lem:empty-filtered-subspace}
    For $\tau,\delta_{\mathrm{test}}\in(0,1)$, let $q$ be sampled uniformly from the grid $T(\tau,M)$ in Eq.~\eqref{eq:spectral-threshold-grid} with $M=2^{\left\lceil\log_2\frac{8}{\tau\delta_{\mathrm{test}}}\right\rceil}$.
    There exists an algorithm that determines whether $\left\|G\right\|_\infty<q$ or $\left\|G\right\|_\infty\ge q$ with success probability at least $1-\delta_{\mathrm{test}}$, over the draw of $q$ and the measurement outcomes, using
    \begin{equation}
        \bigO{\frac{1}{\tau^3\delta_{\mathrm{test}}}\log\left(\frac{2}{\tau}\right)\log\left(\frac{2}{\delta_{\mathrm{test}}}\right)}
    \end{equation}
    queries to $U$ and $U^\dagger$.
\end{Lem}
\begin{proof}
    Set $\Delta=\frac{h}{4}=\Theta(\tau^2\delta_{\mathrm{test}})$. By Eq.~\eqref{eq:grid-threshold-bound} in the proof of Lemma~\ref{lem:filtered-state-preparation}, the probability that $q$ lies within $\Delta$ of an eigenvalue of $G$ is at most $\frac{4}{\tau M}\le\frac{\delta_{\mathrm{test}}}{2}$.
    We first consider the case where no eigenvalue of $G$ lies within $\Delta$ of $q$. Write $g:=g_q$ for the filter in that proof, and denote its implemented encoded block by $\widehat{g}(G)$. Allocate error $\frac{\eta}{2}$ each to polynomial approximation and QSVT synthesis, so that
    \begin{equation}
        \left\|\widehat{g}(G)-g(G)\right\|_{\infty}\le\eta.
    \end{equation}
    Let $S$ be the system register, let $S'$ be a $d$-dimensional reference, and choose a known unitary $U_{\Psi}$ satisfying $U_{\Psi}\ket{0}_S\ket{0}_{S'}=\ket{\Phi^+}_{SS'}$. Define a purification of $G$ as
    \begin{equation}
        \ket{\Gamma_G}:=(U)_{SE}(U_{\Psi})_{SS'}\ket{0}_S\ket{0}_{S'}\ket{0}_E.
    \end{equation}
    Here, the subscripts on the operators indicate the registers on which they act.
    After applying block-encodings of the ideal filter $g(G)$ and its approximation $\widehat{g}(G)$, respectively, we obtain
    \begin{align}
        \left(I_{SS'}\otimes g(G)\right)\ket{\Gamma_G}_{SS'E}\otimes \ket{0}_a+\ket{\perp}
        \\ \left(I_{SS'}\otimes \widehat{g}(G)\right)\ket{\Gamma_G}_{SS'E}\otimes \ket{0}_a+\ket{\perp'},
    \end{align}
    where $\ket{\perp}$ and $\ket{\perp'}$ satisfy $(I\otimes \proj{0}_a)\ket{\perp}=(I\otimes \proj{0}_a)\ket{\perp'}=0$.
    We perform amplitude estimation on the subspace $\mathcal{H}_S\otimes \mathcal{H}_{S'}\otimes\mathcal{H}_E\otimes \mathrm{Span}(\ket{0}_a)$.
    Since these amplitudes are $\sqrt{\Tr g(G)Gg(G)^\dagger}$ and $\sqrt{\Tr \widehat{g}(G)G\widehat{g}(G)^\dagger}$, respectively,
    \begin{align}
        \left|\sqrt{\Tr \widehat{g}(G)G\widehat{g}(G)^\dagger}-\sqrt{\Tr g(G)Gg(G)^\dagger}\right|
        & =\left|\left\|(I\otimes \widehat{g}(G))\ket{\Gamma_G}_{SS'E}\right\|_2-\left\|(I\otimes g(G))\ket{\Gamma_G}_{SS'E}\right\|_2\right|
        \\ &\le\left\|\left(I\otimes (\widehat{g}(G)-g(G))\right)\ket{\Gamma_G}_{SS'E}\right\|_2
        \\ &\le\left\|\left(I\otimes (\widehat{g}(G)-g(G))\right)\right\|_\infty\left\|\ket{\Gamma_G}_{SS'E}\right\|_2
        \\ &\le\eta.
    \end{align}
    Let's take $\eta=\frac{\sqrt{q}}{8}$, then if $s_q=0$, we obtain $\Tr g(G)Gg(G)^\dagger=0$ and $\sqrt{\Tr \widehat{g}(G)G\widehat{g}(G)^\dagger}\le\eta=\frac{\sqrt{q}}{8}$.
    In the case of $s_q\neq 0$, we obtain $\sqrt{\Tr \widehat{g}(G)G\widehat{g}(G)^\dagger}\ge\frac{\sqrt{qs_q}}{2}-\eta\ge\frac{3\sqrt{q}}{8}$.
    Thus the success probability $p:=\Tr(\widehat{g}(G)G\widehat{g}(G)^\dagger)$ satisfies $p\le \frac{q}{64}$ if $s_q=0$ and $p\ge\frac{9q}{64}$ otherwise. Amplitude estimation~\cite{Brassard2000} and median amplification estimate $p$ to error $\frac{q}{32}$ with failure probability at most $\frac{\delta_{\mathrm{test}}}{2}$, using $\bigO{\frac{1}{q}\log\!\left(\frac{2}{\delta_{\mathrm{test}}}\right)}$ circuit applications. Declare $s_q=0$ if the estimate is below $\frac{5q}{64}$, and $s_q>0$ otherwise. The total failure probability is therefore at most $\frac{\delta_{\mathrm{test}}}{2}+\frac{\delta_{\mathrm{test}}}{2}=\delta_{\mathrm{test}}$.
    Overall, $\bigO{\frac{1}{\Delta}\log\left(\frac{2}{\eta}\right)}\times \bigO{\frac{1}{q}\log\left(\frac{2}{\delta_{\mathrm{test}}}\right)}=\bigO{\frac{1}{\tau^3\delta_{\mathrm{test}}}\log\left(\frac{2}{\tau}\right)\log\left(\frac{2}{\delta_{\mathrm{test}}}\right)}$.
\end{proof}

\subsubsection{Finite candidate search.}
\label{subsubsec:finite-candidate-search}

\begin{Lem}[Lipschitz continuity]
    \label{lem:Lipschitz}
    For any unit vectors $\ket{\psi}_E$ and $\ket{\phi}_E$ in $\mathcal{H}_E$,
    \begin{equation}
        |f\left(\ket{\psi}_E\right)-f\left(\ket{\phi}_E\right)|\le 2\|\ket{\psi}_E-\ket{\phi}_E\|_2.
    \end{equation}
\end{Lem}

\begin{proof}[Proof of Lemma~\ref{lem:Lipschitz}]
    For any vector $\ket{z}_E\in\mathcal{H}_E$, recall that $A_z=(I_{\mathcal{H}}\otimes\bra{z}_E)U(I_{\mathcal{H}}\otimes\ket{0}_E)$.
    Because $U$ is unitary,
    \begin{equation}
        \|A_z\|_\infty\le\|\ket{z}_E\|_2,\qquad\|A_z\|_1\le d\|\ket{z}_E\|_2.
    \end{equation}
    In particular, $\|A_\psi\|_1,\|A_\phi\|_1\le d$ and $\|A_\psi-A_\phi\|_1\le d\|\ket{\psi}_E-\ket{\phi}_E\|_2$.
    The triangle inequality gives
    \begin{align}
        |f(\ket{\psi}_E)-f(\ket{\phi}_E)|
        &\le\frac{\left(\|A_\psi\|_1+\|A_\phi\|_1\right)\|A_\psi-A_\phi\|_1}{d^2}\nonumber\\
        &\le 2\|\ket{\psi}_E-\ket{\phi}_E\|_2.
    \end{align}
\end{proof}

\begin{Lem}[A net from Haar-random vectors]
    \label{lem:epsilon-net}
    Let $\dim\mathcal{H}_q=s_q\ge2$ and $\eta,\delta_{\mathrm{net}}\in(0,1)$.
    If $\ket{\psi_1}_E,\ldots,\ket{\psi_N}_E$ are independent Haar-random unit vectors in $\mathcal{H}_q$ and
    \begin{equation}
    \label{eq:random-net-size}
        N\ge\left(\frac{8}{\eta}\right)^{2s_q-1}\left[2s_q\log\!\left(\frac{6}{\eta}\right)+\log\!\left(\frac{1}{\delta_{\mathrm{net}}}\right)\right],
    \end{equation}
    then $\mathcal{N}_\eta:=\{\ket{\psi_j}_E:1\le j\le N\}$ is an $\eta$-net of the unit sphere in $\mathcal{H}_q$ with probability at least $1-\delta_{\mathrm{net}}$.
\end{Lem}

\begin{proof}[Proof of Lemma~\ref{lem:epsilon-net}]
    Identify $\mathcal{H}_q\simeq\mathbb{C}^{s_q}$ with $\mathbb{R}^{2s_q}$.
    Choose a maximal subset $\mathcal{N}_{\frac{\eta}{2}}$ of the unit sphere whose distinct points have Euclidean distance greater than $\frac{\eta}{2}$.
    By maximality, $\mathcal{N}_{\frac{\eta}{2}}$ is an $\frac{\eta}{2}$-net.
    The open Euclidean balls of radius $\frac{\eta}{4}$ centered at its points are disjoint and lie inside a ball of radius $1+\frac{\eta}{4}$, so a volume comparison gives
    \begin{equation}
        |\mathcal{N}_{\frac{\eta}{2}}|\le\left(1+\frac{4}{\eta}\right)^{2s_q}\le\left(\frac{6}{\eta}\right)^{2s_q}.
    \end{equation}
    For a fixed unit vector $\ket{v}_E\in\mathcal{H}_q$ and a Haar-random unit vector $\ket{\psi}_E\in\mathcal{H}_q$, set $X:=\operatorname{Re}\braket{v|\psi}$.
    Under the identification $\mathbb{C}^{s_q}\simeq\mathbb{R}^{2s_q}$, $\ket{\psi}_E$ is uniform on the real unit sphere: it can be generated by normalizing a vector whose real and imaginary coordinates are independent standard real Gaussian variables.
    Moreover, $X$ is the real Euclidean inner product of the corresponding unit vectors.
    By rotational invariance, $X$ therefore has the same distribution as the first coordinate of a uniform point on that sphere.
    The coordinate-density formula in \cite{Karol2000}, with real ambient dimension $n=2s_q$, gives
    \begin{equation}
    \label{eq:sphere-coordinate-density}
        p_{s_q}(x)=\frac{\Gamma(s_q)}{\sqrt{\pi}\,\Gamma(s_q-\tfrac{1}{2})}(1-x^2)^{s_q-\frac{3}{2}},\qquad -1<x<1.
    \end{equation}
    To relate this coordinate to an $\frac{\eta}{2}$-ball, use the normalization of both vectors to expand the squared distance:
    \begin{align}
    \label{eq:distance-coordinate}
        \|\ket{\psi}_E-\ket{v}_E\|_2^2
        &=\|\ket{\psi}_E\|_2^2+\|\ket{v}_E\|_2^2-\braket{\psi|v}-\braket{v|\psi}\nonumber\\
        &=2-2\operatorname{Re}\braket{v|\psi}=2-2X.
    \end{align}
    Consequently, the events are equivalent:
    \begin{equation}
    \label{eq:distance-coordinate-event}
        \|\ket{\psi}_E-\ket{v}_E\|_2\le\frac{\eta}{2}\quad\Longleftrightarrow\quad2-2X\le\frac{\eta^2}{4}\quad\Longleftrightarrow\quad X\ge1-\frac{\eta^2}{8}.
    \end{equation}
    By Eq.~\eqref{eq:distance-coordinate-event}, the $\frac{\eta}{2}$-ball around $\ket{v}_E$ on the unit sphere is the spherical cap defined by $X\ge1-\frac{\eta^2}{8}$.
    Its probability is therefore obtained by integrating the density $p_{s_q}$ of $X$ over the interval $[1-\frac{\eta^2}{8},1]$:
    \begin{equation}
    \label{eq:haar-cap-probability}
        \Pr\!\left[\|\ket{\psi}_E-\ket{v}_E\|_2\le\frac{\eta}{2}
        \right]=\Pr\!\left[X\ge1-\frac{\eta^2}{8}\right]=\int_{1-\frac{\eta^2}{8}}^{1}p_{s_q}(x)\,dx.
    \end{equation}
    This integral is the probability of the single cap centered at $\ket{v}_E$.
    Substituting Eq.~\eqref{eq:sphere-coordinate-density} and integrating yields the following lower bound:
    \begin{align}
    \label{eq:haar-cap-bound}
        \int_{1-\frac{\eta^2}{8}}^{1}p_{s_q}(x)\,dx
        &=\frac{\Gamma(s_q)}{\sqrt{\pi}\,\Gamma(s_q-\tfrac{1}{2})}\int_{1-\frac{\eta^2}{8}}^{1}(1-x^2)^{s_q-\frac{3}{2}}\,dx\nonumber\\
        &\ge\frac{\Gamma(s_q)}{\sqrt{\pi}\,\Gamma(s_q-\tfrac{1}{2})}\int_{1-\frac{\eta^2}{8}}^{1}(1-x)^{s_q-\frac{3}{2}}\,dx\nonumber\\
        &=\frac{\Gamma(s_q)}{\sqrt{\pi}\,\Gamma(s_q-\frac{1}{2})(s_q-\frac{1}{2})}\left(\frac{\eta^2}{8}\right)^{s_q-\frac{1}{2}}\nonumber\\
        &=\frac{\Gamma(s_q)}{\sqrt{\pi}\,\Gamma(s_q+\tfrac{1}{2})}\left(\frac{\eta^2}{8}\right)^{s_q-\frac{1}{2}}\nonumber\\
        &=\frac{4^{s_q}}{\pi s_q\binom{2s_q}{s_q}}\left(\frac{\eta^2}{8}\right)^{s_q-\frac{1}{2}}\nonumber\\
        &\ge\frac{1}{\pi s_q}\left(\frac{\eta^2}{8}\right)^{s_q-\frac{1}{2}}\nonumber\\
        &=\frac{8^{s_q-\frac{1}{2}}}{\pi s_q}\left(\frac{\eta}{8}\right)^{2s_q-1}\nonumber\\
        &\ge\left(\frac{\eta}{8}\right)^{2s_q-1}.
    \end{align}
    The first inequality uses $1-x^2\ge1-x$ on the integration interval and $s_q-\tfrac{3}{2}\ge0$.
    The integral contributes the factor $\frac{1}{s_q-\tfrac{1}{2}}$, which combines with $\Gamma(s_q-\tfrac{1}{2})$ through the recurrence $\Gamma(s_q+\tfrac{1}{2})=(s_q-\tfrac{1}{2})\Gamma(s_q-\tfrac{1}{2})$.
    For the next identity, we used
    \begin{equation}
        \Gamma(s_q)=(s_q-1)!,
        \qquad
        \Gamma(s_q+\tfrac{1}{2})=\frac{(2s_q)!\sqrt{\pi}}{4^{s_q} s_q!}.
    \end{equation}
    The remaining inequalities follow from $\binom{2s_q}{s_q}\le4^{s_q}$ and $8^{s_q-\frac{1}{2}}\ge\pi s_q$ for $s_q\ge2$.
    For each $\ket{v}_E\in\mathcal{N}_{\frac{\eta}{2}}$, the probability that all $N$ samples miss its $\frac{\eta}{2}$-ball is at most $\exp\!\left[-N\left(\frac{\eta}{8}\right)^{2s_q-1}\right]$.
    A union bound therefore bounds the probability of missing at least one such ball by
    \begin{equation}
        \left(\frac{6}{\eta}\right)^{2s_q}
        \exp\!\left[-N\left(\frac{\eta}{8}\right)^{2s_q-1}\right]
        \le\delta_{\mathrm{net}}.
    \end{equation}
    If no ball is missed, every unit vector lies within $\frac{\eta}{2}$ of a point of $\mathcal{N}_{\frac{\eta}{2}}$ and that point lies within $\frac{\eta}{2}$ of a sampled vector.
    The triangle inequality completes the proof.
\end{proof}

\begin{Lem}
    \label{lem:haar-lifted-evaluation}
    Let $R$ be a known positive integer satisfying $s_q\le R$.
    Let $\Omega_{q,L}$ be a joint state on $\mathcal{H}_E^{\otimes L}$ satisfying
    \begin{equation}
        \left\|\Omega_{q,L}-\left(\frac{\Pi_q}{s_q}\right)^{\otimes L}\right\|_1\le\varepsilon_{\mathrm{f}},
    \end{equation}
    where $\varepsilon_{\mathrm{f}},\gamma\in(0,1)$.
    Let $\mathcal{C}$ be a circuit making $M$ queries to a state-preparation unitary for $\ket{\psi}$.
    Then there exists a quantum circuit $\mathcal{C}'$ satisfying
    \begin{equation}
        \left\|\mathcal{C}'\left(\Omega_{q,L}\right)-\left[\mathbb{E}_{\ket{\psi}\sim\mathrm{Haar}(\mathcal{H}_q)}\mathcal{C}\left(W_\psi\right)\right]^{\otimes N}\right\|_1\le \varepsilon_{\mathrm{f}}+\gamma,
    \end{equation}
    using 
    \begin{equation}
        L=\bigO{\frac{RM^2N^2}{\gamma}}
    \end{equation}
    filtered-state registers in $\Omega_{q,L}$.
\end{Lem}
{\hypersetup{linkcolor=black,citecolor=black,urlcolor=black}
\begin{proof}
    We first analyze the construction on the ideal input $\rho_q^{\otimes L}$.
    Random purification~\cite{girardi2026randompurificationchannelsimple} produces a mixture of copies of a common purification of $\rho_q$, randomized on its reference register.
    Let $\mathcal{H}_D\simeq\mathbb{C}^R$ be the reference space with orthonormal basis $\{\ket{i}_D\}_{i=1}^R$.
    For one candidate, we use $b:=R(K-1)+1$ copies of $\rho_q$, where $K$ is the number of candidate copies needed for sample-to-query lifting.
    Since $\rank\rho_q=s_q\le R$, the random purification channel $\Lambda$ from Ref.~\cite[Lemma~2.11]{TangWrightZhandryConjugate} acts on the ideal input as
    \begin{equation}
        \Lambda(\rho_q^{\otimes b})=\int dV\left[(I\otimes V)\ket{\varphi}\!\bra{\varphi}(I\otimes V^\dagger)\right]^{\otimes b},
    \end{equation}
    where $\ket{\varphi}$ is a fixed purification of $\rho_q$ and $dV$ is the normalized Haar measure on unitaries acting on $\mathcal{H}_D$. The same $V$ acts on all $b$ reference registers.

    We measure each reference register in the basis $\{\ket{i}_D\}_{i=1}^R$ and retain $K$ environment registers corresponding to the most frequent outcome, breaking ties by the smallest index.
    The choice of $b$ guarantees by the pigeonhole principle that this outcome occurs at least $K$ times.
    To show that the retained candidate is Haar-random, first fix an outcome label $i_0$ and $K$ specified registers.
    Conditioning on outcome $i_0$ on those $K$ reference registers gives the unnormalized environment state
    \begin{equation}
        \int dV\left[(I\otimes\bra{i_0}_D)V\ket{\varphi}\!\bra{\varphi}V^\dagger(I\otimes\ket{i_0}_D)\right]^{\otimes K}.
    \end{equation}
    For the analysis, choose the purification
    \begin{equation}
        \ket{\varphi}=\frac{1}{\sqrt{s_q}}\sum_{i=1}^{s_q}\ket{e_i}\otimes\ket{i}_D,
    \end{equation}
    where $\{\ket{e_i}\}_{i=1}^{s_q}$ is an orthonormal basis of $\mathcal{H}_q$.
    For any unitary $Z$ acting on $\mathcal{H}_q$,
    \begin{equation}
        (Z\otimes I)\ket{\varphi}=\frac{1}{\sqrt{s_q}}\sum_i^{s_q}\ket{e_i}\otimes (Z^{\mathsf{T}}\oplus I)\ket{i}_D.
    \end{equation}
    Here $Z^{\mathsf{T}}$ is transposed in the basis $\{\ket{e_i}\}$, and the identity in the direct sum acts on the remaining $R-s_q$ reference dimensions.
    In the following formulas, $V$ acts only on $D$, and $dZ$ is Haar measure on unitaries acting on $\mathcal{H}_q$.
    Fix a unit vector $\ket{e_0}\in\mathcal{H}_q$ and choose a unitary $W_V$ on $\mathcal{H}_q$ satisfying
    \begin{equation}
        W_V\ket{e_0}
        =\frac{(I\otimes\bra{i_0}_D)V\ket{\varphi}}
        {\left\|(I\otimes\bra{i_0}_D)V\ket{\varphi}\right\|_2}.
    \end{equation}
    Zero-norm terms contribute zero and can be omitted.
    \begin{align}
        &\int dV\left[(I\otimes \bra{i_0}_D)V\ket{\varphi}\!\bra{\varphi}V^\dagger(I\otimes\ket{i_0}_D)\right]^{\otimes K}\\
        &=
        \int dV\int dZ\left[(I\otimes \bra{i_0}_D)V (I\otimes(Z^{\mathsf{T}}\oplus I))\ket{\varphi}\!\bra{\varphi}(I\otimes(Z^*\oplus I))V^\dagger(I\otimes\ket{i_0}_D)\right]^{\otimes K}\\
        &=
        \int dV\int dZ\left[Z(I\otimes \bra{i_0}_D)V\ket{\varphi}\!\bra{\varphi}V^\dagger(I\otimes\ket{i_0}_D)Z^\dagger\right]^{\otimes K}\\
        &=
        \int dV\left\|(I\otimes \bra{i_0}_D)V\ket{\varphi}\right\|^{2K}\int dZ\left[ZW_V\proj{e_0}W_V^\dagger Z^\dagger\right]^{\otimes K}\\
        &=
        \left(\int dV\left\|(I\otimes \bra{i_0}_D)V\ket{\varphi}\right\|^{2K}\right)\int dZ\left[Z\proj{e_0}Z^\dagger\right]^{\otimes K}.
    \end{align}
    After normalization, the last expression is the Haar $K$-copy average.
    The transformation $V\mapsto V(Z^{\mathsf{T}}\oplus I)$ also preserves the probability of every full measurement record and rotates the retained candidate by $Z$.
    Thus the retained environment registers have the same Haar $K$-copy average when we select the most frequent outcome.
    Applying this procedure independently to $N$ groups gives $N$ independent candidates and uses $L=Nb$ filtered-state registers.

    Let $\ket{\psi_0}:=Z\ket{e_0}$, with $Z$ Haar-random on $\mathcal{H}_q$.
    For each candidate, we use sample-to-query lifting to simulate a randomly chosen preparation unitary for $e^{i\theta}\ket{\psi_0}$. The phase $\theta$ is uniform on $[0,2\pi)$ and is sampled once and held fixed throughout that candidate's simulation.
    Ref.~\cite{TangWrightZhandryConjugate} shows that $K=O(\frac{M^2N}{\gamma})$ copies of $\ket{\psi_0}$ suffice to simulate $M$ queries to a state-preparation unitary with trace-norm output error at most $\frac{\gamma}{2N}$.
    Consequently, $L=Nb\le NRK=\bigO{\frac{RM^2N^2}{\gamma}}$.
    The Schur-transform construction of the random purification circuit uses the number of input copies and the register dimensions and can achieve any prescribed diamond-norm accuracy~\cite[Lemma~2.11]{TangWrightZhandryConjugate}.
    It requires no classical description of $\rho_q$, its eigenbasis, or its exact rank.
    Allocate error at most $\frac{\gamma}{2N}$ to each of its $N$ applications and at most $\frac{\gamma}{2N}$ to each subsequent lifting simulation.
    Contractivity and the triangle inequality bound their total output error by $\gamma$.
    The added purification gates require no queries to $U$ or $U^\dagger$, and this constant-factor allocation leaves the stated copy complexity unchanged.

    Trace-norm contractivity holds for any channel acting on the $L$ input registers. Applying it to the constructed circuit $\mathcal{C}'$ gives
    \begin{equation}
        \left\|\mathcal{C}'(\Omega_{q,L})-\mathcal{C}'(\rho_q^{\otimes L})\right\|_1
        \le\left\|\Omega_{q,L}-\rho_q^{\otimes L}\right\|_1
        \le\varepsilon_{\mathrm{f}}.
    \end{equation}
    Together with the simulation error $\gamma$ on the ideal input, this gives the claimed bound $\varepsilon_{\mathrm{f}}+\gamma$.
\end{proof}
}

\subsubsection{Evaluation of a fixed candidate.}
\label{subsubsec:fixed-candidate-evaluation}

\begin{Lem}[Estimation for a fixed environment state]
    \label{lem:QSVT}
    Let $\ket{\psi}_E$ be a unit vector in $\mathcal{H}_E$, and suppose query access to a state-preparation unitary $W_\psi$ and its inverse is available, with $W_\psi\ket{0}_E=\ket{\psi}_E$.
    For any $\xi,\beta\in(0,1)$, there is a subroutine $\texttt{\upshape Estimatef}(U,W_\psi;\xi,\beta)$ that returns $\widehat{f}_\psi\in[0,1]$ satisfying
    \begin{equation}
        \Pr\!\left[
            |\widehat{f}_\psi-f(\ket{\psi}_E)|\le\xi
        \right]\ge 1-\beta
    \end{equation}
    using $\bigO{\frac{1}{\xi^3}\log\!\left(\frac{2}{\beta}\right)}$ queries to $U$ and $U^\dagger$, together with $\bigO{\frac{1}{\xi^3}\log\!\left(\frac{2}{\beta}\right)}$ queries to $W_\psi$ and $W_\psi^\dagger$.
\end{Lem}

\begin{proof}[Proof of Lemma~\ref{lem:QSVT}]
We present the algorithm and its error analysis, followed by the query-complexity analysis.
We prepare the $d$-dimensional system register $A$ and an auxiliary register $A'$ with the same dimension as $A$.
Let $\ket{\Phi^+}_{AA'}$ denote a maximally entangled state on these registers.
Let $T$ be a known unitary, independent of $\psi$, satisfying $T\ket{0}_{AA'\mathrm{anc}}=\ket{\Phi^+}_{AA'}\ket{0}_{\mathrm{anc}}$.
    \begin{algorithm}[H]
        \caption{$\texttt{\upshape Estimatef}(U,W_\psi;\xi,\beta)$}
        \label{alg:estimate-f}
        \begin{algorithmic}[1]
            \Require Query access to $U,U^\dagger$ and $W_\psi,W_\psi^\dagger$ with $W_\psi\ket{0}_E=\ket{\psi}_E$; $\xi,\beta\in(0,1)$
            \Ensure An estimate $\widehat{f}_\psi\in[0,1]$ of $f(\ket{\psi}_E)$
            \State $U_\psi\gets(I_{\mathcal{H}}\otimes W_\psi^\dagger)U$
            \State Block-encode $A_\psi^\dagger A_\psi$ using $U_\psi^\dagger$ and $U_\psi$ with separate block-encoding ancillas
            \State Compute $P$ from Eq.~\eqref{eq:polynomial-construction} and construct its QSVT unitary $Q_\psi$
            \State $\Pi_{\mathrm{succ}}\gets\proj{\Phi^+}_{AA'}\otimes\proj{0}_{\mathrm{anc}}$
            \State Use amplitude estimation with $Q_\psi T$, $T^\dagger Q_\psi^\dagger$, and $\Pi_{\mathrm{succ}}$ to obtain $\widehat{p}_\psi$ with additive error $\frac{\xi}{8}$ and failure probability at most $\beta$
            \State \Return $\min\{1,\max\{0,4\widehat{p}_\psi\}\}$
        \end{algorithmic}
    \end{algorithm}
    We establish the guarantees of Algorithm~\ref{alg:estimate-f}.
    Since $\|A_\psi\|_\infty\le1$, the spectrum of $A_\psi^\dagger A_\psi$ lies in $[0,1]$, and
    \begin{equation}
    \label{eq:trace-norm-target}
        f(\ket{\psi}_E)=\left(\frac{\Tr\sqrt{A_\psi^\dagger A_\psi}}{d}\right)^2\in[0,1].
    \end{equation}
    The unitary $U_\psi$ and its adjoint block-encode $A_\psi$ and $A_\psi^\dagger$, respectively:
    \begin{align}
        U_\psi&:=(I_{\mathcal{H}}\otimes W_\psi^\dagger)U,
        \label{eq:Upsi-def}\\
        (I_{\mathcal{H}}\otimes\bra{0}_E)U_\psi(I_{\mathcal{H}}\otimes\ket{0}_E)&=A_\psi,
        \label{eq:block-enc-A}\\
        (I_{\mathcal{H}}\otimes\bra{0}_E)U_\psi^\dagger(I_{\mathcal{H}}\otimes\ket{0}_E)&=A_\psi^\dagger.
        \label{eq:block-enc-Adag}
    \end{align}
    The product construction~\cite{Gilyen2019SVD} therefore gives an exact block-encoding of $A_\psi^\dagger A_\psi$ using two environment registers and one call each to $U_\psi$ and $U_\psi^\dagger$.

    \paragraph*{Step 1. Polynomial approximation and QSVT.}
    Since $|\sqrt{|x|}-\sqrt{|y|}|\le\sqrt{|x-y|}$, Jackson's theorem~\cite{Jackson1911} gives a real polynomial $h$ of degree $\bigO{\frac{1}{\xi^2}}$ satisfying
    \begin{equation}
    \label{eq:jackson-polynomial}
        \max_{x\in[-1,1]}\left|h(x)-\sqrt{|x|}\right|\le\frac{\xi}{32}.
    \end{equation}
    The Chebyshev coefficients of $h$ can be computed classically to the required precision in time polynomial in $\frac{1}{\xi}$.
    Taking its even part and rescaling gives the real even polynomial
    \begin{equation}
    \label{eq:polynomial-construction}
        P(x):=\frac{h(x)+h(-x)}{4\left(1+\frac{\xi}{32}\right)},
    \end{equation}
    which has degree $\bigO{\frac{1}{\xi^2}}$ and satisfies, for all $x\in[-1,1]$,
    \begin{equation}
    \label{eq:poly-approx}
        |P(x)|\le\frac{1}{2},
        \qquad
        \left|P(x)-\frac{1}{2}\sqrt{|x|}\right|\le\frac{\frac{\xi}{32}}{1+\frac{\xi}{32}}\le\frac{\xi}{32}.
    \end{equation}
    Because $A_\psi^\dagger A_\psi$ is positive semidefinite, its singular value transformation by $P$ is $P(A_\psi^\dagger A_\psi)$.
    The real-polynomial QSVT construction~\cite{Gilyen2019SVD} yields a unitary $Q_\psi$ whose encoded block $B_\psi$, with all ancillas in the reference state $\ket{0}_{\mathrm{anc}}$, satisfies
    \begin{equation}
    \label{eq:op-norm-bound}
        \left\|B_\psi-\frac{1}{2}\sqrt{A_\psi^\dagger A_\psi}\right\|_\infty\le\frac{\xi}{16}.
    \end{equation}
    Here, the polynomial approximation contributes at most $\frac{\xi}{32}$, and we implement the known phase rotations with total operator-norm error at most $\frac{\xi}{32}$.
    This construction of an even polynomial uses the same oracle-call sequence in both control branches~\cite[Corollary~18]{Gilyen2019SVD}, controlling only known phase rotations, so it requires no controlled access to the unknown dilation unitary~\cite[Lemma~19]{Gilyen2019SVD}.
    The Hadamard gates preparing and unpreparing the additional control qubit are included in $Q_\psi$.

    \paragraph*{Step 2. Estimating the scalar output.}
    Applying $Q_\psi$ to $\ket{\Phi^+}_{AA'}\ket{0}_{\mathrm{anc}}$ and testing the known projector $\Pi_{\mathrm{succ}}=\proj{\Phi^+}_{AA'}\otimes\proj{0}_{\mathrm{anc}}$ gives success probability
    \begin{equation}
    \label{eq:trace-success-probability}
        p_\psi=\left|\bra{\Phi^+}_{AA'}(B_\psi\otimes I_{A'})\ket{\Phi^+}_{AA'}\right|^2
        =\left|\frac{\Tr B_\psi}{d}\right|^2.
    \end{equation}
    By Eq.~\eqref{eq:op-norm-bound},
    \begin{equation}
    \label{eq:trace-approx-bound}
        \left|\frac{\Tr B_\psi}{d}-\frac{\|A_\psi\|_1}{2d}\right|
        \le\left\|B_\psi-\frac{1}{2}\sqrt{A_\psi^\dagger A_\psi}\right\|_\infty
        \le\frac{\xi}{16}.
    \end{equation}
    Since $B_\psi$ is a block of a unitary, $\frac{|\Tr B_\psi|}{d}\le1$, while $\frac{\|A_\psi\|_1}{2d}\le\frac{1}{2}$.
    Therefore,
    \begin{equation}
    \label{eq:probability-bias}
        |4p_\psi-f(\ket{\psi}_E)|
        \le4\cdot\frac{\xi}{16}\left(1+\frac{1}{2}\right)
        =\frac{3\xi}{8}.
    \end{equation}
    Amplitude estimation~\cite{Brassard2000}, followed by median amplification, estimates $p_\psi$ to additive error $\frac{\xi}{8}$ with failure probability at most $\beta$.
    This needs no controlled query to the unknown dilation: for the full circuit $Q_\psi T$ defined in Algorithm~\ref{alg:estimate-f} and any known reflection $\mathcal{R}$,
    \begin{equation}
        \texttt{ctrl-}(Q_\psi T\mathcal{R}T^\dagger Q_\psi^\dagger)
        =(I_2\otimes Q_\psi T)\,\texttt{ctrl-}\mathcal{R}\,(I_2\otimes T^\dagger Q_\psi^\dagger),
    \end{equation}
    where $\texttt{ctrl-}$ denotes the controlled version of an operator.
    Thus, the controlled reflection about the prepared state uses unconditional calls to $Q_\psi$ and $Q_\psi^\dagger$, together with the known operations $T$ and $T^\dagger$, and the controlled reflection for $\Pi_{\mathrm{succ}}$ uses only known operations.
    On the success event of amplitude estimation,
    \begin{equation}
    \label{eq:subroutine-error}
        |4\widehat{p}_\psi-f(\ket{\psi}_E)|
        \le4\frac{\xi}{8}+\frac{3\xi}{8}
        =\frac{7\xi}{8}\le\xi.
    \end{equation}
    Since $f(\ket{\psi}_E)\in[0,1]$, replacing $4\widehat{p}_\psi$ by $\min\{1,\max\{0,4\widehat{p}_\psi\}\}$ cannot increase the error.
    This proves the required accuracy and success probability.

    \paragraph*{Query complexity.}
    The block-encoding product uses a constant number of queries to $U,U^\dagger$ and $W_\psi,W_\psi^\dagger$.
    The degree of $P$ is $\bigO{\frac{1}{\xi^2}}$, so each use of $Q_\psi$ or its inverse costs $\bigO{\frac{1}{\xi^2}}$ such queries.
    Amplitude estimation and median amplification use
    \begin{equation}
    \label{eq:trace-est-cost}
        \bigO{\frac{1}{\xi}\log\!\left(\frac{2}{\beta}\right)}
    \end{equation}
    applications of $Q_\psi$ and $Q_\psi^\dagger$ in total~\cite{Brassard2000}, giving the total cost
    \begin{equation}
    \label{eq:subroutine-query-cost}
        \bigO{\frac{1}{\xi^2}}\,\bigO{\frac{1}{\xi}\log\!\left(\frac{2}{\beta}\right)}
        =\bigO{\frac{1}{\xi^3}\log\!\left(\frac{2}{\beta}\right)}.
    \end{equation}
    All remaining classical computations and known operations require no additional queries to the unknown dilation unitary.
\end{proof}

\section{Query lower bound with forward Stinespring dilation unitary access.}
\label{sec:lower bound}
Next, we establish a query lower bound for estimating OAGF with forward Stinespring dilation unitary access. For sufficiently small constant additive error, we show that $\Omega(\sqrt{d})$ forward queries to the supplied dilation unitary are necessary.
For an $n$-qubit system, this is exponential in $n$.
\begin{Thm}
    Let $d=4l$ for a positive integer $l$.
    Fix a constant $0<\epsilon<\frac{\sqrt{2}-1}{10}$.
    Any quantum algorithm that, given forward query access to a supplied Stinespring dilation unitary $U$ of an unknown quantum channel $\mathcal{E}$ acting on a $d$-dimensional system, estimates $F_{\mathrm{OAGF}}(\mathcal{E})$ to additive error at most $\epsilon$ with success probability at least $\frac{2}{3}$ for every such $\mathcal{E}$ and every compatible supplied $U$ requires $\bigOmega{\sqrt{d}}$ forward queries to $U$ in the worst case.
    The lower bound holds even when $\mathcal{E}$ is promised to have Kraus rank two.
    \label{Thm:lower bound}
\end{Thm}
This lower bound also implies hardness of implementing a unitary channel that maximizes the average gate fidelity with the given channel.
\begin{Cor}
    \label{cor:optimal-unitary-implementation}
    Consider the task of implementing a quantum channel $\mathcal{V}$ such that
    \begin{equation}
        \norm{U(\cdot)U^\dagger-\mathcal{V}}_\diamond\le\epsilon,\qquad
        U\in\argmax_{W\in\mathbb{U}(d)}F_{\mathrm{avg}}(\mathcal{E},W(\cdot)W^\dagger)
    \end{equation}
    given access to $\mathcal{E}$.
    If we choose a dimension-independent constant $\epsilon<0.04$, then we require $\bigOmega{\sqrt{d}}$ uses of quantum channel $\mathcal{E}$ to achieve this task.
\end{Cor}
\begin{proof}
    If we can implement the quantum channel $\mathcal{V}$, 
    \begin{align}
        \abs{\Tr(\mathsf{E}\mathsf{U})-\Tr(\mathsf{E}\mathsf{V})}
        &\le\norm{\mathsf{E}\mathsf{U}-\mathsf{E}\mathsf{V}}_1
        \\ &\le\norm{\mathsf{E}}_\infty\norm{\mathsf{U}-\mathsf{V}}_1
        \\ &\le\epsilon.        
    \end{align}
Using the SWAP test, we can estimate $\Tr(\mathsf{E}\mathsf{V})$ with additive error $\epsilon'$ with $\bigO{\frac{1}{\epsilon'^2}}$ queries, which does not depend on  dimension. Thus, if we convert $N$ queries to $\mathcal{E}$ to one query to $\mathcal{V}$, we can estimate $F_{\mathrm{OAGF}}(\mathcal{E})$ with $\epsilon+\epsilon'$ error.
Taking $\epsilon'=\frac{\sqrt{2}-1.4}{10}$ as a dimension-independent constant and using $\epsilon<0.04$, we obtain
\begin{equation}
    N\times\bigO{\frac{1}{\epsilon'^2}}=\bigOmega{\sqrt{d}},
\end{equation}
which gives $N=\bigOmega{\sqrt{d}}$.
\end{proof}
\subsection{Hard instance.}
We prove the query lower bound in Theorem~\ref{Thm:lower bound} by reducing a channel-discrimination problem to OAGF estimation, defined below.
Let $d = 4l$, where $l$ is a positive integer, and identify the system with the basis $\{\ket{\alpha,u}:\alpha\in[4], u\in[l]\}$.
Consider the channels
\begin{align}
    \mathcal{E}_b(\rho)&=\frac{1}{2}\rho+\frac{1}{2}K^{\,(b)}\rho K^{(b)\dagger},
\end{align}
where $b\in\{Y,N\}$ and
\begin{align}
    K^{\,(b)}=&\mathrm{diag}(\lambda_b)\otimes I_l,\\
    \lambda_Y=&(1,1,1,-1),\quad\lambda_N=(1,1,i,-i).
\end{align}
Applying $\mathcal{E}_b$ on $\ket{\alpha,u}\!\bra{\beta,v}$ gives
\begin{equation}
    \mathcal{E}_b(\ket{\alpha,u}\!\bra{\beta,v})=\frac{1+\lambda_{b,\alpha}\lambda_{b,\beta}^*}{2}\ket{\alpha,u}\!\bra{\beta,v}.
\end{equation}
Our aim is to discriminate $\mathcal{E}_Y$ and $\mathcal{E}_N$.
Define a Stinespring dilation unitary $U_b$ of $\mathcal{E}_b$ as follows:
\begin{equation}
    U_b=\sum_{\alpha=1}^{4}\sum_{u=1}^{l}\proj{\alpha,u}_S\otimes Q_{b,\alpha}, \qquad
    Q_{b,\alpha}:=
    \begin{pmatrix}
        1 & 0\\
        0 & \lambda_{b,\alpha}
    \end{pmatrix}H.
    \label{eq:U_b}
\end{equation}
Indeed, after tracing out the environment system, we obtain 
\begin{equation}
    \Tr_E \left[U_b(\ket{\alpha,u}\!\bra{\beta,v}\otimes \proj{0}_E)U_b^\dagger\right]=\ket{\alpha,u}\!\bra{\beta,v}\Tr_E\frac{\proj{0}+\lambda_{b,\alpha}\lambda_{b,\beta}^*\proj{1}}{2}
    =\mathcal{E}_b(\ket{\alpha,u}\!\bra{\beta,v}).
\end{equation}
By simple calculation, for $\alpha=1,2,3,4$, we have
\begin{equation}
    Q_{N,\alpha}=\sum_{\beta=1}^4\Gamma_{\alpha,\beta}Q_{Y,\beta}, \qquad
    \Gamma=
    \begin{pmatrix}
        1 & 0 & 0 & 0\\
        0 & 1 & 0 & 0\\    
        0 & 0 & \frac{1+i}{2} & \frac{1-i}{2}\\
        0 & 0 & \frac{1-i}{2} & \frac{1+i}{2}. 
    \end{pmatrix}
    \label{Q_Y and Q_N}
\end{equation}
where $\Gamma$ is a unitary matrix.
\begin{Prop}
    The channels $\mathcal{E}_b$, $b\in\{Y,N\}$, satisfy the following properties.
    \begin{equation}
        F_{\mathrm{OEF}}(\mathcal{E}_Y)=\frac{5}{8}=0.625, \quad F_{\mathrm{OEF}}(\mathcal{E}_N)=\frac{3+2\sqrt{2}}{8}\simeq0.728\cdots.
    \end{equation}
    \begin{equation}
        \left|F_{\mathrm{OAGF}}(\mathcal{E}_Y)-F_{\mathrm{OAGF}}(\mathcal{E}_N)\right|\ge\frac{\sqrt{2}-1}{5}.
    \end{equation}
  \label{prop:prop of channels}
\end{Prop}
\begin{proof}
    The channel $\mathcal{E}_b$ has Kraus operators $\frac{I_d}{\sqrt{2}}$ and $\frac{K^{(b)}}{\sqrt{2}}$.
    Applying the variational formula derived in the proof of Lemma~\ref{lem:oef-variational}, we obtain
    \begin{align}
        F_{\mathrm{OEF}}(\mathcal{E}_Y)
        &=\frac{1}{d^2}\max_{\norm{c}=1}\norm{\frac{c_1I_d+c_2K^{(Y)}}{\sqrt{2}}}^2_1\nonumber
        \\ &=\frac{1}{d^2}\max_{\norm{c}=1}\Biggl\|\sum_u^l\Biggl[\left(\frac{1}{\sqrt{2}}c_1+\frac{1}{\sqrt{2}}c_2\right)\left(\proj{1,u}+\proj{2,u}+\proj{3,u}\right)\nonumber
        \\ &\qquad\qquad\qquad\qquad\qquad\qquad\qquad\qquad\qquad\qquad +\left(\frac{1}{\sqrt{2}}c_1-\frac{1}{\sqrt{2}}c_2\right)\proj{4,u}\Biggr]\Biggr\|_1^2
        \\ &=\left\{\max_{\norm{c}=1}\frac{1}{4}\left(3\abs{\frac{1}{\sqrt{2}}c_1+\frac{1}{\sqrt{2}}c_2}+\abs{\frac{1}{\sqrt{2}}c_1-\frac{1}{\sqrt{2}}c_2}\right)\right\}^2
        \\ &=\left\{\max_{\norm{c}=1}\frac{1}{4\sqrt{2}}\left(3\abs{c_1+c_2}+\abs{c_1-c_2}\right)\right\}^2.
    \end{align}
    \begin{align}
        F_{\mathrm{OEF}}(\mathcal{E}_N)
        &=\frac{1}{d^2}\max_{\norm{c}=1}\norm{\frac{c_1I_d+c_2K^{(N)}}{\sqrt{2}}}^2_1\nonumber\\
        &=\frac{1}{d^2}\max_{\norm{c}=1}
        \Biggl\|\sum_u^l\Bigl[
        \left(\frac{1}{\sqrt{2}}c_1+\frac{1}{\sqrt{2}}c_2\right)
        \left(\proj{1,u}+\proj{2,u}\right)\nonumber\\
        &\qquad\qquad\qquad\qquad\qquad\qquad\qquad\qquad
        +\left(\frac{1}{\sqrt{2}}c_1+\frac{i}{\sqrt{2}}c_2\right)\proj{3,u}+\left(\frac{1}{\sqrt{2}}c_1-\frac{i}{\sqrt{2}}c_2\right)\proj{4,u}\Bigr]\Biggr\|_1^2
        \\ &=\left\{\max_{\norm{c}=1}\frac{1}{4}\left(2\abs{\frac{1}{\sqrt{2}}c_1+\frac{1}{\sqrt{2}}c_2}+\abs{\frac{1}{\sqrt{2}}c_1+\frac{i}{\sqrt{2}}c_2}+\abs{\frac{1}{\sqrt{2}}c_1-\frac{i}{\sqrt{2}}c_2}\right)\right\}^2
        \\ &=\left\{\max_{\norm{c}=1}\frac{1}{4\sqrt{2}}\left(2\abs{c_1+c_2}+\abs{c_1+ic_2}+\abs{c_1-ic_2}\right)\right\}^2.
    \end{align}
    Let $\alpha,\beta,\gamma\in\mathbb{C}$ satisfy $1+\Bar{\beta}\gamma=0$ and let $a_1,a_2,a_3\ge0$. We consider the maximization problem
    $M=\max_{\abs{c_1}^2+\abs{c_2}^2=1}a_1\abs{c_1+\alpha c_2}+a_2\abs{c_1+\beta c_2}+a_3\abs{c_1+\gamma c_2}$.
    For any $z_1,z_2,z_3\in \mathbb{C}$, we find that
    \begin{equation}
        \abs{z_1}+\abs{z_2}+\abs{z_3}=\max_{\abs{u_1}=\abs{u_2}=\abs{u_3}=1}\abs{u_1z_1+u_2z_2+u_3z_3}.
    \end{equation}
    Then, 
    \begin{align}
        M
        &=\max_{\substack{\abs{u_1}=\abs{u_2}=\abs{u_3}=1\\ \abs{c_1}^2+\abs{c_2}^2=1}}\abs{a_1u_1(c_1+\alpha c_2)+a_2u_2(c_1+\beta c_2)+a_3u_3(c_1+\gamma c_2)}
        \\ &=\max_{\substack{\abs{u_1}=\abs{u_2}=\abs{u_3}=1\\ \abs{c_1}^2+\abs{c_2}^2=1}}\abs{(a_1u_1+a_2u_2+a_3u_3)c_1+(a_1\alpha u_1+a_2\beta u_2+a_3\gamma u_3)c_2}
        \\ &=\max_{\substack{\abs{u_1}=\abs{u_2}=\abs{u_3}=1}}\sqrt{\abs{a_1u_1+a_2u_2+a_3u_3}^2+\abs{a_1\alpha u_1+a_2\beta u_2+a_3\gamma u_3}^2}\qquad(\text{Cauchy-Schwarz inequality}).
    \end{align}
    Only relative phases between $u_1$ and $u_2,u_3$ contribute to this quantity, so let $u_1=1$, $u_2=e^{i\theta}$ and $u_3=e^{i\phi}$.
    \begin{align}
        M^2
        &=\max_{\theta,\phi}\left(\abs{a_1+a_2e^{i\theta}+a_3e^{i\phi}}^2+\abs{a_1\alpha+a_2\beta e^{i\theta}+a_3\gamma e^{i\phi}}^2\right)
        \\ &=\max_{\theta,\phi}\left\{a_1^2\left(1+\abs{\alpha}^2\right)+a_2^2\left(1+\abs{\beta}^2\right)+a_3^2\left(1+\abs{\gamma}^2\right)+2a_1a_2\mathrm{Re}\left[(1+\Bar{\alpha}\beta)e^{i\theta}\right]+2a_1a_3\mathrm{Re}\left[(1+\Bar{\alpha}\gamma)e^{i\phi}\right]\right\}
        \\ &=a_1^2\left(1+\abs{\alpha}^2\right)+a_2^2\left(1+\abs{\beta}^2\right)+a_3^2\left(1+\abs{\gamma}^2\right)+2a_1a_2\abs{1+\Bar{\alpha}\beta}+2a_1a_3\abs{1+\Bar{\alpha}\gamma}.
    \end{align}
    Therefore,
    \begin{equation}
        F_{\mathrm{OEF}}(\mathcal{E}_Y)=\frac{1}{32}\left(9\left(1+\abs{1}^2\right)+1+\abs{-1}^2\right)=\frac{5}{8},
    \end{equation}
    \begin{equation}
        F_{\mathrm{OEF}}(\mathcal{E}_N)=\frac{1}{32}\left(4\left(1+\abs{1}^2\right)+1+\abs{i}^2+1+\abs{-i}^2+4\abs{1+i}+4\abs{1-i}\right)=\frac{3+2\sqrt{2}}{8}.
    \end{equation}
    Since $\left|F_{\mathrm{OAGF}}(\mathcal{E}_Y)-F_{\mathrm{OAGF}}(\mathcal{E}_N)\right|=\frac{d}{d+1}\left|F_{\mathrm{OEF}}(\mathcal{E}_Y)-F_{\mathrm{OEF}}(\mathcal{E}_N)\right|$ and $\frac{d}{d+1}\ge\frac{4}{5}$ for $d\ge 4$,
    \begin{equation}
        \left|F_{\mathrm{OAGF}}(\mathcal{E}_Y)-F_{\mathrm{OAGF}}(\mathcal{E}_N)\right|\ge\frac{4}{5}\left|F_{\mathrm{OEF}}(\mathcal{E}_Y)-F_{\mathrm{OEF}}(\mathcal{E}_N)\right|=\frac{\sqrt{2}-1}{5}.
    \end{equation}
\end{proof}
\begin{Lem}
    For any $V_1,V_2\in \mathbb{U}(d)$ and any quantum channel $\mathcal{E}$,
    \begin{equation}
        F_{\mathrm{OEF}}(\mathcal{V}_1\circ\mathcal{E}\circ\mathcal{V}_2)=F_{\mathrm{OEF}}(\mathcal{E}),
    \end{equation}
    where $\mathcal{V}_1(\cdot)=V_1(\cdot) V_1^\dagger$ and $\mathcal{V}_2(\cdot)=V_2(\cdot) V_2^\dagger$.
    \label{lem:unitary equivalent}
\end{Lem}
\begin{proof}[Proof of Lemma~\ref{lem:unitary equivalent}]
    \begin{align}
        F_{\mathrm{OEF}}(\mathcal{V}_1\circ\mathcal{E}\circ\mathcal{V}_2)
        &=\frac{1}{d^2}\max_{W\in \mathbb{U}(d)}\sum_s\left|\Tr\left[W^\dagger V_1K_sV_2\right]\right|^2\nonumber
        \\ &=\frac{1}{d^2}\max_{W\in \mathbb{U}(d)}\sum_s\left|\Tr\left[(V_1^\dagger WV_2^\dagger)^\dagger K_s\right]\right|^2\nonumber
        \\ &=\frac{1}{d^2}\max_{W'\in \mathbb{U}(d)}\sum_s\left|\Tr\left[W'^\dagger K_s\right]\right|^2\nonumber
        \\ &=F_{\mathrm{OEF}}(\mathcal{E}).
    \end{align}
\end{proof}
The proof strategy is as follows. Let $\epsilon_{\mathrm{OEF}}$ denote the difference between $F_{\mathrm{OEF}}(\mathcal{E}_Y)$ and $F_{\mathrm{OEF}}(\mathcal{E}_N)$.
If an algorithm estimates $F_{\mathrm{OEF}}(\mathcal{E})$ to additive error strictly smaller than $\frac{\epsilon_{\mathrm{OEF}}}{2}$ using $t$ queries to the channel, it can distinguish $\mathcal{E}_Y$ from $\mathcal{E}_N$ with $t$ queries.
Write $\mathcal{E}_b^{V_1,V_2}:=\mathcal{V}_1\circ\mathcal{E}_b\circ\mathcal{V}_2$ for $V_1,V_2\in\mathbb{U}(d)$.
By Lemma~\ref{lem:unitary equivalent}, the same algorithm can distinguish $\mathcal{E}_Y^{V_1,V_2}$ from $\mathcal{E}_N^{V_1',V_2'}$ for any $V_1,V_2,V_1',V_2'\in\mathbb{U}(d)$.
Hence, the output states $\rho_{\mathcal{E}_Y^{V_1,V_2}}$ and $\rho_{\mathcal{E}_N^{V_1',V_2'}}$ obtained after $t$ channel queries must be distinguishable for all such unitaries.
The same holds for $\mathbb{E}_{V_1,V_2\sim\mathbb{U}(d)}[\rho_{\mathcal{E}_Y^{V_1,V_2}}]$ and $\mathbb{E}_{V_1,V_2\sim\mathbb{U}(d)}[\rho_{\mathcal{E}_N^{V_1,V_2}}]$, which correspond to the averaged case.
We show that this discrimination task requires $\bigOmega{\sqrt{d}}$ uses of channels.
\subsection{Review of path-recording oracle methods.}
\label{subsec:path-recording-methods}
We recall the path-recording framework of Ref.~\cite{Ma2024}, specialized to the two oracles used in our lower-bound proof.
We first define the standard oracle, then its restricted version for our hard instances, and finally state the bounds for replacing Haar-random unitaries by these oracles.

\paragraph*{Standard path-recording oracle.}
A relation of length $k$ is an unordered multiset $R_k=\{(x_i,y_i)\}_{i=1}^k$ with $x_i,y_i\in[d]$.
We use injective relations, for which the output labels $y_i$ are pairwise distinct; the input labels $x_i$ may repeat.
Set $R_0=\varnothing$, write $\mathrm{Im}(R_k)=\{y_1,\ldots,y_k\}$, and associate each relation with an orthonormal basis state $\ket{R_k}_R$ in a relation register $R$.
For $k<d$, the standard path-recording oracle $P$ is the linear map
\begin{equation}
    P:\ket{x}_S\ket{R_k}_R
    \mapsto\frac{1}{\sqrt{d-k}}
    \sum_{y\in[d]\setminus\mathrm{Im}(R_k)}
    \ket{y}_S\ket{R_k\cup\{(x,y)\}}_R.
\end{equation}
Thus, each query coherently records an input-output pair while excluding previously recorded output labels.
The relation register is initialized in $\ket{\varnothing}_R$ and is inaccessible to the algorithm.
The first $t$ applications of $P$ can be implemented to error $O\!\left(\frac{1}{d}\right)$ using $\poly(t,\log d)$ gates~\cite[Appendix~A.2]{Ma2024}.

\paragraph*{Restricted path-recording oracle.}
For our hard instances, $d=4l$, and an output label is written as $y=(\alpha,u)\in[4]\times[l]$.
Let $\mathcal{S}_k^\#$ be the family of relations
\begin{equation}
    R_k=\{(x_i,(\alpha_i,u_i))\}_{i=1}^k,
    \qquad x_i\in[d],\quad\alpha_i\in[4],\quad u_i\in[l],
\end{equation}
whose labels $u_1,\ldots,u_k$ are pairwise distinct.
Set $\mathcal{S}_0^\#=\{\varnothing\}$ and $\mathcal{S}^\#=\bigcup_{k=0}^{l}\mathcal{S}_k^\#$.
For any $R_k\in\mathcal{S}_k^\#$ with $k<l$, there are four choices of $\alpha$ and $l-k$ unused values of $u$, so the number of admissible outputs is
\begin{equation}
    \mathcal{Z}_k=4(l-k)=d-4k.
    \label{eq:z}
\end{equation}
The restricted oracle $P^\#$ is therefore defined by
\begin{equation}
    P^\#:\ket{x}_S\ket{R_k}_R
    \mapsto\frac{1}{\sqrt{d-4k}}
    \sum_{\substack{\alpha\in[4],\,u\in[l]\\u\notin\{u_1,\ldots,u_k\}}}
    \ket{\alpha,u}_S\ket{R_k\cup\{(x,(\alpha,u))\}}_R.
\end{equation}
This restriction ensures that the label $u$ is different at every query, as required in the proof below.

\paragraph*{Haar replacement bounds.}
Consider an algorithm with $t$ forward queries to a fixed unitary $\mathcal{O}$ on $S$ and arbitrary intermediate unitaries $A_0,\ldots,A_t$ on $SS'$, independent of $\mathcal{O}$.
Write its final state as
\begin{equation}
    \ket{\mathcal{A}_t^{\mathcal{O}}}_{SS'}
    :=A_t\mathcal{O}_S A_{t-1}\cdots A_1\mathcal{O}_S A_0\ket{0}_{SS'}.
\end{equation}
Let $\ket{\mathcal{A}_t^{P}}_{SS'R}$ and $\ket{\mathcal{A}_t^{P^\#}}_{SS'R}$ denote the states obtained by replacing every query with $P$ and $P^\#$, respectively, and initializing $R$ in $\ket{\varnothing}_R$.
The operations $A_i$ do not act on $R$, which is traced out only when comparing the final output states.
The auxiliary register $S'$ may include the retained environment $E$, so these bounds also apply when the algorithm manipulates $E$ between queries.

\begin{Lem}[Standard Haar replacement {\cite[Theorem~5]{Ma2024}}]
    For $0\le t\le d$,
    \begin{equation}
        \norm{\mathbb{E}_{\mathcal{O}\sim\mu_{\mathrm{Haar}}}
        \proj{\mathcal{A}_t^{\mathcal{O}}}_{SS'}
        -\Tr_R\left(\proj{\mathcal{A}_t^{P}}_{SS'R}\right)}_1
        \le\frac{2t(t-1)}{d+1}.
    \end{equation}
    \label{Lem:Haar P}
\end{Lem}
\begin{Lem}[Restricted Haar replacement]
    For $0\le t<\frac{d}{4}$,
    \begin{equation}
        \norm{\mathbb{E}_{\mathcal{O}\sim\mu_{\mathrm{Haar}}}
        \proj{\mathcal{A}_t^{\mathcal{O}}}_{SS'}
        -\Tr_R\left(\proj{\mathcal{A}_t^{P^\#}}_{SS'R}\right)}_1
        \le\frac{2t(t-1)}{d+1}
        +2\left(1-\prod_{k=0}^{t-1}\frac{d-4k}{d-k}\right)
        \le\frac{6t(t-1)}{d}.
    \end{equation}
    \label{Lem:Haar P hash}
\end{Lem}
\begin{proof}[Proof of Lemma~\ref{Lem:Haar P hash}]
    The family $\mathcal{S}^\#$ allows arbitrary input labels $x_i$, and deleting pairs preserves the distinctness of the $u_i$ labels.
    It therefore satisfies the consistency condition of Ref.~\cite{Ma2024}.
    Moreover, the number $\mathcal{Z}_k=d-4k$ of admissible extensions depends only on the relation length, so the uniform growth condition also holds.
    Theorem~9 of Ref.~\cite{Ma2024} then gives the first bound, with
    $\frac{(d-t)!}{d!}\prod_{k=0}^{t-1}\mathcal{Z}_k=\prod_{k=0}^{t-1}\frac{d-4k}{d-k}$.
    For $t<\frac{d}{4}$, each factor in this product lies in $[0,1]$, and
    \begin{equation}
        1-\prod_{k=0}^{t-1}\frac{d-4k}{d-k}
        \le\sum_{k=0}^{t-1}\frac{3k}{d-k}
        \le\sum_{k=0}^{t-1}\frac{4k}{d}
        =\frac{2t(t-1)}{d},
    \end{equation}
    where $d-k\ge\frac{3d}{4}$.
    Combining this estimate with $\frac{2t(t-1)}{d+1}\le\frac{2t(t-1)}{d}$ gives the final bound.
\end{proof}
\subsection{Proof of Theorem~\ref{Thm:lower bound}.}
\label{subsection:showing lower bound}
The core of our strategy to prove the query lower bound is the adversary method~\cite{Bennett1997}.
Let $b\in\{Y,N\}$. We apply independently Haar-random unitaries $V_2$ and $V_1$ on $S$ before and after the dilation unitary $U_b$, respectively, and define $\rho_{\mathrm{Haar}}^{(b)}(t)$ as the averaged output state of a $t$-query algorithm (see Figure~\ref{fig:rho^Haar}):
\begin{equation}
\rho_{\mathrm{Haar}}^{(b)}(t):=\mathbb{E}_{V_1,V_2\sim\mu_{\mathrm{Haar}}}\left(\proj{\mathcal{A}_t^{V_1,U_b,V_2}}_{SS'E}\right).
\label{eq:Haar}
\end{equation}

Here, $t$ denotes the number of forward queries to the dilation unitary $(V_1\otimes I_E)U_b(V_2\otimes I_E)$,
and $A_0,\ldots,A_t$ are arbitrary quantum operations independent of $b$ and of the oracles $V_1,V_2,U_b$, representing the intermediate operations performed by the algorithm.
The environment $E$ is retained as part of the output system $SS'E$, and the intermediate operations may act jointly on $SS'E$.

Our goal is to upper bound

\begin{equation}
    \|\rho_{\mathrm{Haar}}^{(Y)}(t)-\rho_{\mathrm{Haar}}^{(N)}(t)\|_1
\end{equation}
and thereby analyze how the distinguishability between the two cases $b=Y$ and $b=N$ grows with the number of queries $t$.

To analyze this distance, we sequentially replace the Haar-random unitaries $V_1$ and $V_2$ by the corresponding path-recording oracle $P$ and restricted path-recording oracle $P^\#$, respectively. More precisely, we define $\rho_{1}^{(b)}(t)$ as the output state obtained by replacing $V_2$ with $P^\#$ (see Figure~\ref{fig:rho^{(1)}}), and $\rho_{2}^{(b)}(t)$ as the output state obtained by additionally replacing $V_1$ with $P$ (see Figure~\ref{fig:rho^{(2)}}):

\begin{align}
    \rho_{1}^{(b)}(t)&:=\mathbb{E}_{V_1\sim\mu_{\mathrm{Haar}}}\Tr_{R_1}\left(\proj{\mathcal{A}_t^{V_1,U_b,P^\#}}_{SS'ER_1}\right),
    \\
    \rho_{2}^{(b)}(t)&:=\Tr_{R_1R_2}\left(\proj{\mathcal{A}_t^{P,U_b,P^\#}}_{SS'ER_1R_2}\right).
\end{align}

The key observation in this hybrid argument is that the resulting state $\rho_{2}^{(b)}(t)$ is independent of $b$.
Therefore, the trace distance between the original output states can be controlled by bounding the error incurred at each step of replacing the Haar-random unitaries with the corresponding path-recording oracles.

\begin{Lem}
    For $0\le t<\frac{d}{4}$
    \begin{equation}
        \rho_{2}^{(Y)}(t)=\rho_{2}^{(N)}(t).
    \end{equation}
    \label{Lem:Hash equal}
\end{Lem}
\newpage
\begin{proof}[Proof of Lemma~\ref{Lem:Hash equal}]
\begin{figure}
    \centering
    \includegraphics[width=\linewidth]{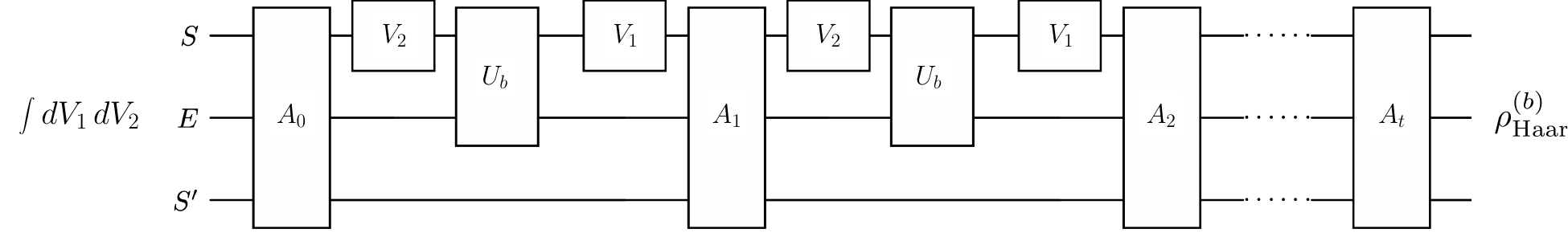}
    \caption{Original $t$-query algorithm with Haar-random unitaries.
The unitaries $V_2$ and $V_1$ are independently chosen from the Haar measure and fixed throughout all queries to the composite dilation unitary $(V_1\otimes I_E)U_b(V_2\otimes I_E)$.
The unitaries $A_0,\ldots,A_t$ are arbitrary and independent of the oracles $V_1, U_b$ and $V_2$, and represent the intermediate processing of the
algorithm on the target system $S$, the ancilla system $S'$, and the retained environment $E$.
Each invocation of $U_b$ acts on the corresponding environment register $E_i$.
Averaging the resulting output state over $V_1$ and $V_2$
gives $\rho_{\mathrm{Haar}}^{(b)}(t)$.}
    \label{fig:rho^Haar}
\end{figure}
\hspace{1cm}
\begin{figure}
    \centering
    \includegraphics[width=\linewidth]{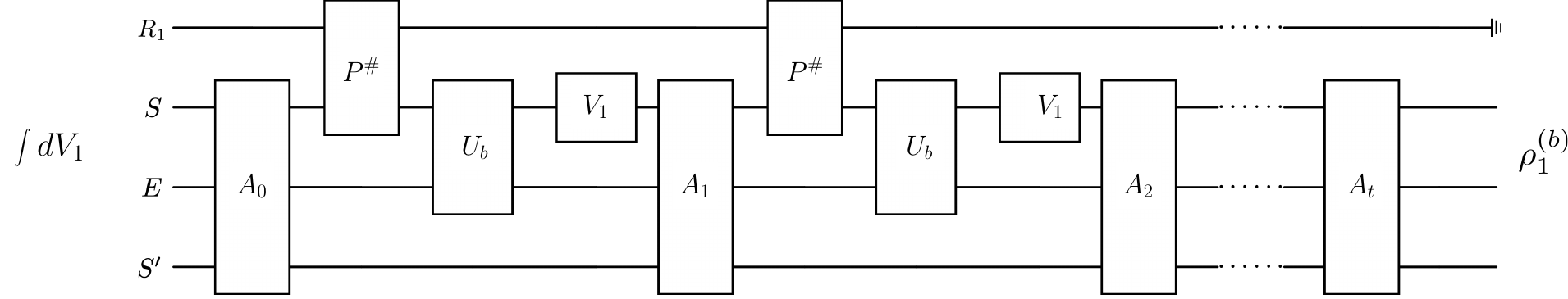}
    \caption{First hybrid in the path-recording argument. Every Haar-random unitary $V_2$ in Fig.~\ref{fig:rho^Haar} is replaced by the restricted path-recording oracle $P^\#$, while $V_1$ remains Haar-random. The oracle $P^\#$ acts jointly on the target system $S$ and the relation register $R_1$, recording the relation $x_i\mapsto(\alpha_i,u_i)$ generated at the $i$-th query. The operations $A_0,\ldots,A_t$, the dilation unitary $U_b$, and the environment registers $E_i$ are not changed from the original algorithm. After tracing out $R_1$ and averaging over the remaining Haar-random unitary $V_1$, the circuit yields the first hybrid state $\rho_{1}^{(b)}(t)$.}
    \label{fig:rho^{(1)}}
\end{figure}

\begin{figure}
    \centering
    \includegraphics[width=\linewidth]{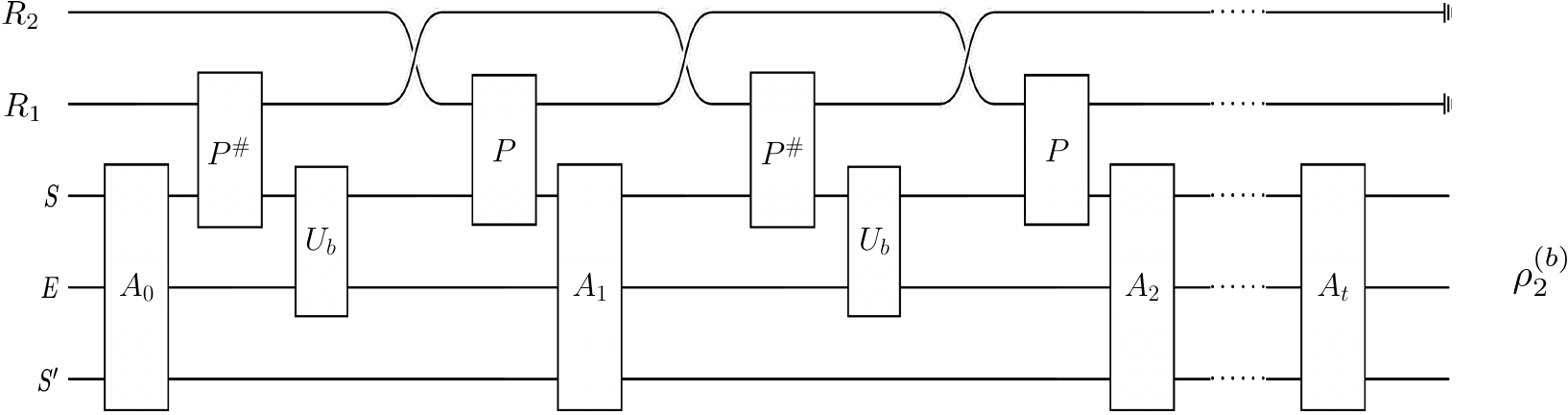}
    \caption{Second hybrid in the path-recording argument. In addition to replacing $V_2$ by the restricted path-recording oracle $P^\#$, every Haar-random unitary $V_1$ is replaced by the standard path-recording oracle $P$. At the $i$-th query, $P^\#$ records the transition $x_i\mapsto(\alpha_i,u_i)$ in the relation register $R_1$, $U_b$ produces the corresponding environment state $\ket{w_{b,\alpha_i}}_{E_i}$, and $P$ records the subsequent transition $(\alpha_i,u_i)\mapsto z_i$ in the relation register $R_2$. Thus, after $t$ queries, $R_1$ and $R_2$ store the relations $R_1(\vec{x},\vec{\alpha},\vec{u})$ and $R_2(\vec{\alpha},\vec{u},\vec{z})$, respectively. Tracing out the two relation registers yields the second hybrid state $\rho_{2}^{(b)}(t)$.}
    \label{fig:rho^{(2)}}
\end{figure}
    Consider the exact representation of $\ket{\mathcal{A}_t^{P,U_b, P^\#}}$(Figure~\ref{fig:rho^{(2)}}).
    For example, in the case of $t=1$, let $\ket{x_1}_S$ be a computational basis state of system $S$ for $x_1\in[d]$ and $\{\ket{\psi_i}\}$ be the basis in the environment system. The transition is then as follows:
    \begin{align}
        &\ket{x_1}_S\ket{\varnothing}_{R_1}\ket{\varnothing}_{R_2}\ket{\psi_i}_E\nonumber
        \\
        \xrightarrow{P^\#} &
        \frac{1}{\sqrt{\mathcal{Z}_0}}\sum_{\substack{\alpha_1\in[4],u_1\in[l]\\\varnothing\cup\{(x_1,(\alpha_1,u_1))\}\in\mathcal{S}_1^\#}}\ket{\alpha_1,u_1}\ket{\{(x_1,(\alpha_1,u_1))\}}_{R_1}\ket{\varnothing}_{R_2}\ket{\psi_i}_E\nonumber
        \\
        \xrightarrow{U_b} &
        \frac{1}{\sqrt{\mathcal{Z}_0}}\sum_{\substack{\alpha_1\in[4],u_1\in[l]\\\varnothing\cup\{(x_1,(\alpha_1,u_1))\}\in\mathcal{S}_1^\#}}\ket{\alpha_1,u_1}\ket{\{(x_1,(\alpha_1,u_1))\}}_{R_1}\ket{\varnothing}_{R_2}Q_{b,\alpha_1}\ket{\psi_i}_E\nonumber
        \\
        \xrightarrow{P} &
        \frac{1}{\sqrt{\mathcal{Z}_0(d-|\varnothing|)}}\sum_{\substack{\alpha_1\in[4],u_1\in[l],z_1\in[d]\\\{(x_1,(\alpha_1,u_1))\}\in\mathcal{S}_1^\#\\z_1\notin \mathrm{Im} \varnothing}}\ket{z_1}\ket{\{(x_1,(\alpha_1,u_1))\}}_{R_1}\ket{\{((\alpha_1,u_1),z_1)\}}_{R_2}Q_{b,\alpha_1}\ket{\psi_i}_E\nonumber
        \\ =&
        \frac{1}{\sqrt{4ld}}\sum_{\substack{\alpha_1\in[4],u_1\in[l],z_1\in[d]\\\{(x_1,(\alpha_1,u_1))\}\in\mathcal{S}_1^\#\\z_1\notin \mathrm{Im} \varnothing}}\ket{z_1}\ket{\{(x_1,(\alpha_1,u_1))\}}_{R_1}\ket{\{((\alpha_1,u_1),z_1)\}}_{R_2}Q_{b,\alpha_1}\ket{\psi_i}_E.\nonumber
    \end{align}
    We used Eq.~\eqref{eq:z} in the last equality.
    By the injectivity of $\mathcal{S}^\#$ and the definition of standard path recording oracle, for every $k\in[t]$, $y_k=(\alpha_k,u_k)\in[d]\setminus\bigcup_i^{k-1}\{y_i\}$ and $z_k\in[d]\setminus\bigcup_i^{k-1}\{z_i\}$.
    If we denote $[d]_\mathrm{dist}^k:=\{\vec{v}=(v_1,\cdots,v_k)\in[d]^k : v_i\neq v_j\,\text{ for }\,i\neq j\}$, $\vec{y}\in[d]^k$ and $\vec{z}\in[d]^k$ are both chosen from $[d]_\mathrm{dist}^k$. 
    Thus, by linearity, applying this transformation to the state $A_0\ket{0}_{SS'E}$, followed by $A_1$, gives
    \begin{equation}
        \ket{\mathcal{A}_1^{P,U_b,P^\#}}=\frac{1}{\sqrt{4ld}}\sum_{\substack{x_1\in[d]\\\alpha_1\in[4],u_1\in[l]\\z_1\in[d]\\\{(x_1,(\alpha_1,u_1))\}\in\mathcal{S}_1^\#\\z_1\notin \mathrm{Im} \varnothing}}\left\{A_1(\ket{z_1}\!\bra{x_1}_S\otimes I_{S'}\otimes Q_{b,\alpha_1})A_0\right\}\ket{0}_{SS'E}\ket{\{(x_1,(\alpha_1,u_1))\}}_{R_1}\ket{\{((\alpha_1,u_1),z_1)\}}_{R_2}.
    \end{equation}
    Repeating the same argument for each query, after $t$ queries, relations $R_1(\vec{x},\vec{\alpha},\vec{u}):=\{(x_i,(\alpha_i,u_i)):i\in[t]\}$ and $R_2(\vec{\alpha},\vec{u},\vec{z}):=\{((\alpha_i,u_i),z_i):i\in[t]\}$ are stored in the registers $R_1$ and $R_2$, respectively, and we obtain
    \begin{align}
        \ket{\mathcal{A}_t^{P,P^\#}}_{SS'R_1R_2E}
        &=\frac{1}{\sqrt{\prod_{i=0}^{t-1}4(l-i)(d-i)}}\sum_{\substack{\vec{x}\in[d]^t,\vec{u}\in[l]_\mathrm{dist}^t \\ \vec{\alpha}\in[4]^t,\vec{z}\in[d]_\mathrm{dist}^t}}\ket{\phi_{\vec{x},\vec{z},\vec{\alpha}}^{(b)}}_{SS'E}\otimes \ket{R_1(\vec{x},\vec{\alpha},\vec{u})}_{R_1}\otimes\ket{R_2(\vec{\alpha},\vec{u},\vec{z})}_{R_2}\\
        &=\frac{1}{\sqrt{4^t(l)_t(d)_t}}\sum_{\substack{\vec{x}\in[d]^t,\vec{u}\in[l]_\mathrm{dist}^t \\ \vec{\alpha}\in[4]^t,\vec{z}\in[d]_\mathrm{dist}^t}}\ket{\phi_{\vec{x},\vec{z},\vec{\alpha}}^{(b)}}_{SS'E}\otimes \ket{R_1(\vec{x},\vec{\alpha},\vec{u})}_{R_1}\otimes\ket{R_2(\vec{\alpha},\vec{u},\vec{z})}_{R_2}.
    \end{align}
    Here, we denote $(l)_t:=\prod_{i=0}^{t-1}(l-i)$, $(d)_t:=\prod_{i=0}^{t-1}(d-i)$ and 
    \begin{equation}
        \ket{\phi_{\vec{x},\vec{z},\vec{\alpha}}^{(b)}}_{SS'E}:=\left(\prod_{i=1}^tA_i\left(\ket{z_i}\!\bra{x_i}\otimes I\otimes Q_{b,\alpha_i}\right)\right)A_0\ket{0}_{SS'E}.
    \end{equation}
    To consider the output on $SS'E$ system, we first trace out the relation registers $R_1$ and
    $R_2$.  For convenience, let
    \begin{equation}
        y_i \coloneqq (\alpha_i,u_i),\qquad\vec{y}\coloneqq(y_1,\cdots,y_t).
    \end{equation}
    Then
    \begin{align}
        \rho_{2}^{(b)}(t)_{SS'E}
        &\coloneqq
        \operatorname{Tr}_{R_1R_2}
        \left[\ket{\mathcal{A}_t^{P,P^\#}}\!\bra{\mathcal{A}_t^{P,P^\#}}\right]\nonumber
        \\ &=\frac{1}{4^t(l)_t(d)_t}
        \sum_{\substack{\vec{x},\vec{x}'\in[d]^t,\,\vec{u},\vec{u}'\in[l]_\mathrm{dist}^t\\\vec{\alpha},\vec{\alpha}'\in[4]^t,\,\vec{z},\vec{z}'\in[d]_\mathrm{dist}^t
        }}\left\langle R_1(\vec{x}',\vec{\alpha}',\vec{u}')\,\middle|\,R_1(\vec{x},\vec{\alpha},\vec{u})\right\rangle
        \nonumber\\
        &\hspace{1.5cm}\times
        \left\langle R_2(\vec{\alpha}',\vec{u}',\vec{z}')\,\middle|\,R_2(\vec{\alpha},\vec{u},\vec{z})\right\rangle\ket{\phi_{\vec{x},\vec{z},\vec{\alpha}}^{(b)}}\!\bra{\phi_{\vec{x}',\vec{z}',\vec{\alpha}'}^{(b)}}.
        \label{eq:trace-out-R1-R2}
    \end{align}
    
    The relation registers are represented in an orthonormal computational basis.
    Therefore, the two inner products in~\eqref{eq:trace-out-R1-R2} are nonzero
    if and only if the corresponding relations are identical.  Since the relation
    $R_1$ is injective, the equality
    \begin{equation}
        R_1(\vec{x}',\vec{\alpha}',\vec{u}')=R_1(\vec{x},\vec{\alpha},\vec{u})  
    \end{equation}
    implies that there exists a unique permutation $\pi\in S_t$ such that
    \begin{equation}
        x_i'=x_{\pi(i)},\qquad(\alpha_i',u_i')=(\alpha_{\pi(i)},u_{\pi(i)})
    \end{equation}
    for every $i\in[t]$.  Moreover, because the same labels
    $(\alpha_i,u_i)$ occur in $R_2$, the condition
    \begin{equation}
        R_2(\vec{\alpha}',\vec{u}',\vec{z}')=R_2(\vec{\alpha},\vec{u},\vec{z})
    \end{equation}
    forces the same permutation $\pi$ and additionally gives
    \begin{equation}
        z_i'=z_{\pi(i)}
    \end{equation}
    for every $i\in[t]$.
    Consequently, only the terms with indices related by a permutation $\pi\in S_t$ remain, namely,
    \begin{align}
        \rho_{2}^{(b)}(t)_{SS'E}
        =
        \frac{1}{4^t(l)_t(d)_t}
        \sum_{\substack{\vec{x}\in[d]^t,\,\vec{u}\in[l]_\mathrm{dist}^t\\\vec{\alpha}\in[4]^t,\,\vec{z}\in[d]_\mathrm{dist}^t}}
        \sum_{\pi\in S_t}\ket{\phi_{\vec{x},\vec{z},\vec{\alpha}}^{(b)}}
        \!\bra{\phi_{\pi\vec{x},\pi\vec{z},\pi\vec{\alpha}}^{(b)}},
        \label{eq:state-after-tracing-relations}
    \end{align}
    where we use the convention
    \begin{equation}
        (\pi\vec{x})_i:=x_{\pi(i)},\qquad(\pi\vec{z})_i:=z_{\pi(i)},\qquad(\pi\vec{\alpha})_i:=\alpha_{\pi(i)}.
    \end{equation}
    Recall the relationship between $Q_{Y,\alpha}$ and $Q_{N,\alpha}$ in Eq.~\eqref{Q_Y and Q_N},
    \begin{align}
        \ket{\phi_{\vec{x},\vec{z},\vec{\alpha}}^{(N)}}
        &=\left(\prod_{i=1}^tA_i\left(\ket{z_i}\!\bra{x_i}\otimes I \otimes Q_{N,\alpha_i}\right)\right)A_0\ket{0}_{SS'E}
        \\ &=\left(\prod_{i=1}^tA_i\left(\ket{z_i}\!\bra{x_i}\otimes I \otimes \sum_{\beta_i=1}^{4}\Gamma_{\alpha_i,\beta_i}Q_{Y,\beta_i}\right)\right)A_0\ket{0}_{SS'E}
        \\ &=\sum_{\vec{\beta}\in[4]^t}\left(\prod_{j=1}^{t}\Gamma_{\alpha_j\beta_j}\right)\left(\prod_{i=1}^tA_i\left(\ket{z_i}\!\bra{x_i}\otimes I \otimes Q_{Y,\beta_i}\right)\right)A_0\ket{0}_{SS'E}
        \\ &=\sum_{\vec{\beta}\in[4]^t}\left(\prod_{j=1}^{t}\Gamma_{\alpha_j\beta_j}\right)\ket{\phi_{\vec{x},\vec{z},\vec{\beta}}^{(Y)}}.
    \end{align}
    It follows that 
    \begin{align}
        \rho_{2}^{(N)}(t)_{SS'E}&=\frac{1}{4^t(l)_t(d)_t}\sum_{\substack{\vec{x}\in[d]^t,\,\vec{u}\in[l]_\mathrm{dist}^t\\\vec{\alpha}\in[4]^t,\,\vec{z}\in[d]_\mathrm{dist}^t}}
        \sum_{\pi\in S_t}\sum_{\vec{\beta},\vec{\beta'}\in[4]^t}\left(\prod_{i=1}^{t}\Gamma_{\alpha_i\beta_i}\Gamma_{\alpha_{\pi(i)}\beta'_i}^*\right)\nonumber\\
        &\quad\times\ket{\phi_{\vec{x},\vec{z},\vec{\beta}}^{(Y)}}\!\bra{\phi_{\pi\vec{x},\pi\vec{z},\vec{\beta'}}^{(Y)}}
        \\ &=\frac{1}{4^t(l)_t(d)_t}\sum_{\substack{\vec{x}\in[d]^t,\,\vec{u}\in[l]_\mathrm{dist}^t\\\vec{z}\in[d]_\mathrm{dist}^t}}
        \sum_{\pi\in S_t}\sum_{\vec{\beta},\vec{\beta'}\in[4]^t}\left(\sum_{\vec{\alpha}\in[4]^t}\prod_{i=1}^{t}\Gamma_{\alpha_i\beta_i}\Gamma_{\alpha_i\beta'_{\pi^{-1}(i)}}^*\right)\nonumber\\
        &\quad\times\ket{\phi_{\vec{x},\vec{z},\vec{\beta}}^{(Y)}}\!\bra{\phi_{\pi\vec{x},\pi\vec{z},\vec{\beta'}}^{(Y)}}
        \\ &=\frac{1}{4^t(l)_t(d)_t}\sum_{\substack{\vec{x}\in[d]^t,\,\vec{u}\in[l]_\mathrm{dist}^t\\\vec{z}\in[d]_\mathrm{dist}^t}}
        \sum_{\pi\in S_t}\sum_{\vec{\beta},\vec{\beta'}\in[4]^t}\left(\prod_{i=1}^{t}\delta_{\beta_i, \beta'_{\pi^{-1}(i)}}\right)\nonumber\\
        &\quad\times\ket{\phi_{\vec{x},\vec{z},\vec{\beta}}^{(Y)}}\!\bra{\phi_{\pi\vec{x},\pi\vec{z},\vec{\beta'}}^{(Y)}}
        \\ &=\frac{1}{4^t(l)_t(d)_t}\sum_{\substack{\vec{x}\in[d]^t,\,\vec{u}\in[l]_\mathrm{dist}^t\\\vec{\beta}\in[4]^t,\,\vec{z}\in[d]_\mathrm{dist}^t}}
        \sum_{\pi\in S_t}\ket{\phi_{\vec{x},\vec{z},\vec{\beta}}^{(Y)}}\!\bra{\phi_{\pi\vec{x},\pi\vec{z},\pi\vec{\beta}}^{(Y)}}
        \\ &=\rho_{2}^{(Y)}(t)_{SS'E}.
    \end{align}
\end{proof}
\begin{Lem}
    For $0\le t<\frac{d}{4}$,
    \begin{equation}
    \norm{\rho_{\mathrm{Haar}}^{(Y)}(t)-\rho_{\mathrm{Haar}}^{(N)}(t)}_1\le\frac{16t(t-1)}{d}.
    \end{equation}
    \label{Lem:distance Haar}
\end{Lem}
\begin{proof}[Proof of Lemma~\ref{Lem:distance Haar}]
    \begin{align}
        &\norm{\rho_{\mathrm{Haar}}^{(Y)}(t)-\rho_{\mathrm{Haar}}^{(N)}(t)}_1
        \\ &\le\norm{\rho_{\mathrm{Haar}}^{(Y)}(t)-\rho_{1}^{(Y)}(t)}_1+\norm{\rho_{1}^{(Y)}(t)-\rho_{2}^{(Y)}(t)}_1
        +\norm{\rho_{2}^{(Y)}(t)-\rho_{\mathrm{Haar}}^{(N)}(t)}_1
        \\ &=\norm{\rho_{\mathrm{Haar}}^{(Y)}(t)-\rho_{1}^{(Y)}(t)}_1+\norm{\rho_{1}^{(Y)}(t)-\rho_{2}^{(Y)}(t)}_1
        +\norm{\rho_{2}^{(N)}(t)-\rho_{\mathrm{Haar}}^{(N)}(t)}_1
        \\ &\le\norm{\rho_{\mathrm{Haar}}^{(Y)}(t)-\rho_{1}^{(Y)}(t)}_1+\norm{\rho_{1}^{(Y)}(t)-\rho_{2}^{(Y)}(t)}_1
        +\norm{\rho_{2}^{(N)}(t)-\rho_{1}^{(N)}(t)}_1+\norm{\rho_{1}^{(N)}(t)-\rho_{\mathrm{Haar}}^{(N)}(t)}_1
        \\ &\le\frac{6t(t-1)}{d}+\frac{2t(t-1)}{d}+\frac{2t(t-1)}{d}+\frac{6t(t-1)}{d}
        \\ &=\frac{16t(t-1)}{d}.
    \end{align}
    We apply Lemma~\ref{Lem:Hash equal} in the first equality and Lemmas~\ref{Lem:Haar P} and~\ref{Lem:Haar P hash} in the third inequality.
\end{proof}
\begin{proof}[Proof of Theorem~\ref{Thm:lower bound}]

Consider the task of discriminating between $\mathcal{E}_Y^{V_1,V_2}$ and $\mathcal{E}_N^{V_1,V_2}$ using the quantum states $\rho^{(Y)}(t)$ and $\rho^{(N)}(t)$. If the algorithm succeeds in this task with probability at least $\frac{2}{3}$ using $t$ queries to $\mathcal{E}_b^{V_1,V_2}$, $b\in\{Y,N\}$, there exists a POVM $\{E_Y, E_N\}$ with $E_Y+E_N=I$ such that, for all $V_1,V_2\in\mathbb{U}(d)$,
\begin{align}
    p_{\mathrm{suc}}(V_1,V_2)
    &=\frac{1}{2}\Tr [E_Y\rho^{(Y)}(t)]+\frac{1}{2}\Tr [E_N\rho^{(N)}(t)]
    \\ &=\frac{1}{2}+\frac{1}{2}\Tr[E_Y(\rho^{(Y)}(t)-\rho^{(N)}(t))]
    \\ &\ge \frac{2}{3}.
\end{align}
Averaging both sides of the inequality, we obtain
\begin{align}
    \frac{2}{3}\le\mathbb{E}_{V_1,V_2}[p_{\mathrm{suc}}(V_1,V_2)]
    &=\frac{1}{2}+\frac{1}{2}\Tr[E_Y(\rho_{\mathrm{Haar}}^{(Y)}(t)-\rho_{\mathrm{Haar}}^{(N)}(t))]
    \\ &\le\frac{1}{2}+\frac{1}{2}\max_{0\le E_Y\le I}\Tr[E_Y(\rho_{\mathrm{Haar}}^{(Y)}(t)-\rho_{\mathrm{Haar}}^{(N)}(t))]
    \\ &=\frac{1}{2}+\frac{1}{4}\left\|\rho_{\mathrm{Haar}}^{(Y)}(t)-\rho_{\mathrm{Haar}}^{(N)}(t)\right\|_1,
\end{align}
where we use the Helstrom measurement~\cite{Helstrom1976,holevo2011} as an optimal measurement in the last equality.
If $t\ge \frac{d}{4}$, the desired $\Omega(\sqrt{d})$ lower bound already holds. For $t<\frac{d}{4}$, Lemma~\ref{Lem:distance Haar} implies
\begin{align}
    \frac{2}{3}\le\frac{1}{2}+\frac{1}{4}\left\|\rho_{\mathrm{Haar}}^{(Y)}(t)-\rho_{\mathrm{Haar}}^{(N)}(t)\right\|_1
    &\le\frac{1}{2}+\frac{4t^2}{d},
\end{align}
i.e., the number of queries to distinguish $\mathcal{E}_Y$ and $\mathcal{E}_N$ with success probability at least $\frac{2}{3}$ is bounded by
\begin{align}
    t&\ge\sqrt{\frac{d}{24}}.
\end{align}
Recall that the gap $\epsilon_{\mathrm{OEF}}$ between $F_{\mathrm{OEF}}(\mathcal{E}_Y)$ and $F_{\mathrm{OEF}}(\mathcal{E}_N)$ is constant according to Proposition \ref{prop:prop of channels}, which does not depend on the dimension of the system.
Since $F_{\mathrm{OEF}}(\mathcal{E})$ is invariant under unitary pre- and post-composition (Lemma~\ref{lem:unitary equivalent}), an algorithm that can estimate $F_{\mathrm{OEF}}(\mathcal{E})$ with additive error strictly smaller than $\frac{\epsilon_{\mathrm{OEF}}}{2}$ with success probability at least $\frac{2}{3}$ can also distinguish $\mathcal{E}_Y^{V_1,V_2}$ and $\mathcal{E}_N^{V_1,V_2}$ with probability at least $\frac{2}{3}$ for all $V_1,V_2\in\mathbb{U}(d)$. Therefore, we can establish the query lower bound $t=\bigOmega{\sqrt{d}}$.
By the OAGF--OEF relation in Sec.~\ref{sec:supplemental-notation}, the OAGF gap between the same hard instances is
\begin{equation}
      \epsilon_{\mathrm{OAGF}}=\frac{d}{d+1}\frac{\sqrt{2}-1}{4}\ge\frac{\sqrt{2}-1}{5},  
\end{equation}
where we used $d=4\ell\ge4$.
Hence, any OAGF estimator with additive error $\epsilon<\frac{\sqrt{2}-1}{10}$ distinguishes the two randomized channel ensembles by thresholding at the midpoint of their OAGF values.
This gives the same $\Omega(\sqrt{d})$ query lower bound.
Both hard instances have Kraus rank two, and composition
with unitary channels preserves this rank.
\end{proof}
\paragraph*{Inverse-only access.}
The same lower bound holds when the algorithm has access only to the inverse of the supplied dilation unitary.
Let $W=(V_1\otimes I_E)U_b(V_2\otimes I_E)$ be a forward hard instance, and define
\[
    h:=I_S\otimes H_E,\qquad V:=hW^\dagger h,
\]
where $H_E$ is the Hadamard gate on the environment qubit.
Writing $\mathcal{E}_X$ for the channel induced by a dilation unitary $X$ with the environment initialized in $\ket{0}_E$, Eq.~\eqref{eq:U_b} gives
\[
    \mathcal{E}_V(\rho)
    =V_2^\dagger K^{(b)\dagger}
    \mathcal{E}_b(V_1^\dagger\rho V_1)
    K^{(b)}V_2.
\]
Unitary invariance therefore implies
\[
    F_{\mathrm{OAGF}}(\mathcal{E}_V)
    =F_{\mathrm{OAGF}}(\mathcal{E}_b)
    =F_{\mathrm{OAGF}}(\mathcal{E}_W).
\]
An inverse-only estimator that succeeds for every supplied dilation unitary must also succeed for $V$.
Each query to $V^\dagger=hWh$ can be implemented using one forward query to $W$ and known Hadamard gates.
This gives a forward-only OAGF estimator for $W$ with the same accuracy, success probability, and query count.
The forward-query lower bound thus requires $\Omega(\sqrt{d})$ inverse queries, even for rank-two channels.
\par

\begin{Remark}
    [Comparison with unitarity]
    In contrast to OEF, the unitarity $\mathfrak{u}(\mathcal{E})$ can be estimated with query complexity independent of the system dimension under coherent black-box access to the channel~\cite{Chen2023}. Therefore, unlike OEF, unitarity estimation does not exhibit a dimension-dependent separation between black-box quantum channel access and access to a Stinespring unitary dilation.

    Moreover, the two hard instances used to establish our lower bound have different OEF values but identical unitarity. 
    Indeed, we have
    \begin{align}
        \mathsf{E}&=\frac{1}{d}\left(I\otimes \mathcal{E}\right)\left[\sum_{\substack{(\alpha,u)\\(\beta,v)}}\ket{\alpha,u}\!\bra{\beta,v}\otimes\ket{\alpha,u}\!\bra{\beta,v}\right]
        \\ &=\frac{1}{d}\sum_{\substack{(\alpha,u)\\(\beta,v)}}\ket{\alpha,u}\!\bra{\beta,v}\otimes\left(C_{\alpha\beta}\ket{\alpha,u}\!\bra{\beta,v}\right).
    \end{align}
    Thus, unitarity can be written as
    \begin{align}
        \mathfrak{u}(\mathcal{E})&=\Tr[\mathsf{E}^2]
        \\ &=\frac{1}{d^2}\Tr\sum_{\substack{(\alpha,u)\\(\beta,v)\\(\alpha',u')\\(\beta',v')}}\ket{\alpha,u}\!\braket{\beta,v|\alpha',u'}\!\bra{\beta',v'}\otimes\left(C_{\alpha\beta}C_{\alpha'\beta'}\ket{\alpha,u}\!\braket{\beta,v|\alpha',u'}\!\bra{\beta',v'}\right)
        \\ &=\frac{1}{d^2}\Tr\sum_{\substack{(\alpha,u)\\(\beta,v)\\(\beta',v')}}\ket{\alpha,u}\!\bra{\beta',v'}\otimes\left(C_{\alpha\beta}C_{\beta\beta'}\ket{\alpha,u}\!\bra{\beta',v'}\right)
        \\
        &=\frac{1}{d^2}\sum_{\substack{(\alpha,u)\\(\beta,v)}}C_{\alpha\beta}C_{\beta\alpha}
        \\ &=\frac{l^2}{d^2}\Tr C^2.
    \end{align}
    This shows that our lower-bound construction exploits information that is not captured by unitarity.
\end{Remark}
\section{Access model separation.}
\label{sec:separation}
We first make explicit the exponential query advantage of access to a Stinespring dilation unitary and its inverse over black-box channel access.
We then use this separation to prove the impossibility of universally simulating random dilation unitaries and their inverses.

\paragraph*{Query separation.}
Fix an additive accuracy $0<\epsilon<\frac{\sqrt{2}-1}{10}$ independent of $d$ and $m$, take Kraus rank $r=2$ and failure probability $\delta=\frac{1}{6}$, and allow any environment dimension $m\ge2$.
For every supplied Stinespring dilation unitary, Theorem~\ref{main-upper-bound} gives an algorithm that estimates OAGF to additive error $\epsilon$ with success probability at least $\frac{5}{6}$, using $O(1)$ queries to the dilation unitary and its inverse in total.
The constant may depend on $\epsilon$ but is independent of $d$ and $m$.
Under black-box access to the channel, estimating OAGF to the same accuracy with success probability at least $\frac{2}{3}$ requires $\Omega(\sqrt{d})$ queries in the worst case.
Since $d=2^n$ for an $n$-qubit system, we obtain the following corollary.
\begin{Cor}[Exponential query separation]
    \label{cor:separation}
    For $n$-qubit channels of Kraus rank $r=2$, with $n\ge2$ and the above parameter choices, OAGF can be estimated with success probability at least $\frac{5}{6}$ using $O(1)$ queries to a supplied Stinespring dilation unitary and its inverse.
    This bound holds for every supplied Stinespring dilation unitary with environment dimension $m\ge2$.
    The implied constant is independent of $n$ and $m$.
    Under black-box channel access, achieving the same accuracy with success probability at least $\frac{2}{3}$ requires $\Omega(2^{\frac{n}{2}})$ queries in the worst case.
\end{Cor}

\paragraph*{Simulation of fixed and random dilations.}
One channel query can be implemented with one forward query to a supplied dilation unitary by initializing the environment in $\ket{0}_E$ and tracing it out after the query.
Conversely, the separation rules out universal simulation of circuits using a dilation unitary and its inverse with only $\poly(\log d,q)$ channel queries, where $q$ is the number of dilation queries.
This remains true when the target is averaged over a randomly chosen compatible dilation.

To state the simulation problem precisely, fix an environment $\mathcal{H}_E$ of dimension $m\ge r$.
Define the set of compatible Stinespring dilation unitaries by
\begin{equation}
\label{eq:compatible-dilation-channels}
    \mathrm{Dilu}_m(\mathcal{E})
    :=\left\{U\in\mathbb{U}(dm):\begin{array}{l}
        \mathcal{E}(\rho)=\operatorname{Tr}_E\!\left[U\bigl(\rho\otimes\ket{0}\!\bra{0}_E\bigr)U^\dagger\right]
        \quad\forall\rho\in\mathcal{L}(\mathcal{H})
    \end{array}\right\}.
\end{equation}
Let $C$ be a quantum circuit making $q$ queries to $U$ and $U^\dagger$.
Its non-oracle operations may depend on known parameters, including $d,m,r,q$, accuracy parameters, and measurement outcomes, but are independent of the unknown channel and the particular supplied dilation.
Let $\rho_C(U,U^\dagger)$ be its final state for a fixed dilation, averaged over the circuit's internal randomness and measurement outcomes.
A simulator is given a classical description of $C$ and black-box access to $\mathcal{E}$, and outputs a state $\sigma_C(\mathcal{E})$ on the same output registers.

For random dilation access, let $\mu_m$ denote the uniform distribution on $\mathrm{Dilu}_m(\mathcal{E})$, with its dependence on $\mathcal{E}$ implicit.
A single draw $U\sim\mu_m$ supplies both the forward and inverse oracle throughout the entire circuit.
Thus, the target averaged output is
\begin{equation}
\label{eq:averaged-dilation-output}
    \overline{\rho}_C(\mathcal{E})
    :=\mathbb{E}_{U\sim\mu_m}
    \left[\rho_C(U,U^\dagger)\right].
\end{equation}

\begin{Cor}[No universal simulation of dilation unitaries and their inverses]
    \label{cor:no-go dilation}
    Fix a constant $0\le\kappa\le\frac{1}{6}$.
    No universal simulator using only $\poly(\log d,q)$ queries to $\mathcal{E}$ can satisfy either of the following guarantees for every admissible channel $\mathcal{E}$ and circuit $C$:
    \begin{enumerate}
        \item[(1)] Approximation up to the choice of compatible dilation:
        \begin{equation}
        \label{eq:fixed-dilation-simulation}
            \inf_{U\in\mathrm{Dilu}_m(\mathcal{E})}
            \frac{1}{2}\left\|\sigma_C(\mathcal{E})-\rho_C(U,U^\dagger)\right\|_1
            \le\kappa.
        \end{equation}
        \item[(2)] Approximation of the averaged output:
        \begin{equation}
        \label{eq:random dilation}
            \frac{1}{2}\left\|\sigma_C(\mathcal{E})-\overline{\rho}_C(\mathcal{E})\right\|_1
            \le\kappa.
        \end{equation}
    \end{enumerate}
    Both impossibilities already hold for rank-$2$ channels with $m=2$.
    The same impossibilities hold for simulators using only $\poly(\log d,q)$ forward queries to a supplied compatible dilation unitary $D$, with their output written as $\sigma_C(\mathcal{E};D)$ and either guarantee required for every such $D$.
\end{Cor}

\begin{proof}[Proof of Corollary~\ref{cor:no-go dilation}]
Choose the rank-$2$ channel family and fixed accuracy from Corollary~\ref{cor:separation}, and set $m=2$ and $\delta=\frac{1}{6}$.
Let $C$ implement Algorithm~\ref{alg:max-estimation} with these parameters.
It uses $q=O(1)$ queries and succeeds with probability at least $\frac{5}{6}$ for every fixed compatible dilation.
Let $M_{\mathcal{E}}$ be the POVM effect corresponding to an estimate within additive error $\epsilon$ of $F_{\mathrm{OAGF}}(\mathcal{E})$, and write $p_{\mathcal{E}}(\omega):=\Tr(M_{\mathcal{E}}\omega)$ for this success probability on an output state $\omega$.
The effect $M_{\mathcal{E}}$ is used only to analyze the success event; the circuit $C$ is independent of the unknown channel.
The pointwise guarantee is
\begin{equation}
    p_{\mathcal{E}}\!\left(\rho_C(U,U^\dagger)\right)
    \ge\frac{5}{6}
    \qquad\text{for every }U\in\mathrm{Dilu}_m(\mathcal{E}).
\end{equation}

For the first statement, trace distance bounds the change in the success probability under any fixed measurement, so for every compatible $U$,
\begin{equation}
    p_{\mathcal{E}}\!\left(\sigma_C(\mathcal{E})\right)
    \ge\frac{5}{6}
    -\frac{1}{2}\left\|\sigma_C(\mathcal{E})-\rho_C(U,U^\dagger)\right\|_1.
\end{equation}
Taking the supremum of the right-hand side over compatible dilations and using Eq.~\eqref{eq:fixed-dilation-simulation} gives
\begin{equation}
    p_{\mathcal{E}}\!\left(\sigma_C(\mathcal{E})\right)
    \ge\frac{5}{6}
    -\inf_{U\in\mathrm{Dilu}_m(\mathcal{E})}
    \frac{1}{2}\left\|\sigma_C(\mathcal{E})-\rho_C(U,U^\dagger)\right\|_1
    \ge\frac{5}{6}-\kappa\ge\frac{2}{3}.
\end{equation}

For the second statement, OAGF is independent of the dilation, so averaging the pointwise success guarantee preserves the same success event:
\begin{equation}
    p_{\mathcal{E}}\!\left(\overline{\rho}_C(\mathcal{E})\right)
    =\mathbb{E}_{U\sim\mu_m}
    \left[p_{\mathcal{E}}\!\left(\rho_C(U,U^\dagger)\right)\right]
    \ge\frac{5}{6}.
\end{equation}
Applying the trace-distance bound to Eq.~\eqref{eq:random dilation} therefore yields
\begin{equation}
\label{eq:supp-random-dilation-simulation-success}
    p_{\mathcal{E}}\!\left(\sigma_C(\mathcal{E})\right)
    \ge p_{\mathcal{E}}\!\left(\overline{\rho}_C(\mathcal{E})\right)
    -\frac{1}{2}\left\|\sigma_C(\mathcal{E})-\overline{\rho}_C(\mathcal{E})\right\|_1
    \ge\frac{5}{6}-\kappa\ge\frac{2}{3}.
\end{equation}
In either case, measuring the simulator's output would estimate OAGF with success probability at least $\frac{2}{3}$ using only $\poly(\log d,q)=\poly(\log d)$ channel queries.
This contradicts the $\Omega(\sqrt{d})$ lower bound in Corollary~\ref{cor:separation}.
Thus, a universal simulator satisfying either guarantee must use $\Omega(\sqrt{d})$ channel queries in the worst case, even for a constant-query circuit $C$ and a single-qubit environment.
If the simulator instead used only $\poly(\log d,q)$ forward queries to a supplied compatible dilation unitary $D$, the same measurement and either approximation guarantee would give an OAGF estimate with success probability at least $\frac{5}{6}-\kappa\ge\frac{2}{3}$ for every such $D$.
This contradicts the forward-query lower bound in Theorem~\ref{Thm:lower bound}.
\end{proof}
\section{Proof of the computational hardness}
\begin{proof}[Proof of Theorem~\ref{Thm:DQC1-hard}]
    Consider an $n$-qubit quantum channel
    \begin{equation}
        \mathcal{E}(\rho)=(1-p)\rho+pV\rho V^\dagger,
    \end{equation}
    where $p\in\left[0,\frac{1}{2}\right]$ and $V$ is a unitary operator that can be implemented by a $\poly(n)$-size quantum circuit.
    Define a unitary $U$ satisfying
    \begin{equation}
        U\ket{\psi}\ket{0}=\sqrt{1-p}\ket{\psi}\ket{0}+\sqrt{p}V\ket{\psi}\ket{1}.
    \end{equation}
    \begin{samepage}
    This is a Stinespring dilation unitary of $\mathcal{E}$ and can be implemented by the following $\poly(n)$-size circuit:
    \begin{center}
    \begin{quantikz}[row sep=0.5cm, column sep=0.7cm]
        \lstick{$\ket{0}_E$} & \gate{R_y(2\theta_p)} & \ctrl{1} & \qw \\
        \lstick{$\ket{\psi}_S$} & \qw & \gate{V} & \qw
    \end{quantikz}
    \end{center}
    where $\theta_p:=\arcsin\sqrt{p}$ and
    $R_y(2\theta_p)\ket{0}_E=\sqrt{1-p}\ket{0}_E+\sqrt{p}\ket{1}_E$.

    \end{samepage}

    Recall the definition of $F_{\mathrm{OEF}}$, 
    \begin{equation}
        \begin{aligned}
            F_{\mathrm{OEF}}(\mathcal{E})
            &=\max_{W\in\mathbb{U}}\bra{\Phi^+}(I\otimes W^\dagger)\mathsf{E}(I\otimes W)\ket{\Phi^+}\\
            &=\frac{1}{d^2}\max_{W\in\mathbb{U}}\left\{(1-p)\left|\Tr W^\dagger\right|^2+p\left|\Tr VW^\dagger\right|^2\right\}.
        \end{aligned}
    \end{equation}
    Substituting $W=I$ into the second representation gives
    \begin{equation}
        F_{\mathrm{OEF}}(\mathcal{E})\ge 1-p+p\left|\frac{\Tr V}{d}\right|^2.
    \end{equation}
    Moreover, by simple calculation, we obtain
    \begin{align}
       F_{\mathrm{OEF}}(\mathcal{E})
       &\le\max_{\ket{\psi}}\bra{\psi}\mathsf{E}\ket{\psi}
       \\ &=\lambda_{\mathrm{max}}(\mathsf{E})
       \\ &\le1-p+p\left|\frac{\Tr V}{d}\right|^2+\frac{p^2}{4(1-2p)}.
    \end{align}
    Therefore it follows
    \begin{equation}
        1-p+p\left|\frac{\Tr V}{d}\right|^2\le F_{\mathrm{OEF}}(\mathcal{E})\le1-p+p\left|\frac{\Tr V}{d}\right|^2+\frac{p^2}{4(1-2p)}.
    \end{equation}
    \begin{equation}
        \frac{F_{\mathrm{OEF}}(\mathcal{E})+p-1}{p}-\frac{p}{4(1-2p)}\le\left|\frac{\Tr V}{d}\right|^2\le \frac{F_{\mathrm{OEF}}(\mathcal{E})+p-1}{p}.
    \end{equation}
    Choosing $p$ and the estimation accuracy for $F_{\mathrm{OEF}}(\mathcal{E})$ appropriately as inverse polynomials in $n$ yields an inverse-polynomial additive approximation to $\left|\frac{\Tr V}{d}\right|^2$.
    On the other hand, approximating $\left|\frac{\Tr V}{d}\right|^2$ to additive error $\epsilon=\frac{1}{\poly(n)}$, given a classical description of a $\poly(n)$-size circuit for $V$, is DQC1-hard~\cite[Theorem~1 and Appendix~D]{Khatri_2019}.
    Therefore, estimating $F_{\mathrm{OEF}}(\mathcal{E})$ to additive error $\epsilon=\frac{1}{\poly(n)}$ given a classical description of $U$ is also DQC1-hard.
    Since $F_{\mathrm{OEF}}(\mathcal{E})=\frac{(d+1)F_{\mathrm{OAGF}}(\mathcal{E})-1}{d}$, an OAGF estimate with additive error $\epsilon$ yields an OEF estimate with additive error at most $2\epsilon$, establishing the claimed DQC1-hardness for OAGF estimation.
\end{proof}

\begin{proof}[Proof of Theorem~\ref{Thm:forward-hardness}]
    Fix a PRU family as in Ref.~\cite[Definition~7]{Ma2024}. Such a family can be constructed assuming the existence of quantum-secure one-way functions~\cite[Theorem~1]{Ma2024}.
    For $b\in\{Y,N\}$, define
    \begin{equation}
        U_b'=\left(V_1\otimes I\right)U_b(V_2\otimes I),
    \end{equation}
    where $U_b$ is defined in Eq.~\eqref{eq:U_b}.
    For both labels $b$, taking $V_1$ and $V_2$ from the fixed PRU family gives a family of polynomial-size dilation circuits.
    Suppose a polynomial-time quantum algorithm $\mathcal{A}$, given only forward oracle access to the supplied dilation unitary, estimates the OAGF to inverse-polynomial additive accuracy with probability at least $\frac{2}{3}$ on every member of this family.
    In each experiment we run $\mathcal{A}$ for at most its polynomial time bound, counting both quantum and classical computation.
    (i) $V_1,V_2$ are sampled from this PRU family.
    (ii) $V_1$ and $V_2$ are sampled from the Haar measure.
    In both cases, $V_1$ and $V_2$ are sampled independently once and held fixed throughout all oracle queries.
    Define a distinguisher $\mathcal{D}$ that samples $b\in\{Y,N\}$ uniformly, keeps $b$ private, and runs $\mathcal{A}$ with forward oracle access to $U_b'$ and accuracy $\epsilon=\frac{1}{\poly(n)}$.
    Each query to $U_b'$ is simulated by applying $V_2$, implementing $U_b$ using its explicit circuit, and then applying $V_1$, using only forward queries.
    The induced channel is $\mathcal{E}_b'=\mathcal{V}_1\circ\mathcal{E}_b\circ\mathcal{V}_2$, where $\mathcal{V}_j(\rho)=V_j\rho V_j^\dagger$.
    By Lemma~\ref{lem:unitary equivalent} and the affine relation between OEF and OAGF,
    \begin{equation}
        F_{\mathrm{OAGF}}(\mathcal{E}_b')=F_{\mathrm{OAGF}}(\mathcal{E}_b).
    \end{equation}
    Proposition~\ref{prop:prop of channels} shows that the OAGF values of $\mathcal{E}_Y$ and $\mathcal{E}_N$ differ by at least $\frac{\sqrt{2}-1}{5}$.
    We compare the estimate $F$ returned by $\mathcal{A}$ with the two known values $F_{\mathrm{OAGF}}(\mathcal{E}_Y)$ and $F_{\mathrm{OAGF}}(\mathcal{E}_N)$, and set $b'$ to the label corresponding to the closer value.
    The distinguisher $\mathcal{D}$ accepts if and only if $b'=b$.
    In case (i), the supplied dilation circuit belongs to this family of polynomial-size circuits. For sufficiently large $n$, $\epsilon$ is smaller than half this gap, so a successful OAGF estimate yields the correct label, and hence
    \begin{equation}
        \Pr_{\mathrm{PRU}}\left[b'=b\right]\ge\frac{2}{3}.
    \end{equation}
    In case (ii), let $t$ denote the number of forward oracle queries made by $\mathcal{A}$. Its polynomial running time implies $t=\poly(n)$. The averaged output states $\rho_{\mathrm{Haar}}^{(Y)}(t)$ and $\rho_{\mathrm{Haar}}^{(N)}(t)$ include the full accessible system $SS'E$. The optimal binary discrimination bound and Lemma~\ref{Lem:distance Haar}, with $d=2^n$, give, for sufficiently large $n$,
    \begin{align}
        \Pr_{\mathrm{Haar}}\left[b'=b\right]
        &\le\frac{1}{2}+\frac{1}{4}\left\|\rho_{\mathrm{Haar}}^{(Y)}(t)-\rho_{\mathrm{Haar}}^{(N)}(t)\right\|_1\nonumber\\
        &\le\frac{1}{2}+\frac{4t(t-1)}{2^n}
        =\frac{1}{2}+\operatorname{negl}(n).
    \end{align}
    Here, $\operatorname{negl}(n)$ denotes a function that decays faster than any inverse polynomial in $n$.
    Thus the success probabilities in cases (i) and (ii) differ by at least $\frac{1}{6}-\operatorname{negl}(n)$.

    By the definition of a PRU family, a polynomial-time quantum algorithm with forward access to one unitary has nearly the same acceptance probability whether that unitary is drawn from the PRU family or from the Haar measure. The difference is negligible in $n$.
    We therefore compare the distributions of $(V_1,V_2)$ in the order
    \[
        (\mathrm{PRU},\mathrm{PRU})
        \longrightarrow(\mathrm{Haar},\mathrm{PRU})
        \longrightarrow(\mathrm{Haar},\mathrm{Haar}).
    \]
    In the first step, we draw $V_2$ independently from the PRU family and run its circuit. The comparison then needs forward access only to $V_1$.
    In the second step, we simulate forward queries to the independent Haar-random $V_1$ in both experiments using the same standard path-recording oracle $P$ as in the proof of the forward-query lower bound (Sec.~\ref{subsec:path-recording-methods}).
    By Lemma~\ref{Lem:Haar P}, this changes each acceptance probability by at most $O\!\left(\frac{t^2}{d}\right)$, which is negligible for $t=\poly(n)$ and $d=2^n$. With $V_1$ simulated, the comparison needs forward access only to $V_2$.
    Thus each step changes the acceptance probability by a negligible amount. This contradicts the gap of at least $\frac{1}{6}-\operatorname{negl}(n)$ between cases (i) and (ii).
\end{proof}
\end{document}